\documentclass[a4paper,11pt]{article}
\usepackage[a4paper,margin=1in]{geometry}
\usepackage[utf8]{inputenc}
\usepackage[T1]{fontenc}
\usepackage[dvipsnames]{xcolor}
\usepackage{graphicx}
\usepackage{amsmath, amssymb, amstext, mathtools}
\usepackage{amsthm}
\usepackage{thm-restate}
\usepackage{enumitem}
\usepackage{subcaption}
\usepackage{tikz}
\usepackage{pgfplots}
\usepackage{float}
\usepackage{hyperref}
\usepackage[numbers,sort&compress]{natbib}
\usepackage{svg}

\allowdisplaybreaks

\definecolor[named]{urlblue}{cmyk}{1,0.58,0,0.21}

\hypersetup{
 breaklinks=true,
 colorlinks=true,
 citecolor=purple!70!blue!60!black,
 linkcolor=purple!70!blue!90!black,
 urlcolor=urlblue,
 pdflang={en},
 pdftitle={Computing Square Colorings on Bounded-Treewidth and Planar Graphs},
 pdfauthor={Akanksha Agrawal, Daniel Marx, Daniel Neuen and Jasper Slusallek}
}

\usetikzlibrary{backgrounds,patterns,arrows,automata,positioning,decorations,shapes,calc}

\tikzstyle{normalvertex}=[circle,fill=white,draw=black,minimum size=10pt,inner sep=0pt]
\tikzstyle{smallvertex}=[circle,fill=white,draw=black,minimum size=6pt,inner sep=0pt]

\newtheorem{theorem}{Theorem}[section]
\newtheorem{lemma}[theorem]{Lemma}
\newtheorem{corollary}[theorem]{Corollary}

\newtheorem{observation}[theorem]{Observation}

\theoremstyle{definition}
\newtheorem{definition}[theorem]{Definition}
\newtheorem{property}[theorem]{Property}

\theoremstyle{remark}

\newtheorem{claim}{Claim}[section]

\newenvironment{claimproof}{\begin{proof}}{\end{proof}}

\newenvironment{correctness}{\vspace{0.1cm}\noindent(\textbf{Correctness.})}{}
\newenvironment{runningtime}{\vspace{0.1cm}\noindent(\textbf{Running Time.})}{}

\newcommand{\defproblem}[3]{
  \vspace{5mm}
\noindent\fbox{
  \begin{minipage}{0.96\textwidth}
  \begin{tabular*}{\textwidth}{@{\extracolsep{\fill}}lr} \textsc{#1} \\ \end{tabular*}
  {\bf{Input:}} #2  \\
  {\bf{Question:}} #3
  \end{minipage}
  }
  \vspace{5mm}
}

\newcommand{\defproblemrestriction}[2]{
  \vspace{5mm}
\noindent\fbox{
  \begin{minipage}{0.96\textwidth}
  \begin{tabular*}{\textwidth}{@{\extracolsep{\fill}}lr} \textsc{#1} \\ \end{tabular*}
  {\bf{Definition:}} #2
  \end{minipage}
  }
  \vspace{5mm}
}

\newcommand{\CC}{\mathcal{C}}

\newcommand{\CY}{\mathcal{Y}}
\newcommand{\CZ}{\mathcal{Z}}

\newcommand{\ZZ}{\mathbb{Z}}
\newcommand{\ZZp}{\mathbb{Z}_{>0}}
\newcommand{\ZZnn}{\mathbb{Z}_{\geq 0}}

\renewcommand{\epsilon}{\varepsilon}

\DeclareMathOperator{\dist}{dist}

\DeclareMathOperator{\tw}{tw}

\DeclareMathOperator{\NZ}{NZ}

\newcommand{\SqCol}{\textnormal{\textsc{Square Coloring}}}
\newcommand{\qSqCol}{\textnormal{\textsc{Square-$q$-Coloring}}}

\newcommand{\red}{{\sf red}}
\newcommand{\blue}{{\sf blue}}
\newcommand{\green}{{\sf green}}
\newcommand{\yellow}{{\sf yellow}}
\newcommand{\orange}{{\sf orange}}
\newcommand{\violet}{{\sf violet}}
\newcommand{\cyan}{{\sf cyan}}

\newcommand{\ETH}{\textnormal{\textsf{ETH}}}
\newcommand{\NP}{\textnormal{\textsf{NP}}}

\newcommand{\eqRel}{\textsf{EQ}}

\newcommand{\ttw}{\textnormal{\textsf{tw}}}
\newcommand{\ccw}{\textnormal{\textsf{cw}}}
\newcommand{\echi}{\bar\chi}

\newcommand{\Vin}{V_{\textnormal{in}}}
\newcommand{\Vout}{V_{\textnormal{out}}}
\newcommand{\Vx}{V_x}

\newcommand{\Cin}{C_{\textnormal{in}}}
\newcommand{\CRel}{C_R}

\newcommand{\chiin}{\chi_{\textnormal{in}}}
\newcommand{\chiout}{\chi_{\textnormal{out}}}
\newcommand{\chix}{\chi_x}

\def\halftargetval{n^6}
\def\targetval{2n^6}

\newcommand{\SubsetGadget}{{\sf Subset}}
\newcommand{\CopyGadget}{{\sf ColorClassCopy}}
\newcommand{\ConstSocket}{{\sf ConstSocket}}
\newcommand{\SwitchSocket}{{\sf SwitchSocket}}
\newcommand{\EdgeSelection}{{\sf EdgeSelection}}
\newcommand{\VectorState}{{\sf VectorState}}
\newcommand{\OWSGadget}{{\sf OWS}}
\newcommand{\VectorSelection}{{\sf VectorSelection}}

\newcommand{\orcid}[1]{\href{https://orcid.org/#1}{\includegraphics[height=1.8ex]{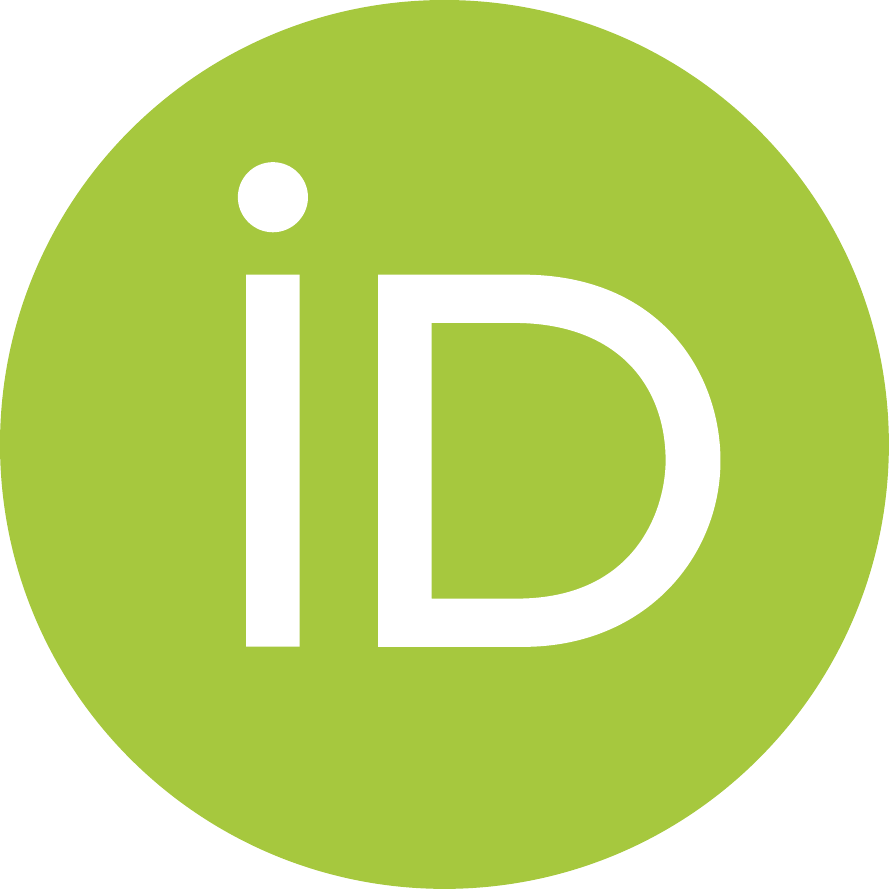}}}
\newcommand{\email}[1]{\href{mailto:#1}{\texttt{#1}}}

\title{Computing Square Colorings on\\Bounded-Treewidth and Planar Graphs}
\author{
Akanksha Agrawal \orcid{0000-0002-0656-7572}\\
Indian Institute of Technology Madras\\
\email{akanksha@cse.iitm.ac.in}
\and
D{\'{a}}niel Marx \orcid{0000-0002-5686-8314}\\
CISPA Helmholtz Center for Information Security\\
\email{marx@cispa.de}
\and
Daniel Neuen \orcid{0000-0002-4940-0318}\\
Simon Fraser University\\
\email{dneuen@sfu.ca}
\and
Jasper Slusallek\\
Saarland University\\
\email{s8jaslus@stud.uni-saarland.de}
}
\date{}

\begin{document}

\maketitle

\begin{abstract}
 A \emph{ square coloring} of a graph $G$ is a coloring of the square $G^2$ of $G$, that is, a coloring of the vertices of $G$ such that any two vertices that are at distance at most $2$ in $G$ receive different colors.
 We investigate the complexity of finding a square coloring with a given number of $q$ colors.
 We show that the problem is polynomial-time solvable on graphs of bounded treewidth by presenting an algorithm with running time $n^{2^{\ttw + 4}+O(1)}$ for graphs of treewidth at most $\ttw$.
 The somewhat unusual exponent $2^\ttw$ in the running time is essentially optimal: we show that for any $\epsilon>0$, there is no algorithm with running time $f(\ttw)n^{(2-\epsilon)^\ttw}$ unless the Exponential-Time Hypothesis (\ETH) fails.

 We also show that the square coloring problem is NP-hard on planar graphs for any fixed number $q \ge 4$ of colors.
 Our main algorithmic result is showing that the problem (when the number of colors $q$ is part of the input) can be solved in subexponential time $2^{O(n^{2/3}\log n)}$ on planar graphs.
 The result follows from the combination of two algorithms.
 If the number $q$ of colors is small ($\le n^{1/3}$), then we can exploit a treewidth bound on the square of the graph to solve the problem in time $2^{O(\sqrt{qn}\log n)}$.
 If the number of colors is large ($\ge n^{1/3}$), then an algorithm based on protrusion decompositions and building on our result for the bounded-treewidth case solves the problem in time $2^{O(n\log n/q)}$. 
\end{abstract}

\makeatletter{\renewcommand*{\@makefnmark}{}
 \footnotetext{Research supported by the European Research Council (ERC) consolidator grant No.~725978 SYSTEMATICGRAPH.}
\makeatother}

\section{Introduction}
\label{sec:introduction}
The \emph{square $G^2$} of a graph $G$ has the same vertex set as $G$ and two vertices in $G^2$ are adjacent if and only if they are at distance at most $2$ in $G$.
A \emph{square coloring} of $G$ is a proper vertex coloring of $G^2$, or equivalently, an assignment of colors to the vertices of $G$ such that not only adjacent vertices receive different colors, but this is also true for vertices at distance $2$.
Observe that if $G$ has maximum degree $\Delta$, then a square coloring of $G$ certainly needs at least $\Delta+1$ colors.

The notion of square coloring appeared in many different forms in the combinatorics and computer science literature.
Wegener \cite{Wegener77} conjectured that a planar graph can be square-colored with $\frac{3}{2}\Delta+1$ colors (for $\Delta\ge 7$).
This conjecture started a long and still active line of research, with the goal of obtaining upper bounds on the square-chromatic number of various graph classes \cite{BousquetMD22,DBLP:journals/jct/Thomassen18,MR2974301,DBLP:journals/ejc/DvorakKNS08,DBLP:journals/jct/MolloyS05,MR1953224,DBLP:journals/siamdm/AgnarssonH03,MR1888178}.
A {\em strong edge coloring} is the distance-2 version of an edge coloring, i.e, a strong edge coloring of $G$ is a square coloring of the line graph of $G$. Erd\H os and Ne{\v{s}}et{\v{r}}il \cite{MR999926} conjectured that $\frac{5}{4}\Delta^2$ colors are always sufficient for a strong edge coloring. While the conjecture is still open, there are many partial results towards this bound and on similar bounds for, e.g., planar graphs \cite{MR3873551,DBLP:journals/endm/BensmailBH15,DBLP:journals/dm/HudakLSS14,MR1438613,MR1412876}.
Some of these combinatorial upper bounds are algorithmic and give polynomial-time approximation algorithms for computing the square-chromatic number of, e.g., planar graphs (see also \cite{DBLP:conf/icci/LloydR92,DBLP:journals/siamdm/AgnarssonH03}).

In computer science, graph coloring and its variants are often used as a model for assignment problems where adjacent vertices are in conflict and hence cannot receive the same resource (time slot, frequency, processor, etc.) \cite{DBLP:journals/dam/FialaGKKK18,MR2002171,MR1653829,MR4078897,DBLP:conf/icalp/FialaGK05,MR2032294,DBLP:journals/siamdm/ChangK96,DBLP:journals/tcs/FialaGK11}.
Then the \SqCol\ problem (given a graph $G$ and a number $q$, determine whether there is a square coloring of $G$ with $q$ colors) is the natural extension where not only adjacent vertices, but even vertices at distance $2$ are in conflict. 
In an even more general setting, an $L(p,q)$-labeling models a frequency assignment problem: assuming that the colors are the integers, adjacent vertices have to receive colors that have difference at least $q$, while vertices at distance $2$ have to receive colors that have difference at least $p$. Now a square coloring is precisely a $L(1,1)$-labeling.
The \SqCol\ problem also received significant attention in the field of distributed computing (under the name distance-2 coloring or d2-coloring), as it appears naturally for example in unique naming and derandomization \cite{DBLP:conf/podc/HalldorssonKM20,DBLP:conf/sirocco/FraigniaudHN20,DBLP:conf/focs/GhaffariHK18,DBLP:journals/cj/LeeL14,DBLP:journals/tcs/BlairM12}.

How does the complexity of \SqCol\ differ from the usual \textsc{Vertex Coloring} problem?
The problem is known to be \NP-hard for every fixed $q \ge 4$ \cite{DBLP:journals/mp/McCormick83,DBLP:conf/icci/LloydR92,EricksonTB05}.
Note that if $q=3$, then every YES-instance has maximum degree $2$, hence \SqCol\ is polynomial-time solvable.
The goal of this paper is to look at the exact complexity of solving the \SqCol\ problem on certain important classes of graphs.
One can observe that the square of an interval graph is always an interval graph \cite{DBLP:journals/akcej/PaulPP16}, hence \SqCol\ is also polynomial-time solvable on interval graphs.
The square of a tree is not necessarily a tree, but it is chordal \cite{MR115937}, hence \SqCol\ on trees is also polynomial-time solvable.

Treewidth is a graph invariant that, roughly speaking, measures how close the graph is to having a treelike structure.
The combinatorial and algorithmic aspects of treewidth were extensively studied in the literature, both as a standalone goal and in applications to other graph classes (see, e.g., \cite{DBLP:conf/sirocco/Bodlaender07}).
The algorithmic importance of treewidth comes from the fact that many hard algorithmic problems are polynomial-time solvable on graphs of bounded treewidth. For example, for every fixed $\ttw$, the chromatic number problem is linear-time solvable for graphs of treewidth at most $\ttw$.
More precisely, \textsc{$q$-Coloring} can be solved in time $q^{O(\ttw)}\cdot n$, which, together with the fact that a graph of treewidth $\ttw$ has chromatic number at most $\ttw+1$, implies that a $2^{O(\ttw\log\ttw)}\cdot n$ algorithm for computing the chromatic number.

It was shown by Zhou, Kanari, and Nishizeki~\cite{ZhouKN00} that \SqCol\ is also polynomial-time solvable on bounded-treewidth graphs.
More precisely, given an $n$-vertex graph $G$ with treewidth $\ttw$, their algorithm decides in time $(q+1)^{2^{8\ttw+1}}\cdot n^{O(1)}$ if $G$ admits a square coloring with $q$ colors.
Our first result is improving the exponent of the running time from $256^{\ttw}$ to roughly $2^{\ttw}$.

\begin{theorem}
 \label{thm:tw-alg-intro}
 \SqCol\ can be solved in time $(q+1)^{2^{\ttw + 4}} \cdot n^{O(1)}$ on graphs of treewidth $\ttw$.
\end{theorem}

The algorithm follows standard dynamic programming techniques on tree decompositions; however, we make some extra effort to ensure that $2^\ttw$ appears in the exponent and not $c^\ttw$ for some $c>2$. Note that this form of running time is somewhat unusual.
Typically, when considering a problem on graphs of treewidth $\ttw$, then either the problem is fixed-parameter tractable (FPT) parameterized by treewidth (that is, it can be solved in time $f(\ttw) \cdot n^{O(1)}$ for some function $f$) or W[1]-hard, but $n^{O(\ttw)}$ time is sufficient to solve it.
Fiala, Golovach, and Kratochv\'\i l~\cite{DBLP:journals/tcs/FialaGK11} showed that \SqCol\ is W[1]-hard parameterized by treewidth, hence $\ttw$ has to appear in the exponent of $n$ in the running time (observe that $n$ colors are always sufficient to obtain a square coloring of $G$, i.e., we can always assume $q \leq n$).
It may seem inefficient that the exponent of the running time depends on the treewidth in such a drastic way in our algorithm, but we show that the exponential dependence seems to be unavoidable. The lower bound assumes the Exponential-Time Hypothesis (\ETH) of Impagliazzo, Paturi, and Zane \cite{ImpagliazzoPZ01}.

\begin{theorem}
 \label{thm:tw-lb-intro}
 Assuming \ETH, for any $\epsilon>0$ and any function $f$, there is no $f(\ttw)n^{(2-\epsilon)^\ttw}$ time algorithm solving \SqCol\ on graphs of treewidth $\ttw$.
\end{theorem}

That is, our algorithm with $2^\ttw$ in the exponent is essentially optimal.
We are not aware of any other natural problem with a similar dependence on treewidth.
Interestingly, for the measure cliquewidth $\ccw$, it is known that the best possible running time for \textsc{Vertex Coloring} has $2^\ccw$ in the exponent (assuming \ETH) \cite{DBLP:conf/soda/GolovachL0Z18}.

Next, we turn our attention to planar graphs.
The \textsc{$q$-Coloring} problem is \NP-hard on planar graphs for $q=3$, but becomes polynomial-time solvable for every $q \ge 4$ (because of the Four Color Theorem). By contrast, we show that \SqCol\ is \NP-hard on planar graphs for every fixed $q \ge 4$.

\begin{theorem}
 \label{thm:planar-lb-intro}
 \qSqCol\ is \NP-hard on planar graphs for every fixed $q \ge 4$.
\end{theorem}  

Even though \textsc{$3$-Coloring} is \NP-hard on planar graphs, it can be still solved more efficiently than on general graphs.
It is known that an $n$-vertex planar graph has treewidth $O(\sqrt{n})$.
This combinatorial bound and the $2^{O(\ttw)}\cdot  n$ algorithm for \textsc{$3$-Coloring} immediately imply a subexponential $2^{O(\sqrt{n})}$ time algorithm.
On the other hand, for general graphs, a $2^{o(n)}$ time algorithm would violate \ETH.

Does \SqCol\ also admit a subexponential algorithm on planar graphs?
As there is no constant bound on the number of colors needed for square coloring planar graphs, we consider the version where the number $q$ of colors is part of the input. The treewidth-based approach does not seem to work.
First, even though treewidth is $O(\sqrt{n})$, Theorem~\ref{thm:tw-alg-intro} would give only a double exponential $n^{2^{O(\sqrt{n})}}$ time algorithm.
We can try to use the $\ttw^{O(\ttw)}\cdot n$ algorithm on the square of the planar graph.
However, the square of an $n$-vertex planar graph $G$ can have treewidth up to $n-1$ (for example, when $G$ is a star with $n-1$ leaves and hence $G^2$ is an $n$-clique).
Therefore, this approach would give only a $2^{O(n\log n)}$ time algorithm.
Nevertheless, our main algorithmic result is a positive answer to this question:

\begin{theorem}
 \label{thm:planar-alg-intro}
 \SqCol\ can be solved in time $2^{O(n^{2/3}\log n)}$ on planar graphs.
\end{theorem}

The slightly unusual exponent $2/3$ comes from a trade off between two algorithms with running time $2^{O(\sqrt{qn}\log n)}$ and $2^{O(n\log n/q)}$, respectively.
By using the former for $q\le n^{1/3}$ and the latter for $q\ge n^{1/3}$, the bound $O(n^{2/3}\log n)$ on the exponent follows.

\subsection{Our Techniques}

In this section, we briefly overview the techniques and main ideas used in the results of the paper.

\paragraph{Algorithm for bounded-treewidth graphs.}
Let us recall first the definition of tree decompositions (see Section~\ref{sec:preliminaries} for more details).
A tree decomposition of a graph $G$ consists of a rooted tree $T$ and a bag $\beta(t)\subseteq V(G)$ for every node $t$ of $T$ with the following properties: (i) every vertex $v$ of $G$ appears in at least one bag, (ii) for every vertex $v$ of $G$, the bags containing $v$ correspond to a connected subtree of $T$, (iii) if two vertices of $G$ are adjacent, then there is at least one bag containing both of them.
The width of a tree decomposition is the size of the largest bag minus one, and the treewidth of a graph is the smallest possible width of a decomposition.
For a node $t$ of $T$, let us denote by $V_t$ the union over all bags $\beta(t')$ where $t'$ is a descendant of $t$ (including $t$ itself). 

A standard way of designing algorithms on tree decompositions is to define some number of subproblems for each node $t$ of the decomposition, and then solve them in a bottom-up way.
Typically, these subproblems ask about the existence of partial solutions having a certain type. We classify the partial solutions into some number of types in such a way that if two partial solutions have the same type and one has an extension into a full solution, then the same extension would work for the other solution as well. The subproblems at node $t$ would correspond to finding which types of partial solutions are possible. Finally, we argue that if we have solved every subproblem for every child $t'$ of $t$, then the subproblems at $t$ can be solved efficiently. For example, for the \textsc{$q$-coloring} problem, the partial solutions at node $t$ are proper colorings of the graph $G[V_t]$. We define types by classifying the partial solutions according to how they color $\beta(t)$. As $|\beta(t)| \le \ttw+1$, this gives at most $q^{\ttw+1}$ types. Then easy recurrence relations show how to solve these subproblems if all the subproblems are already solved for every child of $t$.

Let us observe that we cannot define the types of partial solutions the same way in the case of the \SqCol\ problem.
It very well may be that two colorings of $G[V_t]$ agree on $\beta(t)$, but one has an extension to $G$, whereas the other one does not.
For example, let $\chi_t$ be a square coloring of $V_t$ and let $u$ be a vertex not in $V_t$.
Then, whether $\chi_t$ can be extended to a square coloring $\chi$ of $G$ where $\chi(u)$ is \red, depends not only on whether $u$ has a neighbor in $\beta(t)$ that is colored \red\ by $\chi_t$, but also on whether those vertices have a neighbor in $V_t\setminus \beta(t)$ that is colored \red\ by $\chi_t$ (see Figure~\ref{fig:twdp}).
Thus, we cannot define the types of partial solutions at node $t$ based only on how they color $\beta(t)$, but we need to take into account the colors on the neighbors of these vertices.
This suggests a more refined way of defining the type of a partial solution based on how each vertex $v\in \beta(t)$ is colored and also which subset of the $q$ colors appears already in the neighborhood of each $v\in\beta(t)$.
However, as $\beta(t)$ can be up to $\ttw+1$, this definition would result in up to $(q\cdot 2^q)^{\ttw+1}$ types. As $q$ can be of order $n$, this is clearly too many.

\begin{figure}
 \centering
 \begin{tikzpicture}
  
  \draw[fill=gray!30] (-0.5,0.02) arc (150:390:2.306cm and 3cm);
  \draw[fill=white] (1.5,0) ellipse (2cm and 0.4cm);
  
  \draw[fill=gray!30] (6.5,0.02) arc (150:390:2.306cm and 3cm);
  \draw[fill=white] (8.5,0) ellipse (2cm and 0.4cm);
  
  \node at (1.5,-5) {$\chi_t$};
  \node at (8.5,-5) {$\chi_t'$};

  \node[smallvertex,fill=blue!80] (v0-11) at ($(0,0) + (2,-1)$) {};
  \node[smallvertex,fill=red!80] (v7-11) at ($(7,0) + (2,-1)$) {};
  
  \foreach \i in {0,7}{
   \node[smallvertex,fill=cyan!80] (v\i-1) at ($(\i,0) + (0,0)$) {};
   \node[smallvertex,fill=orange!80] (v\i-2) at ($(\i,0) + (1,0)$) {};
   \node[smallvertex,fill=cyan!80] (v\i-3) at ($(\i,0) + (2,0)$) {};
   \node[smallvertex,fill=violet!80] (v\i-4) at ($(\i,0) + (3,0)$) {};
   \node[smallvertex,fill=yellow!80] (v\i-5) at ($(\i,0) + (0,1)$) {};
   \node[smallvertex,fill=Green!80] (v\i-6) at ($(\i,0) + (1,1)$) {};
   \node[smallvertex,fill=red!80, label={[label distance = -3pt]45:$u$}] (v\i-7) at ($(\i,0) + (2,1)$) {};
   \node[smallvertex,fill=Green!80] (v\i-8) at ($(\i,0) + (3,1)$) {};
   \node[smallvertex,fill=blue!80] (v\i-9) at ($(\i,0) + (2,2)$) {};
   \node[smallvertex,fill=violet!80] (v\i-10) at ($(\i,0) + (0,-1)$) {};
   \node[smallvertex,fill=red!80] (v\i-12) at ($(\i,0) + (0,-2)$) {};
   \node[smallvertex,fill=Green!80] (v\i-13) at ($(\i,0) + (1,-2)$) {};
   \node[smallvertex,fill=yellow!80] (v\i-14) at ($(\i,0) + (2,-2)$) {};
   \node[smallvertex,fill=yellow!80] (v\i-15) at ($(\i,0) + (0,-3)$) {};
   \node[smallvertex,fill=cyan!80] (v\i-16) at ($(\i,0) + (2,-3)$) {};
   \node[smallvertex,fill=red!80] (v\i-17) at ($(\i,0) + (-0.5,-1)$) {};
   \node[smallvertex,fill=yellow!80] (v\i-18) at ($(\i,0) + (-0.5,-2)$) {};
   
   \node at ($(\i,0) + (-1,0)$) {$\beta(t)$};
   \node at ($(\i,0) + (4,-3)$) {$V_t$};
   
   \foreach \v/\w in {1/5,2/6,2/7,3/7,4/7,5/6,6/9,7/8,7/9,1/10,1/17,2/11,3/11,4/11,10/12,10/13,10/18,11/13,11/14,13/14,12/15,14/16,15/16}{
    \draw[thick] (v\i-\v) edge (v\i-\w);
   }
  }
  
  \draw[line width=1.2pt,densely dotted,<->,bend right] (v7-7) edge node[label={[label distance = -6pt]180:$!$}] {} (v7-11);
 \end{tikzpicture}
 \caption{The figure shows two 7-colorings $\chi_t$ and $\chi'_t$ of $V_t$ that agree on $\beta(t)$, and hence have the same type for the (ordinary) \textsc{$7$-Coloring} problem.
  For example, both can be extended to a coloring of $G$ with the extension shown in the figure. However, if we consider the \SqCol\ problem, they are not of the same type:
  the extension shown in the figure extends $\chi_t$ to a proper square coloring, while it is not a valid extension of $\chi'_t$ (as it creates two \red\ vertices at distance $2$).}
 \label{fig:twdp}
\end{figure}
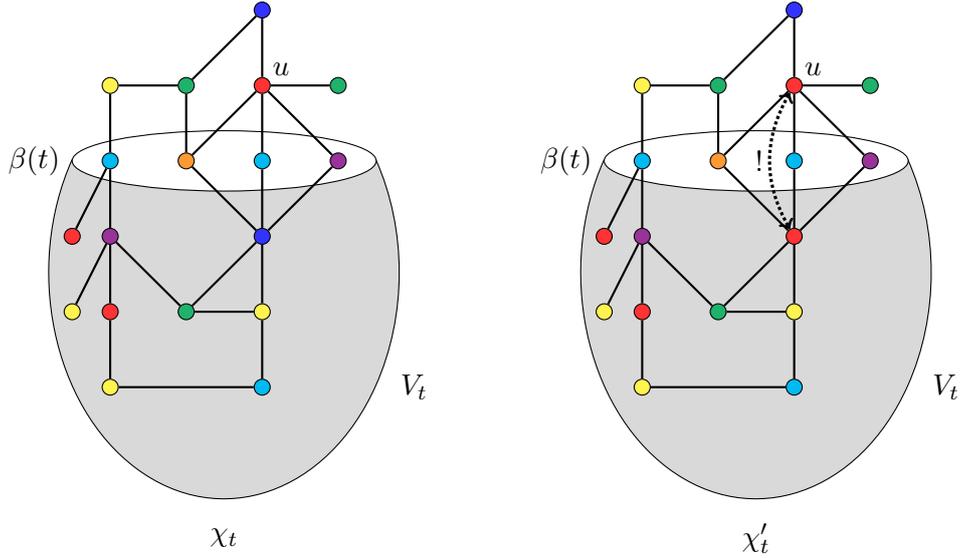

Our main insight is that instead of precisely describing which colors appear in the neighborhood of each vertex, we classify the colors according to where they appear and only consider the number of colors in each class. More precisely, the type of a partial solution $\chi_t$ of $V_t$ depends on the following.
\begin{enumerate}[label = (\Roman*)]
 \item\label{item:overview-dp-1} The coloring $\chi_t$ restricted to $\beta(t)$ (up to $q^{\ttw+1}$ possibilities).
 \item\label{item:overview-dp-2} For each color $c$ that appears on $\beta(t)$ in $\chi_t$, the subset of $S_c$ of $\beta(t)$ that has $c$ in its neighborhood (up to $(2^{\ttw+1})^{\ttw+1}$ possibilities).
 \item\label{item:overview-dp-3} For a subset $A \subseteq\beta(t)$, let $q_A$ be the number of colors $c$ that do not appear on $\beta(t)$ in $\chi_t$, but $A$ is exactly the subset of $\beta(t)$ that has $c$ in its neighborhood.
  Considering every $A \subseteq \beta(t)$, we have $(q+1)^{2^{\ttw+1}}$ possibilities for the values of $q_A$.
\end{enumerate}
All together, this gives roughly $(q+1)^{2^{\ttw+1}}$ different types of partial solutions, which matches the running time we would like to achieve.

The process of solving the subproblems in a bottom-up manner is fairly standard, but there is one more challenge if we want to achieve the claimed running time.
Typically, the main bottleneck in dynamic programming on a tree decomposition appears when handling the \emph{join nodes,} that is, nodes $t$ of the tree decomposition having exactly two children $t',t''$ and $\beta(t)=\beta(t')=\beta(t'')$ holds.
Consider a square coloring $\chi_t$ of $V_t$ and its restrictions $\chi_{t'}$ and $\chi_{t''}$ to $V_{t'}$ and $V_{t''}$, respectively.
Suppose that a color $c$ appears in the neighborhood of $A'\subseteq \beta(t)$ in $\chi_{t'}$, and on $A''\subseteq \beta(t)$ in $\chi_{t''}$.
Then, in the coloring $\chi_t$, color $c$ appears in the neighborhood of exactly $A'\cup A''$.
The main difficulty is that our description of types (as defined in the previous paragraph) does not define where a given color $c$ appears, it only specifies the number of colors that appear in the neighborhood of a given set.
This suggests that if we have a description of the type of $\chi_{t'}$ and $\chi_{t''}$, then we need to somehow match up the colors in $\chi_{t'}$ with the colors of $\chi_{t''}$ to determine what type of partial colorings of $V_t$ they can be combined into.
This can be formulated as an integer linear programming problem.
In order to solve this problem efficiently and have $2^{\ttw}$ in the exponent of $n$, we use another layer of dynamic programming.

\paragraph{Lower bound for bounded-treewidth graphs.}
As a starting point, we first prove a lower bound for a problem involving vectors, and then reduce that to \SqCol.
In the \textsc{Vector $k$-Sum} problem, we are given $k$ lists $A_1,\dots,A_k$, each list containing $n$ integer $m$-dimensional vectors, and a target vector $t$.
The task is to select exactly one vector from each list such that they sum up to exactly $t$.
Problems of similar flavor were considered before (see, e.g., \cite{AbboudLW14,DBLP:conf/stoc/Lin21}), but we need a lower bound with specific parameter settings and restrictions.
By a reduction from \textsc{Subgraph Isomorphism} for $3$-regular graphs and a known lower bound for this problem \cite{DBLP:journals/toc/Marx10}, we prove that if $k=O(2^{w})$, $m=O(2^{w})$, then there is no $f(w)n^{(2-\epsilon)^w}$ time algorithm for any $\epsilon>0$, unless \ETH\ fails.
Our reduction produces instances of \textsc{Vector $k$-Sum} with the following additional properties:
\begin{enumerate}
 \item For every $A_i$, there are at most $3$ coordinates where the vectors in $A_i$ can be non-zero.
 \item Every coordinate is non-zero in at most $2$ lists.
\end{enumerate}
As we shall see, these properties will be crucial for our lower bound.

Given the lower bound in the previous paragraph, our task is to reduce an instance of \textsc{Vector $k$-Sum} with $k=O(2^{w})$, $m=O(2^{w})$ to an instance of \SqCol\ with treewidth roughly $w$.
That is, the treewidth of the new instance should be at most \textit{logarithmic} in $k$ and $m$.
Then it follows that a $f(\ttw)n^{(2-\epsilon)^{\ttw}}$ time algorithm for any $\epsilon>0$ would violate \ETH.
While the reduction is highly technical with an elaborate construction of a sequence of gadgets, here we give a brief overview of the main ideas of the proof and in particular how the logarithmic bound for treewidth can be achieved.
It will be convenient to describe the reduction to a slight extension of the problem where some vertices are ``colorless'': they do not need to be colored in the solution, but they are relevant for computing the distances.
At the end of the proof, we show how the reduction can be modified if these vertices are also colored.

Figure~\ref{fig:intro-gadget-structure} shows the structure of the constructed instance.
On the left, we have $2m$ sets $X_1,\dots,X_m$, $Y_1,\dots,Y_m$ where $|X_i|=|Y_i|$ is exactly the $i$-th coordinate of the target vector $t$.
Let the number $q$ of colors be $\sum_{i=1}^m |X_i|$.
We want to ensure that in every square coloring with $q$ colors,
\begin{enumerate}
 \item disjoint sets of colors appear on $X_i$ and $X_j$ for $i\neq j$,
 \item the same set of colors appear on $X_i$ and $Y_i$.
\end{enumerate}
While there are many different ways of achieving this, the following construction ensures the logarithmic bound on the treewidth. We introduce a set $S$ of  $r\approx \log m$ colorless vertices. For every $i\in [m]$, we choose a distinct subset $S_i\subseteq S$ of size $r/2+1$ and connect every vertex of $X_i$ to every vertex of $S_i$, and every vertex of $Y_i$ to every vertex of $S\setminus S_i$. Note that as $\binom{r}{r/2+1}\approx 2^r \approx m$, we can choose distinct $S_i$'s. Observe that $X_i$ and $Y_i$ do not have common neighbors, allowing the use of the same colors on these two sets. As $S_i\cap S_j\neq\emptyset$ if $i\neq j$, there is a vertex of $S$ adjacent to every vertex of $X_i\cup X_j$, implying that a color cannot be used on both $X_i$ and $X_j$.
Similarly, $(S\setminus S_i)\cap S_j\neq\emptyset$ if $i\neq j$, implying that a color cannot be used on both $Y_i$ and $X_j$. Together with the bound on the number $q$ of colors, this implies that $Y_i$ has to use the same set of colors as $X_i$. Therefore, this gadget defines a partition into $m$ sets of colors, each set having the required size, and each set appearing on two sets of vertices, $X_i$ and $Y_i$.
    
\begin{figure}
 \begin{center}
  \includesvg[width=0.8\linewidth]{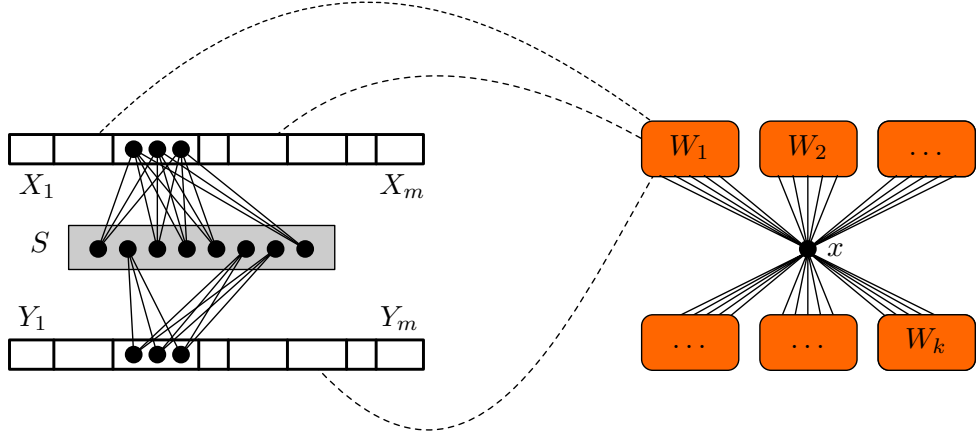}
 \end{center}
 \caption{The high-level structure of the constructed instance of \SqCol\ in the proof Theorem~\ref{thm:tw-lb-intro}.
  Each $W_i$ is connected to three color classes; here only the connections of $W_1$ are indicated.}
 \label{fig:intro-gadget-structure}
\end{figure}
    
Next, we introduce a new colorless vertex $x$ and $k$ \emph{vector selection gadgets} $W_1,\dots,W_k$, representing the lists $A_1,\dots,A_k$.
We design the gadgets in a way that ensures that their treewidth is some constant $\alpha\le 20$.
Suppose that vectors in $A_i$ are nonzero only in the three coordinates $i_1, i_2, i_3$.
Then gadget $W_i$ is attached to vertex $x$, and for every $\ell \in \{1,2,3\}$, to one of $X_{i_\ell}$ and $Y_{i_\ell}$.
As every coordinate is nonzero in the vectors of at most two lists, these attachments can be done in such a way that every $X_j$ or $Y_j$ is used only by one gadget.
Therefore, if we remove vertex $x$ and the set $S$, then the instance falls apart into disjoint gadgets.
As each gadget has treewidth at most $\alpha$, it follows that the constructed graph has treewidth at most $|S|+\alpha+1\approx \log m$, as required.

The role of the vector selection gadgets is the following.
Suppose that gadget $W_i$ is attached to $X_{i_1}\cup X_{i_2}\cup X_{i_3}$ and to vertex $x$.
We know that vertices in $X_{i_1}\cup X_{i_2}\cup X_{i_3}$ receive distinct colors and in every coloring $W_i$ ``exhibits'' some set $C_i$ of colors to $x$, that is, $C_i$ is the set of colors appearing on the neighbors of $x$ in the gadget $W_i$.
The possible square colorings of gadget $W_i$ can be classified into $n$ different ``states'', corresponding to the $n$ vectors in $A_i$.
If a vector $v_i \in A_i$ has value $a_\ell$ at coordinate $i_\ell$ for $\ell\in \{1,2,3\}$, then in the corresponding state exactly $a_\ell$ colors of $X_{i_\ell}$ are exhibited to $x$.
The colors on the neighbors of $x$ should be all different, which means that the gadgets together should exhibit at most $|X_{i_\ell}|$ colors of $X_{i_\ell}$.
In other words, if vector $v_i\in A_i$ corresponds to the state of gadget $W_i$, then $\sum_{i=1}^k v_i$ should be a vector whose $i_\ell$-th coordinate is at most $|X_{i_\ell}|$.
Observing this for every coordinate shows that $\sum_{i=1}^k v_i\le t$.
With additional arguments, this can be extended to show that there is actually equality.
Therefore, the possible combination of states of the vector selection gadgets in square colorings of the constructed graph correspond to the solutions of the \textsc{Vector $k$-Sum} instance.

We remark that the actual proof is somewhat different, for example, the sketch above ignores the fact that the gadget $W_i$ should always exhibit the same number of colors to $x$.
In the proof, we find it more convenient to define the \textsc{Vector $k$-Sum} problem such that $t=0$ and hence the vectors may have positive and negative integer values.
Then we represent each coordinate with two sets of colors and if a vector has value $a_i$ at the $i$-th coordinate, then the gadget exhibits $M-a_i$ and $M+a_i$ colors from these sets, respectively.
The proof idea described above goes through with appropriate modifications.

\paragraph{Algorithm for planar graphs.}
The subexponential algorithm in Theorem~\ref{thm:planar-alg-intro} is obtained as a combination of two algorithms (see Figure~\ref{fig:runtime}).
It is known that an $n$-vertex planar graph has treewidth $O(\sqrt{n})$, but this bound is obviously not true in general for the square of a planar graph.
However, we can obtain a useful bound if we take into account the maximum degree $\Delta$ of the graph as well.

\begin{lemma}\label{lem:twbound-intro}
 Let $G$ be a planar graph of maximum degree $\Delta$.
 Then
 \[\tw(G^2) = O\!\left(\sqrt{n\Delta}\right).\] 
\end{lemma}

We can assume that $\Delta$ is at most the number $q$ of colors, otherwise there is no solution.
Thus, we can assume that the treewidth of $G^2$ is $O(\sqrt{nq})$.
By using a $q^{O(\ttw)}\cdot n^{O(1)}$ algorithm for $q$-coloring a graph of treewidth $\ttw$, we obtain the first algorithm:

\begin{lemma}
 \label{lem:alg-planar-small-degree-intro}
 \qSqCol\ can be solved in time $q^{O(\sqrt{qn})}$ on planar graphs.
\end{lemma}

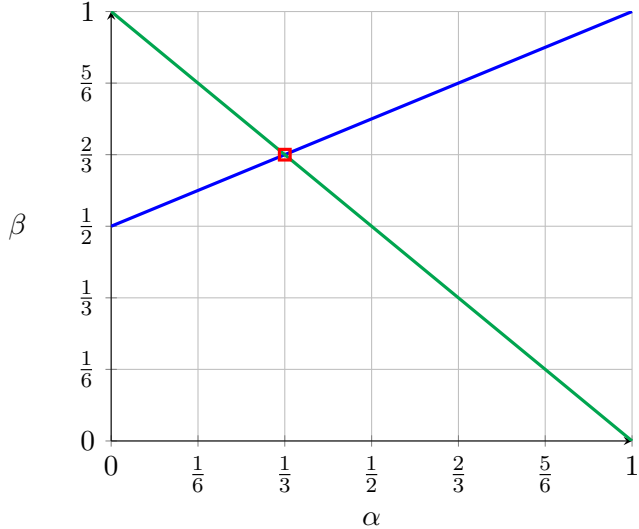
\begin{figure}
 \centering
 \begin{tikzpicture}
  \begin{axis}[axis lines = left, grid, xlabel = {$\alpha$}, ylabel = {$\beta$},
               x label style={at={(axis description cs:0.5,-0.05)}},
               xtick = {0,0.1667,0.3333,0.5,0.6667,0.8333,1},
               xticklabels={$0$, $\frac{1}{6}$, $\frac{1}{3}$, $\frac{1}{2}$, $\frac{2}{3}$, $\frac{5}{6}$, $1$},
               y label style={rotate=-90,at={(axis description cs:0,0.5)}},
               ytick = {0,0.1667,0.3333,0.5,0.6667,0.8333,1},
               yticklabels={$0$, $\frac{1}{6}$, $\frac{1}{3}$, $\frac{1}{2}$, $\frac{2}{3}$, $\frac{5}{6}$, $1$}]
   \addplot[line width = 1.2pt, domain = 0:1, color=blue]{0.5*(x+1)};
   \addplot[line width = 1.2pt, domain = 0:1, color=Green]{1-x};
   \addplot[line width = 1.2pt, color=red, mark=square] coordinates{(0.3333,0.66667)};
  \end{axis}
 \end{tikzpicture}
 \caption{The running time of the algorithms from Lemma \ref{lem:alg-planar-small-degree-intro} (blue) and \ref{lem:alg-planar-many-colors-intro} (green) as a function of the the number $q$ of colors. The $x$-axis shows the number of colors in logarithmic scale, i.e., $q = n^\alpha$.
  The $y$-axis shows the running time in double logarithmic scale, i.e., the algorithms run in time $n^{O(n^\beta)}$.}
 \label{fig:runtime}
\end{figure}  

Let us now discuss the second algorithm, which has running time $2^{O((n\log n)/q)}$.
Let $N_2[v]$ be the set of vertices at distance at most $2$ from $v$.
The initial observation is that if $|N_2[v]| \leq q$, then coloring $v$ does not present any difficulty: if the rest of the graph is colored, there is still at least one color $c$ that is unused in the distance-$2$ neighborhood of $v$ and hence we can assign $c$ to $v$.
Let $U=\{v \mid |N_2[v]|>q\}$.
Once we have a partial square coloring of the vertices of $U$, then it can be extended to a square coloring of $G$.
This implies that $G$ has a square coloring if and only if $G'=G[N[U]]$ has (note that we need to include the neighbors of $U$ into $G'$ to preserve the distance-2 paths between vertices of $U$).
It follows that we can assume that $G=G[N[U]]$, that is, $U$ is a dominating set of $G$.
Using a simple greedy selection argument, we can show that there is a subset $U'\subseteq U$ of size $O(n/q)$ that is a distance-$3$ dominating set of $G$, that is, every vertex of $G$ is at distance at most $3$ from $U'$.

Known results show that if a planar graph has a small distance-$3$ dominating set, then it can be decomposed into a smaller planar graph with small treelike attachments, called protrusions, connected to it.
Formally, an \emph{$(\alpha,\delta,k)$-protrusion decomposition} is a tree decomposition where the root bag has size at most $\alpha$, all the other bags have size at most $k$, and the root has at most $\delta$ children.
It follows from earlier work \cite{BodlaenderFLPST16} (see also \cite{FominLSZ19}) that if $G$ has a distance-$3$ dominating set of size $O(n/q)$, then it has an $(O(n/q),O(n/q),O(1))$-protrusion decomposition.

How to solve the \SqCol\ problem given such a protrusion decomposition $T$?
An obvious approach would be to try every possible coloring of the root bag ($q^{O(n/q)}$ possibilities, which is feasible in our target running time), and then somehow extend the coloring to all children of the root bag.
As the bags below the root have size $O(1)$, this should be very similar to square coloring graphs of treewidth $O(1)$, which is polynomial-time solvable by Theorem~\ref{thm:tw-alg-intro}.
However, this approach is flawed.
If $t',t''$ are two children of the root $r$, then two vertices $v'\in \beta(t')\setminus \beta(r)$ and $v''\in \beta(t'')\setminus \beta(r)$ can be at distance $2$ from each other (via a vertex in $\beta(r)$) and hence may need to receive distinct colors.
This means that even if the colors of the vertices of $\beta(r)$ are fixed, we cannot just extend the coloring to the subtree of each child independently, as such conflicts have to be avoided as well.

A natural extension to circumvent this problem is to not only guess the coloring of the root bag, but also guess the type of the coloring of each of the subtrees of the root as described in Items \ref{item:overview-dp-1} - \ref{item:overview-dp-3}.
To be more precise, let $r$ denote the root of the protrusion decomposition $T$ and let $t_1,\dots,t_\delta$ be the children of $r$.
Let us also denote by $V_i$, $i \in \{1,\dots,\delta\}$, the set of all vertices located in bags below $t_i$ (including $t_i$ itself).
Since $V_i \cap \beta(r)$ has size $O(1)$, the number of different types for each individual subtree is bounded by $q^{O(1)}$.
So in total, there are only $q^{O(\delta)} = q^{O(n/q)}$ possibilities for guessing all the types, which is still feasible in our target running time.
However, this does still not solve the problem pointed out above due to Item \ref{item:overview-dp-3}.
Let us say a color $c$ is \emph{$i$-free} if the color $c$ does not appear in $V_i \cap \beta(r)$ (given a fixed coloring of the root bag).
For each set $A \subseteq V_i \cap \beta(r)$, Item \ref{item:overview-dp-3} only provides the number $q_{i,A}$ of $i$-free colors $c$ so that $A$ is exactly the set that has $c$ in its neighborhood (only considering vertices from $V_i \setminus \beta(r)$).
Hence, we need to identify a suitable set of $i$-free colors $C_{i,A}$ of size $q_{i,A}$ which contains precisely those colors $c$.
The challenge is to do this in a consistent way.
Indeed, as already indicated above, we need to assign colors in such a way that $C_{i,A} \cap C_{i',A'} = \emptyset$ for all distinct $(i,A),(i',A')$ such that $A \cap A' \neq \emptyset$, as otherwise a vertex $v \in A \cap A'$ has two neighbors that are assigned the same color, one in the subtree rooted at $t_i$, and one in the subtree rooted at $t_{i'}$.

To solve this problem, we yet again rely on dynamic programming.
Let $\CZ$ be the set of all pairs $(i,A)$ where $i \in \{1,\dots,\delta\}$ and $A \subseteq V_i \cap \beta(r)$.
Note that $|\CZ| = O(n/q)$.
For each individual color $c \in \{1,\dots,q\}$ and every subset $\CZ' \subseteq \CZ$, we can check whether assigning $c$ to exactly those sets $C_{i,A}$ for which $(i,A) \in \CZ'$ leads to any color conflict.
Note that the number of subsets $\CZ' \subseteq \CZ$ is bounded by $2^{O(n/q)}$, so this is indeed feasible in the given time frame.
Now, checking whether $q_{i,A}$ many $i$-free colors can be assigned to $C_{i,A}$ for every $(i,A) \in \CZ$ in a consistent way can be done by a dynamic program that iteratively increases the number of available colors $q' \leq q$ and checks which ``demands'' on $i$-free colors can be met by only using colors from $\{1,\dots,q'\}$.
Once we arrive at $q' = q$ colors, we can deduce whether all guesses made so far (i.e., the coloring of the root bag and all the types of partial coloring of $V_i$) are consistent, which completes the algorithm.
Overall, we obtain the second algorithm.

\begin{lemma}
 \label{lem:alg-planar-many-colors-intro}
 \qSqCol\ can be solved in time $n^{O(n/q)}$ on planar graphs.
\end{lemma}

Finally, Theorem \ref{thm:planar-alg-intro} follows from applying the algorithm from Lemma \ref{lem:alg-planar-small-degree-intro} for $q \leq n^{1/3}$, and the algorithm from Lemma \ref{lem:alg-planar-many-colors-intro} for $q \geq n^{1/3}$ (see also Figure \ref{fig:runtime}).

\paragraph{NP-hardness for planar graphs.}
The \NP-hardness of \SqCol\ on planar graphs was shown by Lloyd and Ramanathan \cite{DBLP:conf/icci/LloydR92}.
However, they only show hardness for $q = 7$ colors.
We extend this to any fixed number $q \ge 4$ of colors by reducing from \textsc{$3$-Coloring} restricted to planar input graphs.
The \NP-hardness of this problem has been proven in \cite{Stockmeyer73}.
Let $G$ denote a planar graph for which we wish to test $3$-colorability.
In order to create an equivalent instance of \SqCol\ (with $q$ colors) we follow the natural strategy of splitting the set of $q$ colors into a set of $3$ ``candidate'' colors and $q-3$ ``auxiliary'' colors.
Now, the basic idea is to extend $G$ by certain gadgets that ensure that every original vertex of $G$ needs to be colored by one of the ``candidate'' colors, and adjacent vertices in the original graph $G$ need to receive different colors.
The main challenge of this approach is to devise a method to distribute the information of which colors are the ``candidate'' colors over the entire graph.

Here, our main insight is that, for every plane graph $G$ (i.e., a planar graph together with its embedding in the plane) of minimum degree $3$, one can draw a circle on the plane that crosses every edge of $G$ exactly twice (and does not intersect any vertex of $G$).
An example is given in Figure \ref{fig:planar-hardness-cycle}.
Given such a circle, we can distribute the information of which colors are the ``candidate'' colors and which are the ``auxiliary'' colors along the circle while preserving planarity.
Suppose $q \geq 5$.
Figure \ref{fig:planar-hardness-edge} shows the constructed instance of \SqCol\ when zooming in on a single edge $uv \in E(G)$ of the original graph $G$.
In the example, the colors \red, \blue\ and \green\ are the ``candidate'' colors and \yellow, \orange, \violet\ and \cyan\ are the ``auxiliary'' colors.
This information is forwarded along the blue circle by the construction shown the Figure \ref{fig:planar-hardness-edge}.
To be more precise, we introduce sets of vertices $C_i$ and $A_i$ for all $i \in [2|E(G)|]$ where each $C_i$ contains $3$ vertices and each $A_i$ contains $q-3 \geq 2$ vertices.
The sets $C_i$ and $A_i$ are placed alternately on the constructed circle where each $C_i$ is associated with a crossing between the constructed circle and an edge of $G$ (see also Figure \ref{fig:planar-hardness}).
For every $i \in [2|E(G)|]$ we introduce edges between $A_i$ and $C_i$ as well as $C_i$ and $A_{i+1}$ as indicated in Figure \ref{fig:planar-hardness-edge}.
We obtain that, in the square graph, both $A_i \cup C_i$ and $C_i \cup A_{i+1}$ form cliques, which means that all $q$ colors have to be used on both sets.
In particular, this implies that all sets $A_i$ and $A_{i+1}$ have to be colored by the same set of colors which are declared to be the ``auxiliary'' colors.
Overall, we obtain that all sets $A_i$ are colored by the ``auxiliary'' colors and all sets $C_i$ are colored by the ``candidate'' colors.

Looking at Figure \ref{fig:planar-hardness-edge}, it can also be checked that $u$ and $v$ need to be assigned distinct ``candidate'' colors.
Indeed, the only way to color $C_i$ and $C_j$ in a consistent way is to color them in the same way from top to bottom.
So overall, we can thus ensure that all vertices of $G$ are colored by one of the ``candidate'' colors and adjacent vertices need to receive distinct colors.
This almost completes the reduction.
As the last remaining step, we only need to ensure that gadgets living on two edges incident to the same vertex do not interfere with one another.
However, this can easily be assured by some simple gadgets that are introduced at every vertex of $G$.

\begin{figure}
 \centering
 \begin{subfigure}[t]{0.5\linewidth}
  \centering
  \begin{tikzpicture}
   \path[use as bounding box] (90:4.5) -- (210:4.5) -- (330:4.5) -- cycle;
   
   \node[normalvertex] (1) at (0,0) {};
   \node[normalvertex] (2) at (90:3) {};
   \node[normalvertex] (3) at (210:3) {};
   \node[normalvertex] (4) at (330:3) {};
   
   \foreach \i/\j in {1/2,1/3,1/4,2/3,2/4,3/4}{
    \draw[thick] (\i) edge (\j);
   }
   
   \node (12) at (90:1) {};
   \node (13) at (210:1) {};
   \node (14) at (330:1) {};
   \node (21) at (90:2) {};
   \node (31) at (210:2) {};
   \node (41) at (330:2) {};
   
   \node (23) at ($(90:3) + (240:1.732)$) {};
   \node (32) at ($(90:3) + (240:3.464)$) {};
   \node (24) at ($(90:3) + (300:1.732)$) {};
   \node (42) at ($(90:3) + (300:3.464)$) {};
   \node (34) at ($(210:3) + (0:1.732)$) {};
   \node (43) at ($(210:3) + (0:3.464)$) {};
   
   \draw[line width = 4pt, blue, opacity=0.5, rounded corners = 0.3cm]
    (21.center) -- (41.center) -- (31.center) ..controls ($(210:1.5) + (120:0.3)$)..
    (13.center) -- (14.center) -- (12.center) --
    (32.center) ..controls (195:4.5) and (225:4.5).. (34.center) ..controls (265:1.2) and (275:1.2)..
    (43.center) ..controls (315:4.5) and (345:4.5).. (42.center) ..controls (25:1.2) and (35:1.2)..
    (24.center) ..controls (75:4.5) and (105:4.5).. (23.center) -- cycle;
  \end{tikzpicture}
  \caption{The planar graph $K_4$ together with a circle that crosses every edge of $K_4$ exactly twice.}
  \label{fig:planar-hardness-cycle}
 \end{subfigure}
 \begin{subfigure}[t]{0.45\linewidth}
  \centering
  \begin{tikzpicture}
   \draw[blue!25,fill=blue!25] (-2.6,4.4) rectangle (4.6,6);
   \draw[blue!25,fill=blue!25] (-2.6,1) rectangle (4.6,2.6);
   
   \node[normalvertex,fill=blue!80,label={above:$u$}] (22) at ($(1,5)+(0,1.3)$) {};
   \node[normalvertex,fill=red!80] (23) at ($(1,5)+(0,0.2)$) {};
   \node[normalvertex,fill=Green!80] (24) at ($(1,5)+(0,-0.9)$) {};
   
   \node[normalvertex,fill=blue!80] (32) at ($(1,2)+(0,0.9)$) {};
   \node[normalvertex,fill=red!80] (33) at ($(1,2)+(0,-0.2)$) {};
   \node[normalvertex,fill=Green!80,label={below:$v$}] (34) at ($(1,2)+(0,-1.3)$) {};
   
   \foreach[count=\i] \x/\y in {-0.8/5.2,2.8/5.2}{
    \node[normalvertex,fill=yellow!80] (e\i1) at ($(\x,\y)+(0.6,0)$) {};
    \node[normalvertex,fill=cyan!80] (e\i2) at ($(\x,\y)+(-0.6,0)$) {};
    \node[normalvertex,fill=orange!80] (e\i3) at ($(\x,\y)+(0,0.6)$) {};
    \node[normalvertex,fill=violet!80] (e\i4) at ($(\x,\y)+(0,-0.6)$) {};
    \foreach \v/\w in {1/2,1/3,1/4,2/3,2/4}{
     \draw[thick] (e\i\v) edge (e\i\w);
    }
   }
   \foreach \i/\x/\y in {3/-0.8/1.8,4/2.8/1.8}{
    \node[normalvertex,fill=cyan!80] (e\i1) at ($(\x,\y)+(0.6,0)$) {};
    \node[normalvertex,fill=yellow!80] (e\i2) at ($(\x,\y)+(-0.6,0)$) {};
    \node[normalvertex,fill=orange!80] (e\i3) at ($(\x,\y)+(0,0.6)$) {};
    \node[normalvertex,fill=violet!80] (e\i4) at ($(\x,\y)+(0,-0.6)$) {};
    \foreach \v/\w in {1/2,1/3,1/4,2/3,2/4}{
     \draw[thick] (e\i\v) edge (e\i\w);
    }
   }
   
   \draw[thick] (e12) edge ($(e12.center)+(150:0.6)$);
   \draw[thick] (e12) edge ($(e12.center)+(180:0.6)$);
   \draw[thick] (e12) edge ($(e12.center)+(210:0.6)$);
   \draw[thick] (e32) edge ($(e32.center)+(150:0.6)$);
   \draw[thick] (e32) edge ($(e32.center)+(180:0.6)$);
   \draw[thick] (e32) edge ($(e32.center)+(210:0.6)$);
   \draw[thick] (e21) edge ($(e21.center)+(30:0.6)$);
   \draw[thick] (e21) edge ($(e21.center)+(0:0.6)$);
   \draw[thick] (e21) edge ($(e21.center)+(-30:0.6)$);
   \draw[thick] (e41) edge ($(e41.center)+(30:0.6)$);
   \draw[thick] (e41) edge ($(e41.center)+(0:0.6)$);
   \draw[thick] (e41) edge ($(e41.center)+(-30:0.6)$);
   
   \foreach \i/\j in {22/23,23/24,24/32,32/33,33/34,e11/22,e11/23,e11/24,e22/22,e22/23,e22/24,e31/32,e31/33,e31/34,e42/32,e42/33,e42/34}{
    \draw[thick] (\i) edge (\j);
   }
   
   \draw[dashed, line width = 3pt, rounded corners, gray!80] (0.6,0.1) rectangle (1.4,3.3);
   \draw[dashed, line width = 3pt, rounded corners, gray!80] (0.6,3.7) rectangle (1.4,6.9);
   
   \draw[dashed, line width = 3pt, rounded corners, gray!80] (-1.7,0.9) rectangle (0.1,2.7);
   \draw[dashed, line width = 3pt, rounded corners, gray!80] (1.9,0.9) rectangle (3.7,2.7);
   \draw[dashed, line width = 3pt, rounded corners, gray!80] (-1.7,4.3) rectangle (0.1,6.1);
   \draw[dashed, line width = 3pt, rounded corners, gray!80] (1.9,4.3) rectangle (3.7,6.1);
   
   \node at (1,-0.3) {$C_j$};
   \node at (1,7.3) {$C_i$};
   
   \node at (2.8,6.5) {$A_i$};
   \node at (-0.8,6.5) {$A_{i+1}$};
   \node at (-0.8,0.5) {$A_j$};
   \node at (2.8,0.5) {$A_{j+1}$};
  \end{tikzpicture}
  \caption{The constructed instance of \SqCol\ for $q = 7$ when zooming in on an edge $uv$ of the original graph $G$.}
  \label{fig:planar-hardness-edge}
 \end{subfigure}
 \caption{Visualization of the reduction from \textsc{$3$-Coloring} to \SqCol\ on planar graphs.}
 \label{fig:planar-hardness}
\end{figure}
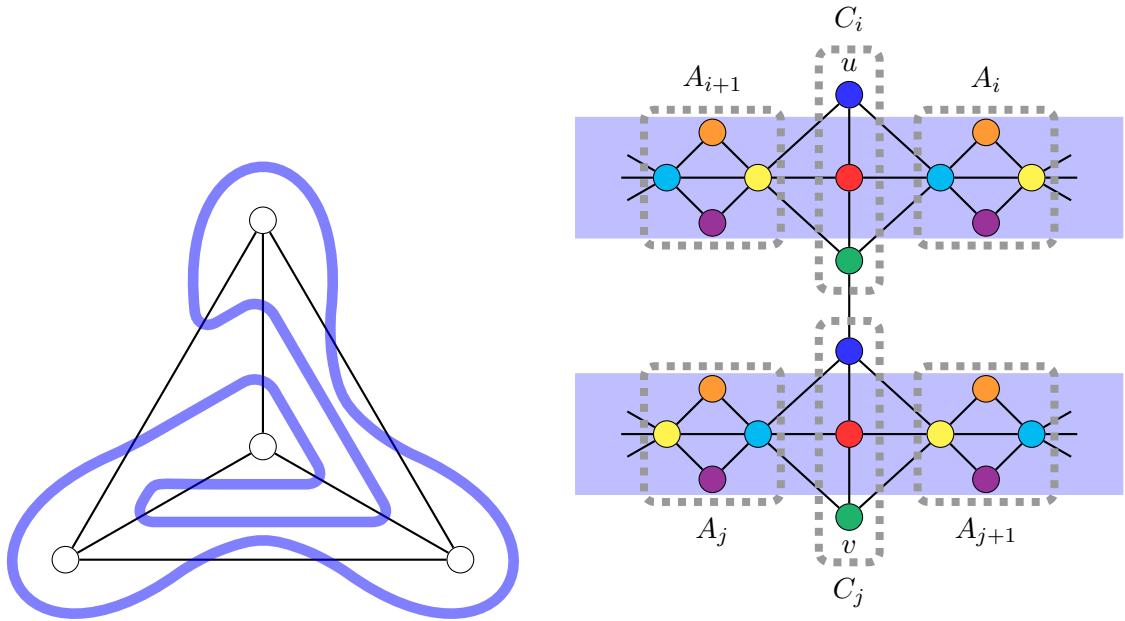

Observe that the above construction requires $q \geq 5$ since the \red\ vertices in Figure \ref{fig:planar-hardness-edge} always have degree $4$ (independent of the number of colors $q$).
For $q = 4$, we design a specialized reduction which is slightly easier than the one described above.
Here, the main insight is that we only need to distribute one ``auxiliary'' color which can be done in a more direct way while preserving planarity.

Finally, let us also point out that the constructed instance of \SqCol\ has $O(qn)$ many vertices (where $n$ denotes the number of vertices of $G$).
As a result, building on known lower bounds for \textsc{$3$-Coloring} on planar graphs (see, e.g., \cite[Theorem 14.9]{CyganFKLMPPS15}), we also obtain the following conditional lower bound for \SqCol. 

\begin{theorem}\label{thm:planar-eth-lb-intro}
 Assuming \ETH, there is no $2^{o(\sqrt{n})}$ time algorithm solving \SqCol\ on planar graphs.
\end{theorem}

In particular, assuming \ETH, the algorithm from Lemma \ref{lem:alg-planar-small-degree-intro} is essentially optimal for constant number of colors $q$.

\subsection{Structure of the Paper}

After providing the necessary preliminaries in the next section, we prove Theorem \ref{thm:tw-alg-intro} in Section \ref{sec:tw-alg}.
Afterward, we show a matching lower bound (Theorem \ref{thm:tw-lb-intro}) in Section \ref{sec:tw-lower-bound}.
In Section \ref{sec:planar-alg}, we prove Lemmas \ref{lem:alg-planar-small-degree-intro} and \ref{lem:alg-planar-many-colors-intro} which together imply Theorem \ref{thm:planar-alg-intro}.
Finally, we provide the hardness results for \SqCol\ on planar graphs in Section \ref{sec:planar-lower-bound} where we show Theorems \ref{thm:planar-lb-intro} and \ref{thm:planar-eth-lb-intro}.

\section{Preliminaries}
\label{sec:preliminaries}
\paragraph{Basics.}

For an integer $k \geq 1$ we denote $[k] \coloneqq \{1,2,\dots,k\}$ and $[k]_0 \coloneqq \{0,1,\dots,k\}$.
We define $\ZZp$ to be the set of positive integers, and $\ZZnn$ to be the set of non-negative integers.
For a finite set $X$ we denote by $2^{X}$ the powerset of $X$, i.e., the set of all subsets of $X$, and $\binom{X}{k}$ denotes the set of all $k$-element subsets of $X$.
If not otherwise stated, $\log$ denotes the logarithm with base $2$.

We use standard notation for graphs.
A \emph{graph} is pair $G = (V(G),E(G))$ with finite vertex set $V(G)$ and edge set $E(G) \subseteq \binom{V(G)}{2}$.
Unless stated otherwise, all graphs considered in this paper are simple (i.e., there are no loops or mulitedges) and undirected.
We use $uv$ as a shorthand for edges $\{u,v\} \in E(G)$.
The \emph{(open) neighborhood} of a vertex $v \in V(G)$ is denoted by $N_G(v) \coloneqq \{w \in V(G) \mid vw \in E(G)\}$.
The \emph{degree} of $v$ is the size of its neighborhood, i.e., $\deg_G(v) \coloneqq |N_G(v)|$.
The \emph{closed neighborhood} is $N_G[v] \coloneqq N_G(v) \cup \{v\}$.
More generally, for a set $X \subseteq V(G)$, we define $N_G(X) \coloneqq (\bigcup_{v \in X} N_G(v)) \setminus X$ and $N_G[X] \coloneqq N_G(X) \cup X$.
We usually omit the index $G$ if it is clear from context.

A \emph{path of length $k$} between two vertices $v,w \in V(G)$ is a sequence of pairwise distinct vertices $v = u_0,\dots,u_k = w$ such that $u_{i-1}u_i \in E(G)$ for all $i \in [k]$.
The \emph{distance} between $v$ and $w$, denoted by $\dist_G(v,w)$, is the length of the shortest path between $v$ and $w$.
For $X \subseteq V(G)$ we write $G[X]$ to denote the \emph{induced subgraph} on vertex set $X$, and $G - X \coloneqq G[V(G) \setminus X]$ denotes the induced subgraph on the complement of $X$.

\paragraph{Colorings.}

Let $G$ be a graph.
For $q \geq 1$ a \emph{(proper) $q$-coloring} of $G$ is a mapping $\chi\colon V(G) \to [q]$ such that $\chi(v) \neq \chi(w)$ for all edges $vw \in E(G)$.
A \emph{square $q$-coloring} of $G$ is a mapping $\chi\colon V(G) \to [q]$ such that $\chi(v) \neq \chi(w)$ for all distinct $v,w \in V(G)$ such that $\dist_G(v,w) \leq 2$.
Observe that a mapping $\chi\colon V(G) \rightarrow [q]$ is a square $q$-coloring of $G$ if and only if $\chi$ is a $q$-coloring of the \emph{square graph} $G^{2}$ defined by $V(G^{2}) \coloneqq V(G)$ and
\[E(G^{2}) \coloneqq \{vw \mid v \neq w, \dist_G(v,w) \leq 2\}.\]
A \emph{coloring} of $G$ is a $q$-coloring of $G$ for some number $q \geq 1$.
Similarly, a \emph{square coloring} of $G$ is a square $q$-coloring of $G$ for some number $q \geq 1$.
We consider the following computational problem.

\defproblem{Square Coloring}{A graph $G$, and an integer $q$.}{Is there a square $q$-coloring of $G$?}

Also, we call \qSqCol\ the variant of the problem where the number of colors is fixed to $q$.

Given a graph $G$ and an arbitrary mapping $\chi\colon V(G) \to [q]$, we define a \emph{(color) conflict} to be a pair of vertices violating the square coloring constraint,
that is, a (color) conflict is pair of distinct vertices $v,w \in V(G)$ such that $\dist_G(v,w) \leq 2$ and $\chi(v) = \chi(w)$.

\paragraph{Tree decompositions.}

Next, we define tree decompositions and state some basic properties.
For a more thorough introduction to treewidth and its many applications, we refer the reader to~\cite[Chapter 7]{CyganFKLMPPS15}.

Let $G$ be a graph.
A \emph{tree decomposition} of $G$ is a pair $(T,\beta)$ consisting of a rooted tree $T$ and a mapping $\beta\colon V(T) \to 2^{V(G)}$ such that
\begin{enumerate}[label = (T.\arabic*)]
 \item $\bigcup_{t \in V(T)} \beta(t) = V(G)$,
 \item for every edge $vw \in E(G)$ there is some node $t \in V(T)$ such that $\{u,v\} \subseteq \beta(t)$, and
 \item for every $v \in V(G)$ the set $\{t \in V(T) \mid v \in \beta(t)\}$ induces a connected subtree of $T$.
\end{enumerate}
The \emph{width} of a tree decomposition $(T,\beta)$ is defined as $\max_{t \in V(T)}|\beta(t)|-1$.
The \emph{treewidth} of a graph $G$, denoted by $\tw(G)$, is the minimum width of a tree decomposition of $G$.

The following fact on tree decompositions is well-known.

\begin{observation}
 \label{obs:td-connected-subtree-from-connected-subgraph}
 Let $G$ be a graph and let $(T,\beta)$ be a tree decomposition of $G$.
 Also let $S \subseteq V(G)$ such that $G[S]$ is connected.
 Then
 \[T_S \coloneqq \{t \in V(T) \mid \beta(t) \cap S \neq \emptyset\}\]
 induces a connected subtree of $T$.
\end{observation}

When designing algorithms on graphs of bounded treewidth, one usually builds on nice tree decompositions.
Let $(T,\beta)$ be a tree decomposition and denote $X_t \coloneqq \beta(t)$ for $t \in V(T)$.
The decomposition is \emph{nice} if $X_r = \emptyset$ where $r$ denotes the root of $T$, $X_\ell = \emptyset$ for all leaves $\ell \in V(T)$, and every internal node $t \in V(T)$ has one of the following types:
\begin{description}
 \item[Introduce:] $t$ has exactly one child $t'$ and $X_t = X_{t'} \cup \{v\}$ for some $v \notin X_{t'}$; the vertex $v$ \emph{is introduced at $t$},
 \item[Forget:] $t$ has exactly one child $t'$ and $X_t = X_{t'} \setminus \{v\}$ for some $v \in X_{t'}$; the vertex $v$ \emph{is forgotten at $t$}, or
 \item[Join:] $t$ has exactly two children $t_1,t_2$ and $X_t = X_{t_1} = X_{t_2}$.
\end{description}

It is well-known that every tree decomposition $(T,\beta)$ of $G$ of width $\ttw$ can be turned into a nice tree decomposition of the same width $\ttw$ of size $O(\ttw \cdot |V(T)|)$ in time $O(\ttw^{2} \cdot \max(|V(G)|,|V(T)|))$ (see, e.g., \cite[Lemma 7.4]{CyganFKLMPPS15}).

When dealing with tree decompositions of graphs, we adopt the standard notion of using the word ``node'' for nodes of $T$ and the word ``vertex'' for vertices of $G$.

For upper-bounding the treewidth of a graph, there is a useful characterization via a graph-theoretic game called \textit{cops and robbers}.
The game is parameterized by a number $k$ and a graph $G$, the playing field.
It has two teams, one with $k$ cops and one with a single robber.
The rules are as follows.
The $k$ cops select their starting vertices in the graph.
Then the robber may choose their starting vertex.
The cops can always see the robber and adapt their strategy accordingly.
Similarly, the robber can see the cops.
The game now proceeds in rounds, where in each round, one of the cops chooses an arbitrary destination vertex and takes off via helicopter in the direction of that vertex.
While the cop is traveling, the robber sees where they will land and may now move arbitrarily along edges of the graph, as long as they do not pass through stationary cops.
When the robber has finished moving, the cop lands.
The cops win if they catch the robber after a finite number of moves.
Otherwise, the robber wins.
We say that \emph{$k$ cops can catch a robber on $G$} if the $k$ cops have a winning strategy in this game.

\begin{theorem}[\cite{SeymourT93}]
 \label{thm:cops-and-robbers-characterization-of-treewidth}
 A graph $G$ has treewidth at most $\ttw$ if and only if $\ttw+1$ cops can catch a robber on $G$.
\end{theorem}

\paragraph{Planar graphs.}
A \emph{plane graph} is a pair consisting of a graph $G$ and an embedding of $G$ in the plane (without any edge crossings).
A graph $G$ is \emph{planar} if it can be embedded in the plane.
Let $G$ be a planar graph and fix some embedding of $G$ in the plane.
For each $v \in V(G)$ we can cyclically order the incident edges $E_G(v) \coloneqq \{e \in E(G) \mid v \in e\}$ clockwise according to the fixed embedding.
In our constructions of planar, we occasionally require the existence of planar embeddings where this cyclic ordering of incident edges satisfies certain properties.
Let us remark at this point that fixing such a cyclic ordering for every $v \in V(G)$ uniquely describes an embedding of $G$ in the plane (if it exists) up to homeomorphisms.
Hence, it usually suffices to give such a cyclic ordering for every vertex of $G$ to describe the embedding.
This is usually referred to as a \emph{combinatorial embedding of $G$}.
It is well-known that, given a planar graph $G$, a combinatorial embedding of $G$ can be computed in polynomial time (see, e.g., \cite{MoharT01}).

\section{Dynamic Programming for Graphs of Bounded Treewidth}
\label{sec:tw-alg}
In this section, we show that \SqCol\ can be solved in polynomial time over graphs of bounded treewidth.
More precisely, the main result of this section is the following theorem.

\begin{theorem}[Theorem \ref{thm:tw-alg-intro} restated]
 \label{thm:tw-dp}
 There is an algorithm solving \SqCol\ in time $(q+1)^{2^{(\ttw+4)}}\cdot n^{O(1)}$ on input graphs of treewidth $\ttw$.
\end{theorem}

We prove this by describing a dynamic programming algorithm on the tree decomposition.

First observe that, using the algorithm from \cite{Bodlaender96}, we can compute a tree decomposition of $G$ of width $\ttw$ in time $2^{O(\ttw^{3} \log \ttw)} \cdot n = (q+1)^{2^{(\ttw+4)}}\cdot n^{O(1)}$.
This can be converted to a nice tree decomposition of the same width with $O(\ttw \cdot n)$ bags in time $O(\ttw^2 \cdot n)$~\cite[Lemma 7.4]{CyganFKLMPPS15}.
This means that if we can show the existence of an algorithm with running time $(q+1)^{2^{(\ttw+4)}}\cdot n^{O(1)}$ for the same problem where an optimal nice tree decomposition is already given, we are done, since that running time dominates both aforementioned running times.

Hence, it suffices to show that such an algorithm exists.
For the remainder of this section let us fix an input graph $G$, a number of colors $q \leq n$, and a nice tree decomposition $(T,\beta)$ of width $\ttw$ such that $|V(T)| = O(\ttw \cdot n)$.
For $t \in V(T)$ we denote $X_t \coloneqq \beta(t)$ and $V_t$ denotes the union of all bags located in the subtree rooted at node $t$ (including $t$ itself).

We use dynamic programming (DP) over the nice tree decomposition.
The obvious approach would be to have a DP table where a state of a bag encodes for each vertex of the bag what its color is, as well as a list of what colors are within distance one.
This way we can always check whether a coloring of a vertex is valid by asserting that none of its neighbors have the same color and that none of its neighbors have the color in their list.
However, this introduces a factor of $2^{\tw \cdot q}$ in the running time, where $q$ is the number of colors.
Since the necessary number of colors may be as large as $n$, this can be exponential in $n$, which we do not want.

Instead, we track equivalence classes of colors.
For a fixed node $t \in V(T)$ and a coloring $\echi: V_t \to [q]$, we say that a vertex $v \in X_t$ is \emph{cone-adjacent} to a color $c \in [q]$ if there exists a vertex $u \in N(v) \cap (V_t \setminus X_t)$ such that~$\echi(u) = c$.
For a fixed color $c \in [q]$, we call the set of vertices $v \in X_t$ which are cone-adjacent to $c$ the \emph{bag-adjacent} set of $c$, i.e.,
\[A_t(\echi,c) \coloneqq \{v \in X_t \mid v \text{ is cone-adjacent to } c\}.\]
Now, for a specific $t \in V(T)$ and a coloring $\echi: V_t \to [q]$, we say that two colors $c_1,c_2 \in [q]$ are \emph{equivalent} if $A_t(\echi,c_1) = A_t(\echi,c_2)$.
We naturally identify an equivalence class by the bag-adjacent set $A \in 2^{X_t}$ that is shared by all colors the equivalence class contains.

The main insight in our dynamic programming algorithm is that apart from the coloring of the vertices in a bag $X_t$, the only relevant information that gets transmitted from one side of the bag to the other is the size of each of the equivalence classes.
Specifically, we store the number of \emph{free colors} in each equivalence class, i.e., colors that do not also appear in the bag.
  
Now let us describe the dynamic programming algorithm in detail.
We begin by specifying the structure and content of the dynamic programming table.

\begin{definition}[Entries of the DP table]
 \label{def:dp-table}
 For each node $t \in V(T)$, each $\chi: X_t \to [q]$, each $\xi: X_t \to 2^{X_t}$ and each mapping $\rho: 2^{X_t}\to [q]_0$, there is an entry $D[t][\chi][\xi][\rho]$ which is set to true ($\top$) if there exists a square $q$-coloring $\echi$ of $G[V_t]$ such that
  \begin{enumerate}[label=(DP.\arabic*)]
  \item\label{item:dp-chi} $\chi(v) = \echi(v)$ for all $v \in X_t$,
  \item\label{item:dp-xi} $\xi(v) = A_t(\echi,\chi(v))$ for all $v \in X_t$, i.e., the color of $v$ is contained in the equivalence class $\xi(v)$, and
  \item\label{item:dp-rho} \[\rho(A) = |\{c \in [q] \mid A_t(\echi,c) = A\} \setminus \{\chi(v) \mid v \in X_t\}|\] for all $A \in 2^{X_t}$, i.e., $\rho(A)$ gives the number of free colors (i.e., colors in $[q] \setminus \chi(X_t)$) that are in the equivalence class $A$.
  \end{enumerate}
\end{definition}

Note that for every $t \in V(T)$ and every coloring $\echi$ of $V_t$ there exist corresponding induced mappings $\chi$, $\xi$ and $\rho$.
In particular, for every square $q$-coloring of $G[V_t]$, there is a corresponding entry in the partial table $D[t]$ that is set to true.

Also observe that the DP table defined above has
\begin{equation}
 \label{eq:size-dp-table}
 |D[t]| \leq q^{\ttw+1}\cdot (2^{\ttw+1})^{\ttw+1} \cdot (q+1)^{2^{(\ttw+1)}} \leq (q+1)^{\ttw+1 + (\ttw+1)^2 + 2^{\ttw+1}}
\end{equation}
many entries for each node $t \in V(T)$.

The following lemma is the main technical result of this section.

\begin{lemma}
 \label{la:dp-table}
 There is an algorithm that, given a graph $G$, a number $q$ and a nice tree decomposition $(T,\beta)$ of $G$ of width $\ttw$, computes all table entries $D[t][\chi][\xi][\rho]$ in time $|V(T)| \cdot (q+1)^{2^{(\ttw+4)}} \cdot n^{O(1)}$.
\end{lemma}

Before diving into the proof of Lemma \ref{la:dp-table}, let us first derive Theorem \ref{thm:tw-dp} assuming the lemma holds true.

\begin{proof}[Proof of Theorem \ref{thm:tw-dp}]
 Let $G$ denote the input graph and $q$ the number of colors.
 The algorithm first computes a tree decomposition of $G$ with $O(n)$ bags and width $\ttw = \tw(G)$ in time $2^{O(\ttw^{3} \log \ttw)} \cdot n = (q+1)^{2^{(\ttw+4)}}\cdot n^{O(1)}$ using the algorithm from \cite{Bodlaender96}.
 This decomposition can be turned into a nice tree decomposition with $O(\ttw \cdot n)$ bags and width $\ttw$ in time $O(\ttw^2 \cdot n)$.
 
 The algorithm computes all entries $D[t][\chi][\xi][\rho]$ of the DP table in time $(q+1)^{2^{(\ttw+4)}}\cdot n^{O(1)}$.
 After having calculated all entries, we take the root $r$ of $T$ (which, by definition of a nice tree decomposition, satisfies $X_r = \emptyset$) and output YES if and only if $D[r][\varepsilon_1][\varepsilon_2][\rho] = \top$ for the two empty functions $\varepsilon_1: \emptyset \to [q], \varepsilon_2: \emptyset \to \{\emptyset\}$ and $\rho: \{\emptyset\} \to [q]_0$ with $\rho(\emptyset) = q$.

 The correctness of this algorithm follows from Definition~\ref{def:dp-table}.
 Indeed, $D[r][\varepsilon_1][\varepsilon_2][\rho] = \top$ if and only if there exists a coloring of $V_r = V(G)$ such that for the (empty) bag of the root, we have that \ref{item:dp-chi} the coloring coincides with the empty coloring for this bag, \ref{item:dp-xi} for each vertex of the empty bag, its color is contained in the equivalence class defined by $\xi$, and \ref{item:dp-rho} for each $S \in 2^\emptyset$, we have that $\rho(S)$ gives the number of free colors that are in the equivalence class $S$ -- in this case, we check that all $q$ colors are indeed in the equivalence class of the empty set.
 In other words, $D[r][\varepsilon_1][\varepsilon_2][\rho] = \top$ if and only if there exists a square coloring of $G$.
 
 To complete the proof, observe that the overall running time of the algorithm is bounded by $(q+1)^{2^{(\ttw+4)}}\cdot n^{O(1)}$.
\end{proof}

We now turn to the proof of Lemma \ref{la:dp-table}.
Here, the following definition turns out to be useful which formulates basic, local conditions for a table entry to evaluate to true. 

\begin{definition}[Local validity]
 \label{def:local-validity}
 We say that a 4-tuple $(t, \chi, \xi, \rho)$ is \textit{locally invalid} if any of the following conditions are met:
 \begin{enumerate}[label=(L.\arabic*)]
  \item\label{item:locally-invalid-1} there exist two adjacent vertices $u,v \in X_t$ such that~$\chi(u) = \chi(v)$ (i.e., the coloring in the bag is invalid due to two vertices with distance one having the same color),
  \item\label{item:locally-invalid-2} there exist three vertices $u,v,w \in X_t$ such that~$uv, vw \in E(G)$ and $\chi(u) = \chi(w)$ (i.e., the coloring in the bag is invalid due to two vertices with distance two having the same color), or
  \item\label{item:locally-invalid-3} there exist two adjacent vertices $u,v \in X_t$ such that~$v \in \xi(u)$ (i.e., a vertex $u$ has a color which occurs in the set of cone-adjacent colors of one of its neighbors $v$ and is thus within distance two).
 \end{enumerate}
\end{definition}

It is easy to check that $D[t][\chi][\xi][\rho]$ is false for every locally invalid tuple $(t, \chi, \xi, \rho)$.

Also, we shall use another dynamic programming algorithm as a subroutine to solve certain integer linear programs efficiently.
More precisely, we rely on the following subroutine.

\begin{lemma}
 \label{la:ilp-solver}
 Given an ILP instance of the form $\{Ax = b; \forall i: 0 \leq x_i \leq \|b\|_{\infty};\ x \text{ integer}\}$ where $A$ and $b$ have non-negative entries, we can determine feasibility of that instance in time $O(V \cdot C \cdot (\|b\|_{\infty}+1)^{C+1})$, where $V$ is the number of variables and $C$ is the number of constraints (i.e., $A$ is a $C\times V$ matrix). 
\end{lemma}

\begin{proof}
 Let $x = (x_1, \ldots, x_V)$.
 We show that there exists a dynamic programming algorithm that achieves this running time.
 We have a DP table $D$ structured as follows.
 For each $0 \leq i \leq V$ and each $(d_1, \dots, d_C) \in \{0, \dots, \|b\|_{\infty}\}^C$, we have that $D[i][d_1][d_2]\ldots[d_C] = \top$ (here, $\top$ refer to true) if there is an assignment to the variables $x_1, \dots, x_V$ such that
 \begin{enumerate}
  \item $\forall j > i: x_j = 0$ and
  \item $\forall j \in [C]: A_j\cdot x = d_j$ (where $A_j$ is the $j$-th row in $A$)
 \end{enumerate}
 and $\perp$ otherwise.
 Intuitively, the entry reflects whether there exists a partial assignment to the first $i$ variables such that for each $j$, the left-hand side of the $j$-th equality constraint has value $d_j$.
 Clearly, the solution to the entire instance is then stored in $D[V][b_1][b_2]\ldots[b_C]$.
 
 We compute these values bottom-up.
 The base case $D[0][0][0]\ldots[0] = \top$ is obvious.
 Now for each value of $i$ from $1$ to $V$, we do the following.
 For each choice of $x_i \leftarrow v_i$ where $0\leq v_i \leq \|b\|_{\infty}$, we compute the product $w \coloneqq A \cdot v$, where $v$ is a vector with $v_i$ in the $i$-th position and zeros everywhere else.
 Then, for each value of $d = (d_1,\dots, d_C) \in \{0, \dots, \|b\|_{\infty}\}^C$, we test whether $D[i-1][d_1][d_2]\ldots[d_C]$ is true.
 If so, we also set $D[i][d_1+w_1][d_2+w_2]\ldots[d_C+w_C]$ to true (assuming $d_j + w_j \leq \|b\|_{\infty}$ for all $j \in [V]$, otherwise we just continue with the next $d$).
 
 It is easy too see that the algorithm is correct.
 For running time, we do $V+1$ iterations in the outermost loop over values of $i$, $\|b\|_{\infty}+1$ iterations in the second loop over values of $v_i$ and $(\|b\|_{\infty}+1)^C$ iterations in the innermost loop over values of $d$. In each iteration of the innermost loop we use $O(C)$ many steps to update a single DP table entry.
 Hence, the running time is $O(V \cdot C \cdot (\|b\|_{\infty}+1)^{C+1})$.
\end{proof}

Now, we are ready to give the proof of Lemma \ref{la:dp-table}.

\begin{proof}[Proof of Lemma \ref{la:dp-table}]
 We calculate the entries of the DP table $D$ bottom-up starting with the leaf nodes.
 Before executing the main algorithm, we initialize all entries of the table to false.
 
 \paragraph{Leaf:}
 Suppose $t \in V(T)$ is a leaf node.
 Since $(T,\beta)$ is a nice tree decomposition, it holds that $X_t = \emptyset$.
 Thus, the only possible values for $\chi$ and $\xi$ are the empty functions $\varepsilon_1, \varepsilon_2$, and $\rho$ only has the empty set in its domain.
 For each $\rho\colon \{\emptyset\}\to[q]_0$, we define
 \begin{align*}
  D[t][\varepsilon_1][\varepsilon_2][\rho] = \begin{cases} \top & \text{if $\rho(\emptyset) = q$}\\
                                                           \bot & \text{otherwise}
                                             \end{cases}.
 \end{align*}
 The correctness of this base cases is obvious.
 For the running time, note that, for a fixed leaf node $t \in V(T)$, we only need to compute $q+1$ entries of the DP table which can be done in time $n^{O(1)}$.
 
 \paragraph{Introduce:}
 Suppose $t \in V(T)$ is an introduce node with the unique child $t' \in V(T)$.
 Hence, there is some $v \notin X_{t'}$ such that  $X_t = X_{t'} \cup \{v\}$.

 We iterate over all triples $(\chi',\xi',\rho')$ such that $D[t'][\chi'][\xi'][\rho']$ is true.
 For each possible extension $\chi$ of $\chi'$ (i.e., for all $\chi\colon X_t \to [q]$ such that $\chi|_{X_{t'}} = \chi'$, or equivalently all colorings of $v$), we do the following.
 First, we define an extension $\xi$ of $\xi'$.
 If $\chi(v)$ was not a free color in $\chi'$, i.e., there is some $v' \in X_{t'}$ such that $\chi(v') = \chi(v)$, we define
 \begin{align*}
  \xi(u) := \begin{cases}
             \xi'(u)  &\text{if } u \neq v \\
             \xi'(v') &\text{if } u=v
            \end{cases}
 \end{align*}
 for all $u \in X_t$.
 Note that this is well-defined since $\xi'(v') = \xi(v'')$ for all $v',v'' \in X_{t'}$ such that $\chi(v') = \chi(v'')$ by Condition \ref{item:dp-xi}.
 
 In the other case $\chi(v)$ is a free color, and we also iterate over all possible choices of the equivalence class $A_v \subseteq X_{t'}$ such that $\rho'(A_v) > 0$ and define
 \begin{align*}
   \xi(u) := \begin{cases}
              \xi'(u)   &\text{if } u \neq v \\
              A_v       &\text{if } u=v
             \end{cases}
 \end{align*}
 for all $u \in X_t$.
 Finally, we define $\rho$ by setting
 \begin{align*}
  \rho(A) \coloneqq \begin{cases}
                     0          &\text{if } v \in A \\
                     \rho'(A)-1 &\text{if $\chi(v)$ is a free color in $\chi'$ and $A = A_v$} \\
                     \rho'(A)   &\text{otherwise}
                    \end{cases}
 \end{align*}
 for all $A \in 2^{X_{t}}$. 
 Now, for each such $(\chi,\xi,\rho)$ such that $(t, \chi, \xi, \rho)$ is locally valid, we set $D[t][\chi][\xi][\rho]$ to true.
 
 \begin{correctness}
  We need to show that we correctly compute the all entries of $D[t]$ according to Definition \ref{def:dp-table}. 
  
  First, suppose the algorithm sets $D[t][\chi][\xi][\rho]$ to true in iteration $(\chi',\xi',\rho')$.
  In particular, this means that $D[t'][\chi'][\xi'][\rho']$ is true.
  Hence, by induction, there exists a square $q$-coloring $\chi_{V_{t'}}$ of $G[V_{t'}]$ satisfying the conditions of Definition~\ref{def:dp-table}.
  
  We first argue that we may assume without loss of generality that
  \begin{equation}
   \label{eq:free-color-introduce}
   A_{t'}(\chi_{V_{t'}},\chi(v)) = \xi(v).
  \end{equation}
  Indeed, if there exists some $v' \in X_{t'}$ such that $\chi(v) = \chi(v') = \chi'(v')$ then $A_{t'}(\chi_{V_{t'}},\chi(v)) = A_{t'}(\chi_{V_{t'}},\chi'(v')) = \xi'(v') = \xi(v)$
  by Condition \ref{item:dp-xi} and the definition of $\xi$.
  Otherwise, $\xi(v) = A_v$ for some $A_v \subseteq X_{t'}$ such that $\rho'(A_v) > 0$.
  Since $\rho'(A_v) > 0$, there is some color $c \in [q] \setminus \{\chi'(u) \mid u \in X_{t'}\}$ such that $A_v = A_{t'}(\chi_{V_{t'}},c)$.
  By renaming colors, we may assume without loss of generality that $c = \chi(v)$ which implies that $A_{t'}(\chi_{V_{t'}},\chi(v)) = A_v = \xi(v)$.
  
  Now, we extend $\chi_{V_{t'}}$ to a coloring $\chi_{V_{t}}$ of $V_t$ by setting
  \begin{align*}
   \chi_{V_t}(u) = \begin{cases}
                    \chi(u)          &\text{if } u = v \\
                    \chi_{V_{t'}}(u) &\text{otherwise }
                   \end{cases}.
  \end{align*}
  We first argue that $\chi_{V_t}$ is a square coloring of $G[V_t]$.
  Recall that we only set $D[t][\chi][\xi][\rho]$ to true if $(t,\chi,\xi,\rho)$ is locally valid, i.e., if Conditions \ref{item:locally-invalid-1} - \ref{item:locally-invalid-3} are satisfied.
  Since $\chi_{V_{t'}}$ is a square coloring of $G[V_{t'}]$, any potential color conflict has to involve the vertex $v$ introduced at $t$, i.e., either there is a vertex $u \in V_{t'}$ such that $\chi_{V_t}(u) = \chi_{V_t}(v)$ and $\dist_{G[V_t]}(u,v) \leq 2$, or there are $u,w \in N_{G[V_t]}(v)$ such that $\chi_{V_t}(u) = \chi_{V_t}(w)$.
  Note that by the definition of nice tree decompositions $v$ is not adjacent to any vertices in $V_t \setminus X_t$.
  Hence there are only four cases we need to consider:
  \begin{enumerate}
   \item There is some $u \in X_t \cap N(v)$ such that $\chi(u) = \chi_{V_t}(u) = \chi_{V_t}(v) = \chi(v)$. However, this is prevented by \ref{item:locally-invalid-1}.
   \item There are two vertices $u,w \in X_t \cap N(v)$ such that $\chi(u) = \chi_{V_t}(u) = \chi_{V_t}(w) = \chi(w)$. However, this is prevented by \ref{item:locally-invalid-2}.
   \item There are vertices $w \in X_t \cap N(v)$ and $v \neq u \in X_t \cap N(w)$ such that $\chi(u) = \chi_{V_t}(u) = \chi_{V_t}(v) = \chi(v)$.
    However, this is also prevented by \ref{item:locally-invalid-2}.
   \item There are vertices $w \in X_t \cap N(v)$ and $u \in N(w) \cap (V_t \setminus X_t)$ such that $\chi_{V_{t'}}(u) = \chi_{V_t}(u) = \chi_{V_t}(v) = \chi(v)$.
    This means that $w \in A_{t'}(\chi_{V_{t'}},\chi(v)) = \xi(v)$ using Equation \eqref{eq:free-color-introduce}.
    However, this is prevented by \ref{item:locally-invalid-3}.
  \end{enumerate}

  Next, we show that $\chi_{V_t}$ satisfies Conditions \ref{item:dp-chi} - \ref{item:dp-rho}.
  Clearly, \ref{item:dp-chi} is satisfied since $\chi|_{X_{t'}} = \chi'$.
  For \ref{item:dp-xi}, observe again that $v$ is not adjacent to any vertices in $V_t \setminus X_t$.
  This implies that $\xi(u) = \xi'(u) = A_{t'}(\chi_{V_{t}'},\chi'(u)) = A_t(\chi_{V_t},\chi(u))$ for all $u \in X_{t'}$.
  Also, $\xi(v) = A_{t'}(\chi_{V_{t}'},\chi(v)) = A_t(\chi_{V_t},\chi(u))$ by Equation \eqref{eq:free-color-introduce}.
  
  To show \ref{item:dp-rho}, let $A \in 2^{X_t}$.
  If $v \in A$ then $A \neq A_t(\chi_{V_t},c)$ for all $c \in [q]$ since $v$ is not adjacent to any vertices in $V_t \setminus X_t$.
  So $\rho(A) = 0 = |\{c \in [q] \mid A_t(\chi_{V_t},c) = A\} \setminus \{\chi(v) \mid v \in X_t\}|$ as desired.
  Otherwise, $A \in 2^{X_{t'}}$ and, similar as above, we have that $\{c \in [q] \mid A_t(\chi_{V_t},c) = A\} = \{c \in [q] \mid A_{t'}(\chi_{V_{t'}},c) = A\}$.
  Since $A_{t'}(\chi_{V_{t'}},\chi(v)) = \xi(v) = A_v$ in the case where $\chi(v)$ is a free color, Condition \ref{item:dp-rho} follows.
  
  \medskip
  
  Conversely, suppose there exists a square $q$-coloring $\chi_{V_t}$ of $G[V_t]$ satisfying \ref{item:dp-chi} - \ref{item:dp-rho}.
  We show that the algorithm sets $D[t][\chi][\xi][\rho]$ to true.
  Certainly, we can restrict $\chi_{V_t}$ to a square coloring $\chi_{V_{t'}} := \chi_{V_t}|_{V_{t'}}$ of $G[V_{t'}]$, and indeed it is easy to see (analogous to the arguments for \ref{item:dp-chi} - \ref{item:dp-rho} above) that this coloring satisfies the conditions of Definition~\ref{def:dp-table} for the corresponding tuple $(t',\chi', \xi',\rho')$.
  Hence, by induction, the algorithm sets $D[t'][\chi'][\xi'][\rho']$ to true.
  So to prove that the algorithm sets $D[t][\chi][\xi][\rho]$ to true, it suffices to show that $(t,\chi,\xi,\rho)$ is locally valid.
  
  However, this is easy to see via the following argument.
  Invalidity conditions \ref{item:locally-invalid-1} and \ref{item:locally-invalid-2} cannot be fulfilled, because $\chi_{V_t}$ is a valid square coloring and $\chi = \chi_{V_t}|_{X_t}$.
  Also, \ref{item:locally-invalid-3} is also not fulfilled by Condition \ref{item:dp-xi}.
  So $(t,\chi,\xi,\rho)$ is locally valid which concludes the proof.
 \end{correctness}
 
 \begin{runningtime}
  Let us analyze the running time of computing all DP table entries for $t$.
  For each possible choice of $(\chi',\xi',\rho')$, we iterate over $q \leq n$ choices of the color of $v$ and at most $2^{\ttw+1}$ many possible choices for the set $A_v$ (in the case that $\chi(v)$ is a free color).
  The time for computing $\chi$, $\xi$ and $\rho$, and then indexing the correct $D[t][\chi][\xi][\rho]$ can be naively bounded by $2^{\ttw}\cdot n^{O(1)}$.
  
  So overall, the running time can be bounded by
  \[|D[t']| \cdot 2^{2\ttw} \cdot n^{O(1)} = (q+1)^{\ttw+1 + (\ttw+1)^2 + 2^{\ttw+1}} \cdot 2^{2\ttw} \cdot n^{O(1)} = (q+1)^{2^{\ttw+2}} \cdot n^{O(1)}\]
  using Equation \eqref{eq:size-dp-table}.
 \end{runningtime}
 
 \paragraph{Forget:}
 Next, suppose $t \in V(T)$ is a forget node with unique child $t' \in V(T)$.
 So there is some $v \in X_{t}$ such that  $X_t = X_{t'} \setminus \{v\}$.
 We iterate over all triples $(\chi',\xi',\rho')$ such that $D[t'][\chi'][\xi'][\rho']$ is true.
 
 We define $\chi := \chi'|_{X_t}$, and also define $\xi$ via
 \[\xi(u) := \begin{cases}
              \xi'(u) \setminus \{v\}                      &\text{if } \chi'(u) \neq \chi'(v)\\
              \xi'(u) \setminus \{v\} \cup (N(v) \cap X_t) &\text{if } \chi'(u) = \chi'(v)
             \end{cases}\]
 for all $u \in X_t$. 
 Finally, we define $\rho: 2^{X_{t}} \to [q]_0$ via
 \[\rho(A) := \begin{cases}
               \rho'(A) + \rho'(A\cup \{v\}) + 1 & \text{if $\chi'(v) \notin \chi'(X_t)$ and $\xi(v) = A$}\\
               \rho'(A) + \rho'(A\cup \{v\})     & \text{otherwise }
              \end{cases}\]
 for all $A \in 2^{X_{t}}$. 
 We set $D[t][\chi][\xi][\rho]$ to true.
 
 \begin{correctness}
  First suppose the algorithm sets $D[t][\chi][\xi][\rho]$ to true in iteration $(\chi',\xi',\rho')$.
  In particular, this means that $D[t'][\chi'][\xi'][\rho']$ is true.
  Hence, by induction, there exists a square $q$-coloring $\chi_{V_{t'}}$ of $G[V_{t'}]$ satisfying the conditions of Definition~\ref{def:dp-table} (with respect to $(t',\chi',\xi',\rho')$).
  
  We show that choosing $\chi_{V_t} \coloneqq \chi_{V_{t'}}$ as a coloring for $V_t$ suffices, i.e., that it also satisfies the conditions of Definition~\ref{def:dp-table} with respect to $(t,\chi,\xi,\rho)$.
  Since $V_t = V_{t'}$ it immediately follows that $\chi_{V_t}$ is a square coloring of $G[V_t]$.
  
  Condition \ref{item:dp-chi} is clearly satisfied since $\chi = \chi'|_{X_t}$.
  So consider \ref{item:dp-xi} and let $u \in X_t$.
  If $\chi'(u) \neq \chi'(v)$ then
  \[\xi(u) = \xi'(u) \setminus \{v\} = A_{t'}(\chi_{V_{t'}},\chi'(u)) \setminus \{v\} = A_{t'}(\chi_{V_t},\chi(u)) \setminus \{v\} = A_t(\chi_{V_t},\chi(u))\]
  as desired.
  Otherwise, $\chi'(u) = \chi'(v)$ and
  \begin{align*}
   \xi(u) &= \xi'(u) \setminus \{v\} \cup (N(v) \cap X_t)\\
          &= A_{t'}(\chi_{V_{t'}},\chi'(u)) \setminus \{v\} \cup (N(v) \cap X_t)\\
          &= A_{t'}(\chi_{V_t},\chi(u)) \setminus \{v\} \cup (N(v) \cap X_t)\\
          &= A_t(\chi_{V_t},\chi(u)).
  \end{align*}
  Finally, to see that \ref{item:dp-rho} is satisfied, consider that for any $A \in 2^{X_t}$, the equivalence classes $A$ and $A\cup \{v\}$ are conflated.
  The off-by-one in the case $\chi'(v) \notin \chi(X_t) \land \xi(v) = A$ is due to the color $\chi'(v)$ becoming a free color again (and in compliance with Definition~\ref{def:dp-table} we only count free colors in $\rho$).
  
  \medskip
  
  Conversely, suppose that there exists a square $q$-coloring $\chi_{V_t}$ of $G[V_t]$ satisfying the conditions in Definition~\ref{def:dp-table}.
  We show that the algorithm sets $D[t][\chi][\xi][\rho]$ to true.
  Certainly, $\chi_{V_{t'}} = \chi_{V_t}$ is a valid square coloring of $G[V_{t'}]$, and, since $\xi$ and $\rho$ are correctly updated (as argued in the last paragraph), it satisfies the conditions of Definition~\ref{def:dp-table}.
  Hence, by induction, $D[t'][\chi'][\xi'][\rho']$ is set to true.
  So the algorithm sets $D[t][\chi][\xi][\rho]$ to true.
 \end{correctness}
 
 \begin{runningtime}
  Let us again analyze the running time of computing all DP table entries for $t$.
  For each possible choice of $(\chi',\xi',\rho')$, we calculate $(\chi,\xi,\rho)$ in time $2^{\ttw}n^{O(1)}$ and set the corresponding entry $D[t][\chi][\xi][\rho]$ to true.
  So overall, we get a running time that is bounded
  \[|D[t']| \cdot 2^{\ttw} \cdot n^{O(1)} = (q+1)^{\ttw+1 + (\ttw+1)^2 + 2^{\ttw+1}} \cdot 2^{\ttw} \cdot n^{O(1)} = (q+1)^{2^{\ttw+2}} \cdot n^{O(1)}\]
  using Equation \eqref{eq:size-dp-table}.
 \end{runningtime}
 
 \paragraph{Join:}
 Finally, assume $t \in V(T)$ is a join node, i.e., $t$ has exactly two children $t',t'' \in V(T)$ and $X_t = X_{t'} = X_{t''}$.
 Consider an entry $D[t][\chi][\xi][\rho]$ of the DP table.
 To determine whether $D[t][\chi][\xi][\rho]$ is true we iterate over all triples $(\chi',\xi',\rho')$ and $(\chi'',\xi'',\rho'')$ such that $D[t'][\chi'][\xi'][\rho']$ and $D[t''][\chi''][\xi''][\rho'']$ are both true.
 We set $D[t][\chi][\xi][\rho]$ to true if
 \begin{enumerate}[label = (\roman*)]
  \item\label{item:join-chi} $\chi = \chi' = \chi''$,
  \item\label{item:join-xi-conflict} $\xi'(u) \cap \xi''(u) = \emptyset$ for all $u \in X_t$,
  \item\label{item:join-xi} $\xi(u) = \xi'(u) \cup \xi''(u)$ for all $u \in X_t$, and
  \item\label{item:join-rho} there exists a function $\eta: 2^{X_{t'}} \times 2^{X_{t''}} \to [q]_0$ such that
   \begin{enumerate}[label=(J.\arabic*)]
    \item\label{item:join-eta-1} \[\forall A' \in 2^{X_{t'}}, A'' \in 2^{X_{t''}}\colon A' \cap A'' \neq \emptyset \implies \eta(A', A'') = 0,\]
    \item\label{item:join-eta-2} \begin{align*}
                                  \forall A' \in 2^{X_{t'}}\colon   &\sum_{A'' \in 2^{X_{t''}}} \eta(A',A'') = \rho'(A') \text{\ \ \ and\ \ \ }\\
                                  \forall A'' \in 2^{X_{t''}}\colon &\sum_{A' \in 2^{X_{t'}}} \eta(A',S'') = \rho''(A''),
                                 \end{align*}
    \item\label{item:join-eta-3} \[\forall A \in 2^{X_t}: \rho(A) = \sum_{\substack{A',A'' \in 2^{X_t}\\A'\cup A'' = A}} \eta(A',A'').\]
   \end{enumerate}
 \end{enumerate}
 
 Let us provide some intuition on these conditions.
 Condition \ref{item:join-chi} is self-explanatory.
 Condition \ref{item:join-xi-conflict} ensures that there are no conflicts when merging partial square-colorings for the two subcones $V_{t'}$ and $V_{t''}$ for non-free colors (i.e., colors from $\{\chi(v) \mid v \in X_t\}$).
 Indeed, if $v \in \xi'(u) \cap \xi''(u)$, this means that $v$ is adjacent to some $v' \in V_{t'} \setminus X_t$ of color $\chi(u)$, and $v$ is also adjacent to some $v'' \in V_{t''} \setminus X_t$ of the same color $\chi(u)$.
 Condition \ref{item:join-xi} ensures that $\xi$ is updated correctly: if the color of $u$ is cone-adjacent to the vertices in $\xi'(u)$ in the first subcone, and to the vertices in $\xi''(u)$ in the second subcone, it is cone-adjacent exactly to the vertices in $\xi'(u) \cup \xi''(u)$ in the entire cone $V_t$.
 Finally, for Condition \ref{item:join-rho}, we need to check that $\rho$ is correctly obtained from $\rho'$ and $\rho''$ which is achieved by an auxiliary mapping $\eta$ describing the number of shared free colors between two equivalence classes from the two different subcones.
 More precisely, for two equivalence classes $A' \in 2^{X_{t'}}, A'' \in 2^{X_{t''}}$, we interpret the value $\eta(A',A'')$ as the number of free colors that occur in both the equivalence classes identified with $A'$ and $A''$.
 
 With this interpretation in mind, \ref{item:join-eta-1} checks that no color is shared between equivalence classes that have a non-empty shared neighborhood, i.e., no color used in one of the subcones is within distance two of that same color in the other subcone (this is the counterpart of Condition \ref{item:join-xi-conflict} for free colors).
 Condition \ref{item:join-eta-2} checks that all numbers add up in the correct way.
 And finally, knowing the number shared colors between two equivalence classes, we can deduce $\rho$ using \ref{item:join-eta-3}.
 Indeed, any fixed free color $c$ occurs in some equivalence class $A'$ in the first subcone and an equivalence class $A''$ in the second subcone (note that the free colors in both subcones are identical since $\chi' = \chi''$).
 Hence, $c$ is a shared color of $A'$ and $A''$ and is counted in (and only in) $\eta(A',A'')$.
 The color $c$ is in the equivalence class corresponding to $A' \cup A'' = A$ in the entire subcone $V_t$.
 
 Since $\chi,\xi,\rho,\chi',\xi',\rho',\chi'',\xi'',\rho''$ are fixed in each iteration, Conditions \ref{item:join-chi} and \ref{item:join-xi} can trivially be checked in polynomial time.
 For Condition \ref{item:join-rho}, we view Conditions \ref{item:join-eta-1} - \ref{item:join-eta-3} as an integer linear program (ILP) with unknowns $\eta(A', A'')$, $A \in 2^{X_{t'}}$, $A'' \in 2^{X_{t''}}$.
 Now, we can check the feasibility of this ILP using Lemma \ref{la:ilp-solver} and set $D[t][\chi][\xi][\rho]$ accordingly.
 
 \begin{correctness}
  First suppose the algorithm sets $D[t][\chi][\xi][\rho]$ to true via triples $(\chi',\xi',\rho')$ and $(\chi'',\xi'',\rho'')$ and the mapping $\eta$ satisfies \ref{item:join-eta-1} - \ref{item:join-eta-3}.
  Since $D[t'][\chi'][\xi'][\rho']$ and $D[t''][\chi''][\xi''][\rho'']$ are true, we know by induction that there exist square $q$-colorings $\chi_{V_{t'}}$ of $G[V_{t'}]$ and $\chi_{V_{t''}}$ of $G[V_{t''}]$ which satisfy the conditions in Definition~\ref{def:dp-table}.
  We show that if the algorithm sets $D[t][\chi][\xi][\rho]$ to true, then we can combine $\chi_{V_{t'}}$ and $\chi_{V_{t''}}$ into a square coloring $\chi_{V_t}$ for $G[V_t]$ using the function $\eta$.
  Towards this end, we start by showing the following claim.
  
  \begin{claim}
   \label{claim:join-match-colors}
   There is a square $q$-colorings $\lambda_{V_{t''}}$ of $G[V_{t''}]$ such that
   \begin{enumerate}
    \item $\lambda_{V_{t''}}$ satisfies \ref{item:dp-chi} - \ref{item:dp-rho} with respect to $(t'',\chi'',\xi'',\rho'')$, and
    \item for every $A' \in 2^{X_{t'}}$ and $A'' \in 2^{X_{t''}}$ it holds that
    \[\eta(A',A'') = |\{c \in [q] \mid A_{t'}(\chi_{V_{t'}},c) = A' \wedge A_{t''}(\lambda_{V_{t''}},c) = A''\} \setminus \{\chi(v) \mid v \in X_t\}|.\]
   \end{enumerate}
  \end{claim}
  \begin{claimproof}
   We obtain $\lambda_{V_{t''}}$ from $\chi_{V_{t''}}$ by renaming colors.
   More precisely, we define a bijection $\pi \colon [q] \to [q]$ and define $\lambda_{V_{t''}}(u) \coloneqq \pi(\chi_{V_{t''}}(u))$ for all $u \in V_{t''}$.
   
   First, we choose the identity for all colors in the bag $X_t$, i.e., $\pi(c) = c$ for all $c \in \{\chi''(v) \mid v \in X_t\}$.
  
   In the second step, we use $\eta$ to permute free colors.
   For each pair of equivalence classes $A' \in 2^{X_{t'}}, A'' \in 2^{X_{t''}}$, we arbitrarily match $\eta(A',A'')$ free colors from the equivalence class $A'$ to free colors from the equivalence class $A''$, i.e., we define $\pi$ in such a way that
   \[\eta(A',A'') = |\{c \in [q] \mid A_{t'}(\chi_{V_{t'}},c) = A' \wedge A_{t''}(\chi_{V_{t'}},\pi^{-1}(c)) = A''\} \setminus \{\chi(v) \mid v \in X_t\}|\]
   Note that this is always possible, and indeed that it achieves a full matching, by condition \ref{item:join-eta-2}.
   Also, $\lambda_{V_{t''}}$ still satisfies \ref{item:dp-chi} - \ref{item:dp-rho} with respect to $(t'',\chi'',\xi'',\rho'')$ since only free colors are renamed.
  \end{claimproof}
  
  Using the last claim, we may assume without loss of generality that $\chi_{V_{t''}} = \lambda_{V_{t''}}$.
  We define a coloring $\chi_{V_t}$ of $V_t$ via 
  \[\chi_{V_t}(v) = \begin{cases}
                     \chi_{V_{t'}}(v)  & \text{if } v \in V_{t'} \\
                     \chi_{V_{t''}}(v) & \text{if } v \in V_{t''}
                    \end{cases}.\]
  First observe that $\chi_{V_t}$ is well-defined by Condition \ref{item:join-chi}.
  
  We show that $\chi_{V_t}$ is a square coloring of $G[V_t]$.
  Let $u,w \in V_t$ be distinct vertices such that $\dist_{G[V_t]}(u,w) \leq 2$.
  If $uw \in E(G)$ then $u,w \in V_{t'}$ or $u,w \in V_{t''}$, and the respective square coloring $\chi_{V_{t'}}$ or $\chi_{v_{t''}}$ prohibits that $u$ and $v$ have the same color.
  
  Otherwise, there is some $v \in V_t$ such that~$uv,vw \in E(G)$.
  As before, if $u,v,w \in V_{t'}$ or $u,v,w \in V_{t''}$, the respective square coloring $\chi_{V_{t'}}$ or $\chi_{v_{t''}}$ prohibits that $u$ and $v$ have the same color.
  Since $u,v$ and $v,w$ are adjacent, each pair of vertices must be in the same subcone.
  So without loss of generality the only other case we need to consider is $u \in V_{t'}\setminus X_t$, $v \in X_t$ and $w \in V_{t''}\setminus X_t$.
  Suppose without loss of generality that $\chi_{V_t}(u) = \chi_{V_t}(w)$.
  First suppose there is some $u' \in X_t$ such that $\chi_{V_t}(u) = \chi(u')$.
  Then $v \in \xi'(u)$ and $v \in \xi''(w)$ which contradicts Condition \ref{item:join-xi-conflict}.
  
  Otherwise, let $A_u \coloneqq A_t(\chi_{V_t},\chi(u))$ and $A_w \coloneqq A_t(\chi_{V_t},\chi(w))$.
  We have $v \in A_u$ and $v \in A_w$ by definition and hence, $A_u \cap A_w \neq \emptyset$.
  So $\eta(A_u,A_w) = 0$ by condition \ref{item:join-eta-1}.
  This means that
  \[\{c \in [q] \mid A_{t'}(\chi_{V_{t'}},c) = A_u \wedge A_{t''}(\chi_{V_{t''}},c) = A_w\} \setminus \{\chi(v) \mid v \in X_t\} = \emptyset\]
  by Claim \ref{claim:join-match-colors}.
  But $c \coloneqq \chi_{V_t}(u) = \chi_{V_t}(w)$ is contained in the first set, but not the second one.
  So this gives again a contradiction.
  Overall, we get that $\chi_{V_t}$ is a square coloring of $G[V_t]$.
  
  It remains to show that $\chi_{V_t}$ satisfies the conditions in Definition~\ref{def:dp-table}.
  Condition \ref{item:dp-chi} immediately follows from \ref{item:join-chi}.
  
  For $u \in X_t$ we have that
  \begin{align*}
   \xi(u) &= \xi'(u) \cup \xi''(u)\\
          &= A_{t'}(\chi_{V_{t'}},\chi'(u)) \cup A_{t''}(\chi_{V_{t''}},\chi''(u))\\
          &= A_{t'}(\chi_{V_{t'}},\chi(u)) \cup A_{t''}(\chi_{V_{t''}},\chi(u))\\
          &= \Big\{v \in X_t \;\Big|\;         \Big(\exists w \in V_{t'} \setminus X_{t'} \colon vw \in E(G) \wedge \chi_{V_{t'}}(w) = \chi(u)\Big)\\
          &\quad\quad\quad\quad\quad\quad \vee \Big(\exists w \in V_{t''} \setminus X_{t''} \colon vw \in E(G) \wedge \chi_{V_{t''}}(w) = \chi(u)\Big)\Big\}\\
          &= A_{t}(\chi_{V_{t}},\chi(u))
  \end{align*}
  which implies Condition \ref{item:dp-xi}.
  
  For Condition \ref{item:dp-rho}, we first observe that
  \[A_t(\chi_{V_t},c) = A_{t'}(\chi_{V_{t'}},c) \cup A_{t''}(\chi_{V_{t''}},c)\]
  for all colors $c \in [q]$ using the same arguments as above.
  Hence, Condition \ref{item:join-eta-3} and Claim \ref{claim:join-match-colors} imply that
  \[\rho(A) = \sum_{\substack{A',A'' \in 2^{X_t}\\A'\cup A'' = A}} \eta(A',A'') = |\{c \in [q] \mid A_t(\chi_{V_t},c) = A\} \setminus \{\chi(v) \mid v \in X_t\}|\]
  for all $A \in 2^{X_t}$ as desired.
  
  \medskip
  
  Conversely, suppose there exists a square $q$-coloring $\chi_{V_t}$ of $G[V_t]$ satisfying the conditions in Definition~\ref{def:dp-table}.
  We show that the algorithm sets $D[t][\chi][\xi][\rho]$ to true.
  
  To do this, we consider $\chi_{V_t}|_{V_{t'}}$ and $\chi_{V_t}|_{V_{t''}}$.
  Certainly, these colorings are square colorings of $G[V_{t'}]$ and $G[V_{t''}]$, respectively.
  They immediately induce triples $(\chi',\xi',\rho')$ and $(\chi'',\xi'',\rho'')$ by choosing these mappings in such a way that $\chi_{V_t}|_{V_{t'}}$ witnesses that $D[t',\chi',\xi',\rho']$ is true.
  The colorings $\chi_{V_t}|_{V_{t'}}$ and $\chi_{V_t}|_{V_{t''}}$ also define $\eta$ by setting
  \[\eta(A',A'') \coloneqq |\{c \in [q] \mid A_{t'}(\chi_{V_t},c) = A' \wedge A_{t''}(\chi_{V_t},c) = A''\} \setminus \{\chi(v) \mid v \in X_t\}|.\]
  Now, it is easy to check that Conditions \ref{item:join-chi} - \ref{item:join-rho} are satisfied.
  So the algorithm sets $D[t][\chi][\xi][\rho]$ to true.
 \end{correctness}
 
 \begin{runningtime}
  We analyze the running time of computing all DP table entries for $t$.
  
  Iterating over all choices of $\chi, \xi, \rho, \chi', \xi', \rho', \chi'', \xi'', \rho''$ takes
  \[|D[t]| \cdot |D[t']| \cdot |D[t'']| \leq (q+1)^{3\cdot(\ttw+1 + (\ttw+1)^2 + 2^{\ttw+1})}\]
  many iterations by Equation \eqref{eq:size-dp-table}. 
  For each iteration, we only require polynomial time to check Conditions \ref{item:join-chi} - \ref{item:join-xi}.
  For the ILP used to check Condition \ref{item:join-rho}, the number of unknowns is at most $2^{2(\ttw+1)}$.
  Note that \ref{item:join-eta-1} just says that certain unknowns are zero.
  We can hence ignore them and delete them from all equations.
  The number of remaining linear equations from \ref{item:join-eta-2} and \ref{item:join-eta-3} is at most $(2^{\ttw + 1} + 2^{\ttw + 1}) + 2^{\ttw + 1} = 3\cdot 2^{\ttw+1}$.
  The unknowns are constrained by $0 \leq \eta(S',S'') \leq q$. 
  Hence, the ILP can be solved in time in time $O(2^{2\cdot(\ttw+1)} \cdot 2^{\ttw + 1} \cdot (q+1)^{3\cdot 2^{\ttw+1}+1})$ by Lemma \ref{la:ilp-solver}.
  
  So overall, the running time can be bounded by
  \begin{align*}
       &(q+1)^{3\cdot(\ttw+1 + (\ttw+1)^2 + 2^{\ttw+1})} \cdot 2^{3\cdot(\ttw+1)} \cdot (q+1)^{3\cdot 2^{\ttw+1}+1} \cdot n^{O(1)}\\
   =\; &(q+1)^{6\cdot(\ttw+1) + 3\cdot(\ttw+1)^2 + 6\cdot2^{\ttw+1} + 1}\cdot n^{O(1)}\\
   =\; &(q+1)^{2^{\ttw+4}}\cdot n^{O(1)}.\\
  \end{align*}
 \end{runningtime}
 We complete the proof by observing that, for every node $t$, all entries of $D[t]$ can be computed in time $(q+1)^{2^{\ttw+4}}\cdot n^{O(1)}$.
 So overall, the running time is bounded by $|V(T)| \cdot (q+1)^{2^{\ttw+4}}\cdot n^{O(1)}$.
\end{proof}

\section{Lower Bound for Graphs of Bounded Treewidth}
\label{sec:tw-lower-bound}
We have seen an algorithm solving \SqCol\ in time $(q+1)^{2^{\ttw+4}}\cdot n^{O(1)}$.
We now show that assuming the Exponential Time Hypothesis (\ETH), this running time is essentially optimal.
More specifically, we show:

\begin{theorem}[Theorem \ref{thm:tw-lb} restated]
 \label{thm:tw-lb}
 Assuming \ETH, for any $\epsilon>0$ and any function $f$, there is no $f(\ttw)n^{(2-\epsilon)^\ttw}$ time algorithm solving \SqCol\ on graphs of treewidth $\ttw$.
\end{theorem}

\subsection{Problems Used in the Reduction}

In our proof of Theorem~\ref{thm:tw-lb}, we use two other problems, which we define below. The first is the starting
point of our reduction, the second is an intermediate problem.

\subsubsection{The (Colored) Subgraph Isomorphism Problem}

The starting point for the hardness result is the well-known \textsc{Subgraph Isomorphism} problem:

\defproblem{Subgraph Isomorphism}{A pattern graph $H$ and a host graph $G$}
{Does there exist a subgraph of $G$ that is isomorphic to $H$?}

We actually use a colored variant of this, in which the vertices
of $H$ are colored using $|V(H)|$ unique colors. The graph $G$ is also
colored arbitrarily using these $|V(H)|$ colors. The problem now has the
additional constraint that the isomorphism must preserve colors. We
define a convenient version of this as follows:

\defproblem{Colored Subgraph Isomorphism}{
 A pattern graph $H$ with $k$ vertices, a host graph $G$ with $k \cdot n$ vertices,
 and a function $f: V(G) \to V(H)$ such that $|f^{-1}(h)| = n$ for all $h \in V(H)$.}{
 Is there a subgraph $G'$ of $G$ such that $f|_{V(G')}$ is an isomorphism from $G'$ to $H$?}

\subsubsection{The (Restricted) Vector $k$-Sum Problem}

We first define the basic \textsc{Vector $k$-Sum} problem, before stating a restricted version that we use in our reduction.

\defproblem{Vector $k$-Sum}{
 $k$ lists $A_1, \ldots, A_k$ of $m$-dimensional integer vectors, each of size $n$, and an $m$-dimensional target vector $t \in \ZZ^m$.}{
 Are there $a_1 \in A_1, \ldots, a_k \in A_k$ such that $\sum_{i=1}^{k} a_i = t$?}

 For a vector $a \in \ZZ^m$ and $j \in [m]$ we denote by $a[j]$ the $j$-th entry of $a$.
 
\begin{definition}\label{defn:node-representing-vector-group}
  Suppose $m,n \in \ZZp$.
  We say a list of $m$-dimensional vectors $A$ is a \emph{node-representing vector list (with parameters $m,n$)} if there are $D^+, D^- \subseteq [m]$ with $|D^+ \cup D^-| = |D^+| + |D^-| = 3$ such that
  \begin{align*}
    a[j] \in \begin{cases}
               [1,n^2]   & \text{if } j \in D^+\\
               [-n^2,-1] & \text{if } j \in D^-\\
               \{0\}     & \text{otherwise}
             \end{cases}
  \end{align*}
  for all $a \in A$ and all $j \in [m]$.
  We call the elements of $D^+$ the \emph{positive dimensions},
  the elements of $D^-$ the \emph{negative dimensions} and
  the elements of $D^+\cup D^-$ the \emph{non-zero dimensions} of $A$.
\end{definition}

\defproblemrestriction{Restricted Vector $k$-Sum}{
 Let parameters $m,k,n \in \ZZp$ be arbitrary.
 We define a \textsc{Restricted Vector $k$-Sum} (with parameters $m,k,n$) instance to be a \textsc{Vector $k$-Sum} instance that satisfies all of the following conditions:
 \begin{enumerate}[nolistsep,label = (\arabic*)]
  \item\label{item:restricted-vector-list-1} The target vector of the instance is $(0,...,0) \in \ZZ^m$.
  \item\label{item:restricted-vector-list-2} There are $k$ vector lists, each of size exactly $n^4$.
  \item\label{item:restricted-vector-list-3} All vector lists are node-representing vector lists with parameters $m,n$.
  \item\label{item:restricted-vector-list-4} Each $j\in[m]$ is a non-zero dimension in exactly two of the $n^4$ vector lists.
 \end{enumerate}
}

\subsection{The Reduction Chain}

To prove Theorem~\ref{thm:tw-lb}, we rely on the
hardness of certain instances of \textsc{Colored Subgraph Isomorphism}
under \ETH, and construct a two-step reduction from those instances
of \textsc{Colored Subgraph Isomorphism} to instances
of \SqCol.

For the former, we use a result that follows directly from \cite{DBLP:journals/toc/Marx10}.
It shows the hardness of \textsc{Colored Subgraph Isomorphism} with cubic pattern graphs parameterized by its number of edges.
Recall that a graph $H$ is \emph{cubic} if every vertex has degree exactly $3$. 

\begin{theorem}\label{thm:subiso-lower-bound}
  Let $m_0>0$ be arbitrary. The \textsc{Colored Subgraph Isomorphism} problem, restricted to cubic pattern graphs with $m \geq m_0$ edges,
  cannot be solved in time $f(m)n^{o(m/\log m)}$ for any function $f$, unless \ETH\ fails.
\end{theorem}

This follows from \cite[Corollary 6.1]{DBLP:journals/toc/Marx10} when using as a
graph family $\mathcal{G}$ a family of cubic expander graphs with at
least $m_0$ edges. It is well-known that cubic expander graphs have treewidth
$\Omega(m)$ (where $m$ denotes the number of edges) \cite{GroheM09}.

Now, our two-step reduction consists of a reduction
from \textsc{Colored Subgraph Isomorphism} with cubic pattern graphs
to \textsc{Restricted Vector $k$-Sum} instances, and then from those
instances to \SqCol\ instances.
More precisely, we prove the following theorems:

\begin{restatable}{theorem}{thmSubisoToRestrictedVectorSum}\label{thm:subiso-to-restricted-vector-sum}
  There exists a polynomial-time algorithm $\mathcal{A}$ that, given
  a \textsc{Colored Subgraph Isomorphism} instance with a $k$-vertex
  cubic pattern graph $H$ (which has $m=3k/2$ edges) and a
  $(k\cdot n)$-vertex host graph $G$, outputs an
  equivalent \textsc{Restricted Vector $k$-Sum} instance with
  parameters $m,k,n$.
\end{restatable}

\begin{restatable}{theorem}{thmVectorSumToDistanceTwoColoring}\label{thm:vector-sum-to-distance-two-coloring}
  There exists an algorithm $\mathcal{B}$ such that for every
  $\varepsilon > 0$ there exists an
  $m_{\varepsilon} \in \ZZp$ such that for all
  $m>m_{\varepsilon}$, the following holds. Let
  $k \in \ZZp$ be arbitrary. Given a \textsc{Restricted
  Vector $k$-Sum} instance with parameters $m,k,n$, the algorithm
  $\mathcal{B}$ produces an equivalent \SqCol\
  instance with treewidth at most $(1+\varepsilon)\log(m)+O(1)$ and
  with $(k + m + n)^{O(1)}$ vertices.

  Moreover, the algorithm runs in time $(k + m + n)^{O(1)}$.
\end{restatable}

Assuming both theorems hold, we can already prove
Theorem~\ref{thm:tw-lb}, i.e., assuming \ETH, \SqCol\ cannot be solved in time
$f(\ttw)n^{(2-\varepsilon)^{\ttw}}$ for any $\varepsilon > 0$.

\begin{proof}[Proof of Theorem \ref{thm:tw-lb}]
  Without loss of generality let $0 < \varepsilon \leq 1$ (otherwise
  it is trivial). Choose $0 < \widetilde{\varepsilon} \leq 1$ such
  that $1-\widetilde{\varepsilon} = \log(2-\varepsilon)$; note that
  this is always possible. Now assume that there
  exists an algorithm $\mathcal{C}$ solving \SqCol\ in
  time $f(\ttw)N^{(2-\varepsilon)^{\ttw}} =
  f(\ttw)N^{2^{(1-\widetilde{\varepsilon})\ttw}}$ on graphs of treewidth at most $\ttw$.
  Without loss of generality, we may assume that $f$ is monotonically increasing.

  First, we use Theorem~\ref{thm:vector-sum-to-distance-two-coloring}
  with $\varepsilon' = \widetilde{\varepsilon}$, obtaining an
  algorithm $\mathcal{B}$ and a constant $m_{\varepsilon'} > 0$.  Our
  goal is to design an algorithm with running time $h(m)n^{o(m/\log
  m)}$ for \textsc{Colored Subgraph Isomorphism}, where $h$ is some
  arbitrary function and where the pattern graphs are restricted to cubic graphs with $m>m_{\varepsilon'}$ edges.
  By Theorem~\ref{thm:subiso-lower-bound}, this implies that \ETH\ is false.

  Hence let such a \textsc{Colored Subgraph Isomorphism} instance be
  given. We use algorithm $\mathcal{A}$ from
  Theorem~\ref{thm:subiso-to-restricted-vector-sum} to convert this
  into an equivalent \textsc{Restricted Vector $k$-Sum} instance with
  parameters $m,k,n$. This conversion runs in time $(n + m)^{O(1)}$.

  For all $m>m_{\varepsilon'}$, the algorithm
  $\mathcal{B}$ from above converts a \textsc{Restricted Vector
  $k$-Sum} instance with parameters $m,k,n$ into
  a \SqCol\ instance with treewidth $\ttw \leq (1+\widetilde{\varepsilon})\log(m) + \alpha$, where $\alpha$ denotes a suitable absolute constant,
  and with $N = (k + m + n)^{O(1)}$ vertices.
  Note that $m = 3k/2$, so $N = (m + n)^{O(1)}$.
  This conversion runs in time $(k + m + n)^{O(1)} = (m + n)^{O(1)}$.

  Now, we use algorithm $\mathcal{C}$ on this newly constructed \SqCol\ instance.
  Let us define $g(m) := f((1+\widetilde{\varepsilon})\log(m)+\alpha)$.
  Observe that $f(\ttw) \leq g(m)$.
  Then algorithm $\mathcal{C}$ runs in time
  \[f(\ttw)N^{2^{(1-\widetilde{\varepsilon})\ttw}} = g(m)\left((m + n)^{O(1)}\right)^{2^{(1-\widetilde{\varepsilon})((1+\widetilde{\varepsilon})\log(m)+\alpha)}}.\]
  We rewrite this as
  \[g(m)(m + n)^{O(2^\xi)}\]
  where
  \[\xi := (1-\widetilde{\varepsilon})((1+\widetilde{\varepsilon})\log(m)+\alpha).\]
  Setting $\delta := \widetilde{\varepsilon}^2 > 0$, we get that $\xi = (1-\delta)\log(m)+O(1)$.
  Hence $2^\xi = O(m^{(1-\delta)})$.
  So the running time can be written as
  \[g(m)(m + n)^{O(2^\xi)} = g(m)m^{O(m^{(1-\delta)})}n^{O(m^{(1-\delta)})} = h(m)n^{o(m/\log m)}\]
  for some suitable $h(m) = g(m) \cdot m^{O(m^{(1-\delta)})}$.

  Note that we have an additional running time of $(m + n)^{O(1)}$
  for both of the two reduction steps, but this is dominated by the
  above running time. Hence we have an algorithm
  for \textsc{Colored Subgraph Isomorphism} running in total time $h(m)n^{o(m/\log m)}$ as desired.
\end{proof}

In the next two sections, we construct the reductions that prove Theorems~\ref{thm:subiso-to-restricted-vector-sum} and~\ref{thm:vector-sum-to-distance-two-coloring}.

\subsection{From Subgraph Isomorphism to Restricted Vector Sum}

First we prove Theorem~\ref{thm:subiso-to-restricted-vector-sum}, which provides the reduction from \textsc{Colored Subgraph Isomorphism} to \textsc{Restricted Vector $k$-Sum}.
We restate it here for the readers convenience.

\thmSubisoToRestrictedVectorSum*

\begin{proof}
 Suppose $(H,G,f)$ is a \textsc{Colored Subgraph Isomorphism} instance with of a $k$-vertex
 $m$-edge cubic pattern graph $H$, and a $(k \cdot n)$-vertex host graph $G$.
 We construct an equivalent \textsc{Restricted Vector $k$-Sum} instance as follows.
 
 First, we deal with simple NO-instances.
 Let $h_0 \in V(H)$ and suppose $N_H(h_0) = \{h_1, h_2, h_3\}$.
 If there are no vertices $v_i \in f^{-1}(h_i)$, $i \in \{0,1,2,3\}$, such that $v_0v_1,v_0v_2,v_0v_3 \in E(G)$ then $(H,G,f)$ is clearly a NO-instance,
 and we return a trivial NO-instance of \textsc{Restricted Vector $k$-Sum} with parameters $m,k,n$.

 Otherwise, we set the target vector to be $(0, \ldots, 0) \in \ZZ^m$.
 For each $h \in V(H)$ there is a vector list $A_h$ associated with vertex $h$.
 Let us fix an arbitrary linear order $\prec$ on $V(H)$ as well as an arbitrary bijection $b_E: E(H) \to [m]$ that associates each edge of $H$ with some dimension (all vectors have dimension $m$).
 
 Moreover, for each edge $hh' \in E(H)$, we denote by $E(G)_{hh'}$ the set of edges in $G$ that connect a vertex in $f^{-1}(h)$ to a vertex in $f^{-1}(h')$.
 We also fix an arbitrary injective mapping $b_{hh'}\colon E(G)_{hh'} \to [n^2]$ that associates a distinct non-negative integer with each edge from $E(G)_{hh'}$.
 
 Let $h_0 \in V(H)$ and suppose $N_H(h_0) = \{h_1, h_2, h_3\}$.
 For every $(v_0,v_1,v_2,v_3) \in f^{-1}(h_0) \times f^{-1}(h_1) \times f^{-1}(h_2) \times f^{-1}(h_3)$ such that $v_0v_1,v_0v_2,v_0v_3 \in E(G)$ we add a vector $a_{v_0,\{v_1,v_2,v_3\}} \in \ZZ^m$ to the list $A_{h_0}$ which is defined as
 \begin{align*}
  a_{v_0,\{v_1, v_2, v_3\}}[j] :=
  \begin{cases}
   b_{h_0h_i}(v_0v_i)  & \text{if } \exists i \in [3]\colon j = b_E(h_0h_i) \land h_0 \prec h_i\\
   -b_{h_0h_i}(v_0v_i) & \text{if } \exists i \in [3]\colon j = b_E(h_0h_i) \land h_0 \succ h_i\\
   0                   & \text{otherwise}
  \end{cases}
 \end{align*}
 for all $j \in [m]$.
 
 This almost completes the construction.
 We still have to guarantee that each vector list has exactly $n^4$ vectors.
 So far, each vector list has at most $n^4$ vectors.
 Recall that we already dealt with simple NO-instances in the beginning.
 Hence, each vector list contains at least one vector.
 So we can simply fill each vector list by duplicating any of its vectors until the list has size exactly $n^4$.
 
 \begin{correctness}
  We first argue that the vectors have the correct format, i.e., Conditions \ref{item:restricted-vector-list-1} - \ref{item:restricted-vector-list-4} are satified.
  Clearly, Conditions \ref{item:restricted-vector-list-1} and \ref{item:restricted-vector-list-2} are satisfied.

  For Condition \ref{item:restricted-vector-list-3}, fix a node $h_0 \in V(H)$ and suppose $N_H(h_0) = \{h_1, h_2, h_3\}$.
  Observe that the list $A_{h_0}$ is a node-representing vector list if it has exactly $3$ non-zero dimensions.
  However, this is immediately clear from the definition since each vector $a_{v_0,\{v_1, v_2, v_3\}} \in A_{h_0}$ has non-zero dimensions $b_E(h_0h_1)$, $b_E(h_0h_2)$ and $b_E(h_0h_3)$.
  
  For Condition \ref{item:restricted-vector-list-4}, let us fix some $j \in [m]$.
  Since $b_E$ is a bijection, there are unique $h,h' \in V(H)$ such that $hh' \in E(H)$ and $b_E(hh') = j$.
  By construction, $j$ is a non-zero dimension exactly in the vector lists $A_h$ and $A_{h'}$.
  
  \medskip
  
  Next, we prove that the constructed instance of \textsc{Restricted Vector $k$-Sum} is equivalent to the input instance.
  
  First suppose that $(H,G,f)$ is a YES-instance.
  This means there is a mapping $\sigma\colon V(H) \to V(G)$ such that $\sigma(h) \in f^{-1}(h)$ for every $h \in V(H)$, and $\sigma(h)\sigma(h)' \in E(G)$ for all $hh' \in E(H)$.
  For every $h_0 \in V(H)$ with neighbors $N_H(h_0) = \{h_1, h_2, h_3\}$, we choose from $A_{h_0}$ the vector $a_{\sigma(h_0),\{\sigma(h_1), \sigma(h_2), \sigma(h_3)\}}$.
  Note that this vector exists since $\sigma(h_0)\sigma(h_1), \sigma(h_0)\sigma(h_2), \sigma(h_0)\sigma(h_3) \in E(G)$.
  
  We need to show that all selected vectors sum up to the target vector $(0, \ldots, 0)$.
  Towards this end, fix a dimension $j \in [m]$ and let $hh' \in E(H)$ such that $b_E(hh') = j$.
  As we already observed above, $j$ is a non-zero dimension exactly in the vector lists $A_h$ and $A_{h'}$.
  Without loss of generality suppose that $h \prec h'$.
  Also let $h_1,h_2$ denote the other two neighbors of $h$, and let similarly $h_1',h_2'$ denote the other two neighbors of $h'$.
  Then $a_{\sigma(h),\{\sigma(h'),\sigma(h_1),\sigma(h_2)\}}[j] = b_{hh'}(\sigma(h)\sigma(h'))$ and
  $a_{\sigma(h'),\{\sigma(h),\sigma(h_1'),\sigma(h_2')\}}[j] = - b_{hh'}(\sigma(h)\sigma(h'))$.
  So their sum is $0$ as desired.
  Since this is true for all dimensions, it follows that the chosen vectors form a solution to the \textsc{Restricted Vector $k$-Sum} instance.

  Conversely, suppose there exists a choice of vectors that is a solution to the \textsc{Restricted Vector $k$-Sum} instance.
  We define $\sigma\colon V(H) \rightarrow V(G)$ in such a way that, for all $h \in V(H)$, we have $\sigma(h) = v$ where $a_{v,\{v_1, v_2, v_3\}}$ is the vector chosen from the list $A_h$.
  Clearly, $\sigma(h) \in f^{-1}(h)$ for every $h \in V(H)$.
  
  So it remains to show that $\sigma(h)\sigma(h)' \in E(G)$ for all $hh' \in E(H)$.
  Fix an edge $hh' \in E(H)$ and let $j \coloneqq b_E(hh')$ denote the corresponding dimension.
  Without loss of generality suppose that $h \prec h'$.
  Recall that $j$ is a non-zero dimension exactly in the vector lists $A_h$ and $A_{h'}$.
  
  Suppose $a_{v,\{v', v_1, v_2\}}$ is the vector chosen from list $A_h$ where $v' \in f^{-1}(h')$, and $a_{w,\{w', w_1, w_2\}}$ is the vector chosen from list $A_{h'}$ where $w' \in f^{-1}(h)$.
  Then $a_{v,\{v', v_1, v_2\}}[j] = b_{hh'}(vv')$ and $a_{w,\{w', w_1, w_2\}}[j] = -b_{hh'}(ww')$.
  Since both entries sum up to $0$ it follows that $b_{hh'}(vv') = b_{hh'}(ww')$ which implies that $vv' = ww'$.
  So $v = w'$ and $w = v'$.
  Since $vv' \in E(G)$ by definition of the set $A_h$ it follows that $\sigma(h)\sigma(h') = vw \in E(G)$ as desired.
 \end{correctness}
 
 \begin{runningtime}
  In the construction, for each vertex $h_0 \in V(H)$ with neighbors $h_1,h_2,h_3$, we iterate over all 4-tuples
  $(v_0, v_1, v_2, v_3) \in f^{-1}(h_0) \times f^{-1}(h_1) \times f^{-1}(h_2) \times f^{-1}(h_3)$ as well as all dimensions $j \in [m]$.
  Since we do a constant number of operations for each such choice, the construction takes time $O(kn^4m)$.
  We have $m = 3k/2$, so this can be rewritten as $O(m^2n^4)$.
  This is clearly the dominating term in the running time.
 \end{runningtime}
\end{proof}

\subsection{From Restricted Vector Sum to Square Coloring}

We now show Theorem~\ref{thm:vector-sum-to-distance-two-coloring},
which provides a way to convert a \textsc{Restricted Vector $k$-Sum}
instance to an equivalent \SqCol\ instance of low
treewidth. We restate it here for the readers convenience.

\thmVectorSumToDistanceTwoColoring*

Our proof of this theorem is very long and highly technical. It spans the rest of this section.

Let $\varepsilon > 0$ be given.
We choose $m_{\varepsilon}$ in such a way that
\begin{equation}
 \label{eq:m-epsilon-def}
 m^{\varepsilon} > 2(1+\varepsilon)\log m + 6
\end{equation}
for all $m > m_{\varepsilon}$ which is certainly possible.
We explain this choice later.
Continuing, let $m>m_{\varepsilon}$ also be given.
Finally, let a \textsc{Restricted Vector $k$-Sum} instance with parameters $m,k,n$ be given, for $k,n \in \ZZp$.
Let $A_1, \ldots, A_k$ denote the lists of vectors of the instance.

\subsubsection{Preparation}

Before we describe the reduction algorithm, we note a few key definitions and details used in the reduction.

\paragraph{Colorless and forced-color vertices.}

We use colorless vertices in our construction.
These are normal vertices that are not assigned a color.
Note that neighboring vertices of colorless vertices must still have pairwise distinct colors.
At the end of the reduction, we describe how to get rid of these colorless vertices again.

Similarly, we use vertices with a predetermined color, e.g., when we say that we create a ``\red\ vertex''.
These are normal vertices that we are connected in such a way that they are forced to assume the color we assign (\red\, in this case).
Of course, \red\ is not a fixed color, but is identified as being the color that is used to color a certain vertex somewhere in the graph.

\paragraph{Colors used in reduction.}

We now define our target number of colors $q$.
That is, the result of the reduction is a YES-instance if and only if it is colorable using $q$ colors.

The colors we use are divided into three categories:

\begin{description}
 \item[Counting colors:] There are a total of $2m \cdot \targetval$ so-called \emph{counting colors}, which are grouped into color classes.
  There are $2m$ of these color classes, one pair of color classes for each dimension of the vectors from the \textsc{Restricted Vector $k$-Sum} instance.
  We number the color classes $1,\dots,2m$ and say that dimension $i \in [m]$ corresponds to color classes $2i-1$ and $2i$.
  Each color class has $\targetval$ colors.
\item[Logic colors:] We have an additional three colors which we name \red, \green\ and \blue.
 They are used for local logic computations in some of the gadgets.
\item[Neutral colors:] Finally, we have some number $q_\text{colorless}$ of colors which are used in the final step of the reduction,
where we replace the colorless vertices with colored vertices.
We later define $q_{\text{colorless}}$ to be the total number of colorless vertices in our constructed graph, but for now leave it as a variable.
\end{description}

Hence in total, we have $q = 2m \cdot \targetval + 3 + q_\text{colorless}$ many colors.

\paragraph{The central vertex $x$.}

The output graph of our reduction has a special colorless vertex $x$.
It is special because most important conditions are encoded via its neighborhood.
Accordingly, for most gadgets we construct for our reduction, an essential part of their ``output'' is how they control the coloring of neighbors of $x$.

\paragraph{Relations and gadgets.}

To describe the reduction, we define a specific type of gadget that is used throughout.
Each gadget has a set of input vertices $\Vin$, a set of output vertices $\Vout$, and a special set $\Vx$.
These three sets are always disjoint.
All vertices contained in $\Vx$ are adjacent to the central vertex $x$, and colors appearing on those vertices are used to encode vectors.
We thus refer to the set $\Vx$ also as the \emph{vector output} of the gadget.
In the input, we distinguish between single vertices and sets of vertices which we call \emph{color class inputs}.
Each gadget has some number $d$ of color class inputs which are always denoted by $I_1,\dots,I_d$.
Each set $I_j$ contains $2n^6$ vertices and the reader should think of each set being colored by one color class of the counting colors.

The behaviour of a gadget is defined by a set $\Cin$ and a relation $\CRel$.
The set $\Cin$ contains the set of colorings $\chiin\colon\Vin \to [q]$ of the input vertices that may occur (i.e., the gadget may behave arbitrarily if the input vertices are colored by some coloring not contained in $\Cin$).
The relation $\CRel$ defines what combinations of output colorings are allowed for what inputs.
Formally, the relation $\CRel$ is a set of tuples $(\chiin, \chiout, v_x)$ where $\chiin\colon\Vin \to [q]$ is a coloring from $\Cin$, and $\chiout\colon \Vout \to [q]$, and $v_x \in \ZZnn^d$ (where $d$ denotes the number of color class inputs).
For each $j \in [d]$, the $j$-th entry of $v_x$ corresponds to the $j$-th input color class.
Intuitively speaking, it contains the number of output vertices in $\Vx$ that are colored using colors from $\chiin(I_j)$.
More formally, each gadget has a corresponding function $p: (\Vx \to [q]) \to \ZZnn^d$ which maps colorings of $\Vx$ to output vectors.
For a given coloring $\chiin: \Vin \to [q]$ (which is always clear from context) and a coloring $\chix: \Vx \to [q]$, the $j$-th entry of $p(\chi_x)$ is
\begin{equation*}
  (p(\chix))[j] := |\chiin(I_j) \cap \chix(\Vx)|
\end{equation*}

Our definitions of $\CRel$ always ensure that $\sum_{j=1}^{d}v_x[j] = |\Vx|$, i.e., that all colors on the $x$-adjacent vertices are distinct.

We define each gadget in a \emph{specification table} containing the parameters that the construction and behaviour of the gadget depends on,
a description of the sets of interface vertices $\Vin, \Vout$ and $\Vx$, as well as a description of the sets $\Cin$ and $\CRel$.

For each gadget we use, we provide a construction and then prove that this construction does indeed conform to the behaviour described.
We define the latter formally as follows.

\begin{definition}
 We say that a gadget $G$ \emph{behaves according to $(\Cin,\CRel)$} if it satisfies the following properties:
 \begin{description}
  \item[Output guarantee:] 
   For every square $q$-coloring $\chi$ of $G$ such that $\chi|_{\Vin} \in \Cin$,
   it holds that
   \[(\chi|_{\Vin}, \chi|_{\Vout},p(\chi|_{\Vx})) \in \CRel.\]
  \item[Existence guarantee:]
   For every $(\chiin, \chiout, v_x) \in \CRel$ and every vector $S = (S_1, \ldots, S_d)$ with $S_1 \subseteq \chiin(I_1),S_2 \subseteq \chiin(I_2), \dots, S_d \subseteq \chiin(I_d)$
   such that $|S_j| = v_x[j]$ for all $j \in [d]$, there exists a square $q$-coloring $\chi$ of $G$ such that
   $\chi|_{\Vin} = \chiin$, $\chi|_{\Vout} = \chiout$ and $\chi(\Vx) \cap \chiin(I_j) = S_j$ for all $j \in [d]$.
\end{description}
Furthermore, we say that a gadget $G$ \emph{satisfies a specification table} if $G$ behaves according to the relation defined by the table entries specifying $\Cin$ and $\CRel$.
\end{definition}

We are now ready to describe our reduction.

\subsubsection{The Top-Level Structure}

In this section, we describe the reduction on a very high level assuming the existence of certain gadgets and prove its correctness.
More precisely, we assume the existence of three gadgets and use them to construct the reduction output.
We prove that this output is equivalent to the \textsc{Restricted Vector $k$-Sum} input instance.
The construction of the gadgets we use -- which is very involved -- is then given in the next section.

The gadgets we use here are called the \emph{subset gadget}, the \emph{color class copy gadget} and the \emph{vector selection gadget}.

\paragraph{The subset gadget.}
As mentioned above, we specify the inputs, outputs and behaviour of gadgets via a gadget specification table.
Table~\ref{tbl:summary-subset} does this for the subset gadget.

\begin{table}[h]
\centering
\begin{tabular}{ |p{2cm}|p{12cm}| }
  \hline

  \multicolumn{2}{| p{14cm} |}{Subset Gadget}\\
  
  \hline
  
  parameters & $\alpha, \beta \in \ZZp$ with $\alpha \geq \beta$\\
  
  \hline
  
  $\Vin$ & $\alpha$ vertices \\

  \hline
  
  $\Vout$ & $\beta$ vertices \\
  
  \hline
  
  $\Vx$ & $\emptyset$ \\
  
  \hline

  $\Cin$ & $|\chiin(\Vin)| = |\Vin|$ \\

  \hline
  
  $\CRel$ & $\chiout(\Vout) \subseteq \chiin(\Vin)$ and $|\chiout(\Vout)| = |\Vout|$ \\

  \hline
\end{tabular}
\caption{The specification table of the subset gadget}
\label{tbl:summary-subset}
\end{table}

To build intuition on gadget specification tables, let us explicitly state what the table expresses.
The first row tells us that the construction and behaviour of the subset gadget depend on two parameters $\alpha, \beta \in \ZZp$ with $\alpha \geq \beta$.
The next two rows specify that the gadget has $\alpha$ input vertices and $\beta$ output vertices.
Recall that in all our gadgets, the sets of input and output vertices are always disjoint.
The fourth row tells us that none of the vertices of the subset gadget are adjacent to the central vertex $x$ (as defined in the section on preparation).
The fifth row, which specifies the set $\Cin$, tells us that we only care about the behaviour of the subset gadget if the input coloring colors the input vertices with pairwise distinct colors.
Finally, the last row specifies the relation $C_R$, namely that the output coloring uses pairwise distinct colors that are a subset of the colors used in the input coloring.
Stated more intuitively, the output colors are a subset of the input colors.
Note that if $\alpha = \beta$, the output coloring uses exactly the input colors.
In this case, we also call the gadget an \emph{equality gadget}.

The subset gadget has to satisfy both the output guarantee and the existence guarantee.
In this case, the output guarantee states that as long as the input consists of pairwise distinct colors, the output has to be a subset of those colors.
The existence guarantee states that for any input coloring such that its colors are pairwise distinct and any output coloring such that its colors are a subset of the input colors, there exists a square $q$-coloring of the subset gadget that expands on those two colorings.

We use the subset gadget a lot in our reduction, mainly as equality gadget for copying around colors or groups of colors.

\paragraph{The color class copy gadget.}

Next, we describe the color class copy gadget.
It groups the counting colors into the color classes and creates a copy of each color class.
The gadget is formally specified by Table~\ref{tbl:summary-color-class-copy}.

\begin{table}[h]
\centering
\begin{tabular}{ |p{2cm}|p{12cm}| }
  \hline
  
  \multicolumn{2}{| p{14cm} |}{Color Class Copy Gadget}\\
  
  \hline
  
  parameters & $m,n \in \ZZp$ \\
  
  \hline
  
  $\Vin$ & vertex groups $X_1, \ldots, X_{2m}$, each of size $\targetval$; define $X := \bigcup_{i=1}^{2m} X_i$ \\
  
  \hline
  
  $\Vout$ & vertex groups $Y_1, \ldots, Y_{2m}$, each of size $\targetval$; define $Y := \bigcup_{i=1}^{2m} Y_i$ \\

  \hline
  
  $\Vx$ & $\emptyset$ \\

  \hline

  $\Cin$ & $\chiin(u) \neq \chiin(v)$ for all distinct $u,v \in X$ \\
  
  \hline

  $\CRel$ & $\chiin(X_i) = \chiout(Y_i)$ for all $i \in [2m]$ \\

  \hline
\end{tabular}
\caption{The specification table of the color class copy gadget}
\label{tbl:summary-color-class-copy}
\end{table}

Basically, this gadget takes as input a grouping of distinct colors on $X_1, \ldots, X_{2m}$, and guarantees that the same color groups appear on the out vertex groups $Y_1, \ldots, Y_{2m}$.
The reader may note that this gadget could be constructed by simply connecting $X_i$ to $Y_i$ via a subset gadget.
However, to bound the treewidth of the graph constructed by the entire reduction, we require that there is a separator of small size between $X$ and $Y$.
Towards this end, we need to design a specialized gadget.

\paragraph{The vector selection gadget.}

Finally, we come to the vector selection gadget.
Our reduction contains $k$ of these gadgets, one for each of the vector groups $A_1, \ldots, A_k$.
Its construction and that of its subgadgets is very involved, and turns out to be the longest part in the description of the reduction.

The purpose of this gadget is to simulate the selection of a vector from a node-representing vector list $A$ as defined in Definition~\ref{defn:node-representing-vector-group}.
Recall that this means there exist positive dimensions $D^+$ and negative
dimensions $D^-$ such that $|D^+ \cup D^-| = |D^+| + |D^-| = 3$, and that
we call $D^+ \cup D^-$ the non-zero dimensions. We define $D^+ \cup
D^- =: \{z_1, z_2, z_3\}$.

We remove zero-entries from each $a \in A$ and denote the result
\[\NZ(a) \coloneqq (a[z_1], a[z_2], a[z_3]).\]
The vector selection gadget has six color class inputs, and as such we have $d = 6$ for the dimension $d$ of the vector output.
We call a vector output $v_x \in \ZZ^6$ of this gadget a
\emph{vector-generated output of $A$} if there is some $a \in A$ such that
\[v_x =
(\halftargetval+\NZ(a)[1], \halftargetval-\NZ(a)[1],
\halftargetval+\NZ(a)[2], \halftargetval-\NZ(a)[2],
\halftargetval+\NZ(a)[3], \halftargetval-\NZ(a)[3]).\]
In this case we say that \emph{$v_x$ is generated by $a$ (or $\NZ(a)$)}.
The behaviour of the gadget is then specified by Table~\ref{tbl:summary-vector-selection}.

\begin{table}[h]
\centering
\begin{tabular}{ |p{2cm}|p{12cm}| }
  \hline \multicolumn{2}{| p{14cm} |}{Vector Selection Gadget}

  \\ \hline
  
  parameters & a node-representing vector list $A$ with parameters
  $m,n \in \ZZp$ \\

  \hline
  
  $\Vin$ &
  \vspace{-0.2cm}
  \begin{itemize}[noitemsep, nolistsep]
  \item six groups of $\targetval$ vertices $I_1, I_2, I_3, I_4, I_5, I_6$ (the color class inputs)
  \item three ``logic'' input vertices $r,g,b$
  \end{itemize}
  \vspace{-0.4cm}

  \\ \hline
  
  $\Vout$ & three ``logic'' output vertices $r',g',b'$

  \vspace{-0.4cm}

  \\ \hline

  $\Vx$ & a total of $6\halftargetval$ vertices\\

  \hline
  
  $\Cin$ &

  \vspace{-0.2cm}
  \begin{itemize}[noitemsep, nolistsep]
  \item $|\chiin(I_{1} \cup I_{2} \cup I_{3} \cup I_{4} \cup I_{5}
    \cup I_{6} \cup \{r,g,b\})| = $

    $|I_{1} \cup I_{2} \cup I_{3} \cup I_{4} \cup I_{5} \cup I_{6}
    \cup \{r,g,b\}|$
  \item $\chiin(r) = \red, \chiin(g) = \green, \chiin(b) = \blue$
  \end{itemize}
  \vspace{-0.4cm}
  \\ \hline
  
  $C_R$ &
  \vspace{-0.2cm}
  \begin{itemize}[noitemsep, nolistsep]
  \item 
      $\chiout(r') = \chiin(r)$,
      $\chiout(g') = \chiin(g)$ and
      $\chiout(b') = \chiin(b)$
  \item $v_x$ is a vector-generated output of $A$
  \end{itemize}
  \vspace{-0.4cm}

  \\ \hline

\end{tabular}
\caption{The specification table of the vector selection gadget}
\label{tbl:summary-vector-selection}
\end{table}

This specification is slightly more involved, so let us go through it
one by one. The first row simply tells us that the construction and
behaviour of the gadget depends on a node-representing list $A$ -- this is
the vector list from which we simulate the selection of a
vector. The gadget has six color class inputs and three inputs for the logic colors \red, \green\
and \blue. Furthermore, it has three logic color outputs. If we peek at
the first condition of $\CRel$ in the last row of the table, the gadget actually
guarantees that it outputs its logic color inputs
unchanged, meaning the logic color outputs $r',g',b'$ receive the same
colors as the logic color inputs $r,g,b$, respectively. So the
interesting behaviour does not happen in the coloring of the output,
but rather in the vector output of the gadget, i.e., in the coloring
of the neighbors of $x$. For this, we look at the fourth row of the
table, which specifies that $\Vx$ consists of $6n^6$ vertices.
Looking at the second condition of $\CRel$, we see that the vector
output encoded on these vertices is a vector-generated
output. Being vector-generated means that it encodes a certain vector
from the vector list $A$, but in a slightly modified form as
specified in the definition. Finally, the fifth row says
that we only care about the behaviour of the gadget if all inputs are
colored with pairwise different colors and if $r,g,b$ are colored with
the logic colors \red, \green\ and \blue.

\paragraph{Connecting the gadgets.}

We provide a sketch of the reduction in Figure~\ref{fig:overview}.
Before we describe how we fit the gadgets described above together, let us briefly discuss the intuition behind this construction.

\begin{figure}
 \centering
 \includesvg[width=0.95\linewidth]{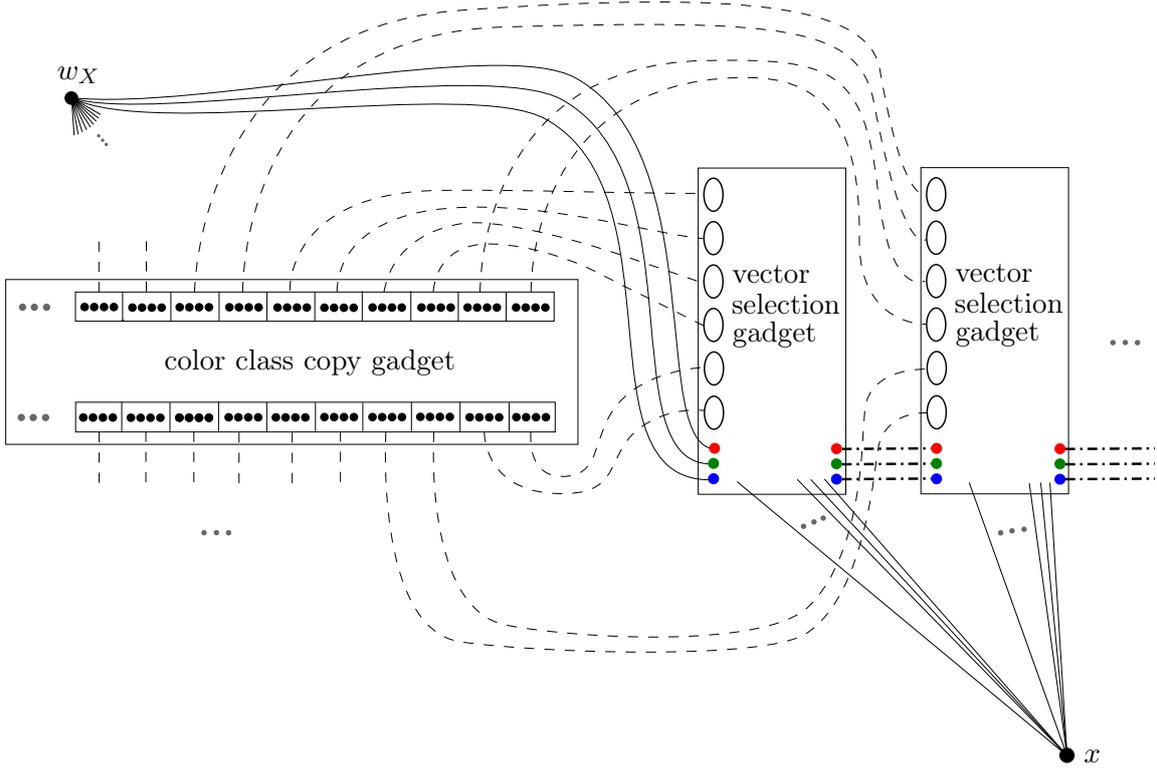}
 \caption{Overview of the construction. Dashed edges represent connections via equality gadgets.}
 \label{fig:overview}
\end{figure}

To reduce the \textsc{Restricted Vector $k$-Sum} instance to
a \SqCol\ instance, we need a concept of numbers,
vectors, summation and equality constraints in the context of square
$q$-colorings. The vector selection gadget already gives us a concept
of vectors and numbers via its vector output by coloring neighbors of
$x$ with a certain number of colors from each of its input color
classes. If two vector selection gadgets get the same color class as
one of their color class inputs, and they both expose a certain number
of colors from this class to $x$, this can be viewed as summation:
the total number of colors from that color class exposed to $x$ is the sum
of colors from that color class that are exposed to $x$ by each of the two vector selection gadgets.
Note that the sets of colors that are exposed have to be pairwise different, since all neighbors of $x$ are within distance $2$.
This gives us a concept of lesser-than-or-equal constraints:
the sum of the number of colors from a certain color class that are exposed by vector selection gadgets has to be less than the total number of colors from that color class.
Finally, we can encode an equality constraints $x = y$ by encoding it as the two lesser-than-or-equal constraints $x \leq y$ and $y \leq x$.

With this in mind, let us connect the gadgets we have so far.
First, we create the central vertex $x$ and a color class copy gadget.
We denote the latter by $G^{(\text{cpy})}$.
Next, we create $k$ vector selection gadgets, which we denote by $G^{(\text{sel},1)},\dots, G^{(\text{sel},k)}$.
For each $i \in [k]$, we define the parameter of $G^{(\text{sel},i)}$ to be the vector list $A_i$.

We identify any of the elements (e.g., sets, vertices, conditions) of $G^{(\text{sel},i)}$ or $G^{(\text{cpy})}$ by superscripting them with $(\text{sel},i)$ or $(\text{cpy})$, respectively.
For example, the input vertex set $\Vin$ of $G^{(\text{sel},1)}$ and the set $X_1$ of $G^{(\text{cpy})}$ become $\Vin^{(\text{sel},1)}$ and $X_1^{(\text{cpy})}$, respectively.

For each $i \in [k]$, we connect all vertices in $\Vx^{(\text{sel},i)}$ to $x$ via an edge.
In fact, these are the only vertices in the entire graph connected to $x$.

We create a colorless vertex $w_X$ and connect it to all vertices in $\Vin^{(\text{cpy})}$, as well as to the vertices $r^{(\text{sel},1)}, g^{(\text{sel},1)}, b^{(\text{sel},1)}$.
This ensures that all colors used in the color classes, as well as the logic colors, are pairwise distinct.
We also ensure that the logic colors are passed to each vector selection gadget by using one equality gadget each to connect, for each $i \in [k-1]$, the vertex
$(r')^{(\text{sel},i)}$ with $r^{(\text{sel},i+1)}$, $(g')^{(\text{sel},i)}$ with $g^{(\text{sel},i+1)}$ and $(b')^{(\text{sel},i)}$ with $b^{(\text{sel},i+1)}$.
Here, connecting two vertices via an equality gadget simply means that one of the vertices is the input and the other is the output of the equality gadget.

All that remains is to route the color classes from the color class copy gadget to the vector selection gadgets.
Recall that in a \textsc{Restricted Vector $k$-Sum} instance, each dimension $j \in [m]$ is a non-zero dimension in exactly two of the vector lists.
Also, each vector list is a node-representing vector list and thus has exactly three non-zero dimensions.
Finally, recall that dimension $j \in [m]$ corresponds to color classes $2j-1$ and $2j$.

For each $i \in [k]$, let vector list $A_i$ have non-zero dimensions $D^+ \cup D^- =: \{z^{(i)}_1, z^{(i)}_2, z^{(i)}_3\}$.
The corresponding color class indices are the elements of the ordered set
\[L^{(i)} \coloneqq (\ell^{(i)}_1,\dots,\ell^{(i)}_6) = (2z^{(i)}_1-1, 2z^{(i)}_1,2z^{(i)}_2-1, 2z^{(i)}_2, 2z^{(i)}_3-1, 2z^{(i)}_3).\]
We wish to connect the color classes with indices from $L^{(i)}$ to the inputs of $G^{(\text{sel},i)}$ in some way.
Note that by the fact that each dimension is a non-zero dimension in at most two vector lists, this means that each color class is an input to at most two vector selection gadgets.
In other words, $|\{i \mid \ell \in L^{(i)}\}| = 2$ for all $\ell \in [2m]$.

Hence, we can go through all $i \in [k]$ and $\ell^{(i)}_j \in L^{(i)}$.
Let us denote $\ell \coloneqq \ell^{(i)}_j$ for simplicity.
If $\ell \notin \bigcup_{i' \in [i-1]}L^{(i')}$, we connect $I^{(\text{sel},i)}_j$ to $X^{(\text{cpy})}_{\ell}$ via a subset gadget.
Otherwise, we connect $I^{(\text{sel},i)}_j$ to $Y^{(\text{cpy})}_{\ell}$ via a subset gadget.
Again, connecting two vertex groups via a subset gadget simply means that one vertex group is the input of the subset gadget and the other is the output.

This concludes the top-level construction of the reduction output instance.
Let us denote the resulting graph by $G^*$.

\paragraph{Correctness.}

We start by proving that the constructed \SqCol\ instance is equivalent to the given \textsc{Restricted Vector $k$-Sum} instance (assuming the described gadgets exist).

\begin{lemma}
 \label{lem:top-level-correctness}
 There are $a_1 \in A_1, \dots, a_k \in A_k$ such that $\sum_{i=1}^k a_i = 0$ if and only if $G^*$ has a square $q$-coloring.
\end{lemma}

For the proof, we need the following property, which is easy to verify for each of our gadget constructions.

\begin{property}
 \label{prop:only-colorless-vertices-adjacent-to-portals}
 For any gadget $G$ in our reduction except for the color class copy gadget, and for any vertex $v \in V(G)$ that is adjacent to a vertex from $\Vin \cup \Vout \cup \Vx$ of $G$, it holds that $v$ is colorless.
\end{property}

This property holds even for the lower-level gadgets which we
specify later. Stated more intuitively, its says that for all but one
gadget, any outwardly exposed vertices are only connected to colorless
vertices within the gadget. This makes our correctness proofs easier
because we do not need to worry about color conflicts between two
inner vertices (i.e., non-input and non-output vertices) that are from
different gadgets.

\begin{proof}[Proof of Lemma \ref{lem:top-level-correctness}]
 For the forward direction, assume the \textsc{Restricted Vector $k$-Sum} instance has a solution $a_1 \in A_1, \dots, a_k \in A_k$ such that $\sum_{i=1}^k a_i = 0$.
 We construct a square coloring $\chi$ of $G^*$ as follows.
 First, we color $\Vin^{(\text{cpy})}$ using all of the $2m\cdot \targetval$ counting colors.
 We also color $r^{(\text{sel},1)}, g^{(\text{sel},1)}, b^{(\text{sel},1)}$ using the logic colors \red, \green\ and \blue, respectively.
 Note that we have now already used every available color once.
 
 We can now arbitrarily color $\Vout^{(\text{cpy})}$ such that the coloring satisfies $\CRel^{(\text{cpy})}$, i.e., the color groupings are preserved.
 We may then use the existence guarantee of $G^{(\text{cpy})}$ to find a coloring for all vertices of $G^{(\text{cpy})}$.
 Now recall that each color class input of a vector selection gadget is connected via a subset gadget to either $X_i^{(\text{cpy})}$ for some $i \in [2m]$ or a $Y_i^{(\text{cpy})}$ for some $i \in [2m]$.
 We color it such that the condition $\CRel$ of the subset gadget is fulfilled, i.e., such that the color class input gets assigned all colors of that color class.
 We can then use the existence guarantee of each subset gadget to find a coloring for the vertices of the gadget.
 We also color, for every $i \in [k]$, the vertices $r^{(\text{sel},i)}, g^{(\text{sel},i)}, b^{(\text{sel},i)}$ with \red, \green\ and \blue, respectively.
 We do the same with $(r')^{(\text{sel},i)}, (g')^{(\text{sel},i)}, (b')^{(\text{sel},i)}$ for each $i \in [k]$.
 We can now use the existence guarantees for the subset gadgets that connect these logic inputs and logic outputs.
 
 All that remains is to expand the coloring $\chi$ we constructed so far to the vertices of the vector selection gadgets.
 Note that for each $i \in [k]$, the input and output vertices of vector selection gadget $G^{(\text{sel}, i)}$ are already colored.
 In particular, for each $\ell := \ell_j^{(i)} \in L^{(i)}$ (with $j \in [6]$),
 we know that $I^{(\text{sel},i)}_j$ is connected to either $X^{(\text{cpy})}_\ell$ or $Y^{(\text{cpy})}_\ell$ via a subset gadget,
 and is hence colored using the color class with index $\ell$.
 To use the existence guarantees, it suffices to specify an output vector $v_x^{(\text{sel}, i)}$ and a vector $S^{(\text{sel},i)}$ that satisfy the conditions of the existence guarantees.
 In other words, we need that $(\chi|_{\Vin^{(\text{sel},i)}}, \chi|_{\Vout^{(\text{sel}, i)}}, v_x^{(\text{sel},i)}) \in \CRel^{(\text{sel}, i)}$
 and $S^{(\text{sel}, i)} = (S_1^{(\text{sel}, i)}, \dots, S_6^{(\text{sel}, i)})$ with $S_1^{(\text{sel}, i)} \subseteq \chi(I_1^{(\text{sel}, i)}), \dots, S_6^{(\text{sel}, i)} \subseteq \chi(I_6^{(\text{sel}, i)})$
 such that $|S_j^{(\text{sel}, i)}| = (v_x^{(\text{sel},i)})[j]$ for all $j \in [6]$.
 The former means that the tuple must satisfy the conditions specified by $\Cin^{(\text{sel}, i)}$ and $C_R^{(\text{sel}, i)}$.
 The conditions in $\Cin^{(\text{sel}, i)}$ and the first condition of $C_R^{(\text{sel}, i)}$ are obviously satisfied.
 Hence it suffices to choose $v_x^{(\text{sel}, i)}$ such that it is a vector-generated output of $A_i$.
 We choose it to be the vector-generated output generated by $a_i$ from the solution of the \textsc{Restricted Vector $k$-Sum} instance, i.e., we set
 \begin{align*}
  v_x^{(\text{sel}, i)} \coloneqq (&\halftargetval+\NZ(a_i)[1], \halftargetval-\NZ(a_i)[1],\\
                                   &\halftargetval+\NZ(a_i)[2], \halftargetval-\NZ(a_i)[2],\\
                                   &\halftargetval+\NZ(a_i)[3], \halftargetval-\NZ(a_i)[3]).
 \end{align*}
 
 Next, we specify $S^{(\text{sel},i)}$.
 We need all neighbors of the central vertex $x$ to have distinct colors which means that $S^{(\text{sel},i)}_j \cap S^{(\text{sel},i')}_{j'} = \emptyset$ for all $i,i' \in [k]$ and $j,j' \in [6]$ such that $(i,j) \neq (i',j')$.
 Consider some $d \in [m]$.
 Let $i,i' \in [k]$ with $i < i'$ be the indices of the two vector lists in which $d$ is a non-zero dimension, i.e., $d = z^{(i)}_j$ and $d = z^{(i')}_{j'}$ for some $j,j' \in [3]$.
 We know that $X_{2d-1}^{(\text{cpy})}$ is connected via subset gadget to $I^{(\text{sel},i)}_{2j-1}$ and $X_{2d}^{(\text{cpy})}$ is connected to $I^{(\text{sel},i)}_{2j}$.
 Furthermore $Y_{2d-1}^{(\text{cpy})}$ is connected to $I^{(\text{sel},i')}_{2j'-1}$ and $Y_{2d}^{(\text{cpy})}$ is connected to $I^{(\text{sel},i')}_{2j'}$.
 
 We have $(v_x^{(\text{sel},i)})[2j-1] = \halftargetval+\NZ(a_i)[d]$ and
 $(v_x^{(\text{sel},i)})[2j] = \halftargetval-\NZ(a_i)[d]$ and
 $(v_x^{(\text{sel},i')})[2j-1] = \halftargetval+\NZ(a_{i'})[d]$ and
 $(v_x^{(\text{sel},i')})[2j] = \halftargetval-\NZ(a_{i'})[d]$.
 Since $i,i'$ are the only indices of vector lists where $d$ is a non-zero dimension, we have $a_i[d] a_{i'}[d] = 0$.
 Thus, $(v_x^{(\text{sel},i)})_{2j-1} + (v_x^{(\text{sel},i')})_{2j-1} = \targetval$ and $(v_x^{(\text{sel},i)})_{2j} + (v_x^{(\text{sel},i')})_{2j} = \targetval$.
 
 Note that each color class has exactly $\targetval$ colors.
 Hence, we can choose $S^{(\text{sel},i)}_{2j-1}$ as any arbitrary $(v_x^{(\text{sel},i)})_{2j-1}$ colors from $\chi(I^{(\text{sel},i)}_{2j-1}) = \chi(X_{2d-1}^{(\text{cpy})})$.
 We use the rest of the colors from $\chi(X_{2d-1}^{(\text{cpy})}) = \chi(I^{(\text{sel},i')}_{2j-1})$ for $S^{(\text{sel},i')}_{2j-1}$.
 Similarly, we choose $S^{(\text{sel},i)}_{2j}$ as any arbitrary $(v_x^{(\text{sel},i)})_{2j}$ colors from $\chi(I^{(\text{sel},i)}_{2j}) = \chi(X_{2d}^{(\text{cpy})})$ and choose $S^{(\text{sel},i')}_{2j-1}$ as the rest of the colors from $\chi(X_{2d}^{(\text{cpy})}) = \chi(I^{(\text{sel},i')}_{2j})$.
 
 Note that apart from guaranteeing that the neighbors of $x$ are colored using distinct colors, it is easy to see that actually, for each counting color, there exists exactly one neighbor of $x$ that is colored using that color.
 
 We can now finally use the existence guarantees of the vector selection gadgets to find colorings for them.
 Hence $\chi$ is now a complete coloring of our graph.

 Clearly, there are no color conflicts between neighbors of $w_X$.
 We have already argued above that there are no color conflicts between neighbors of $x$.
 Furthermore, by the existence guarantees of the gadgets we used, no color conflicts can exist between two vertices from the same gadget.
 Finally, color conflicts between two vertices from different gadgets are prevented by Property \ref{prop:only-colorless-vertices-adjacent-to-portals}.
 So $\chi$ is a valid square $q$-coloring.
 
 \medskip
 
 For the backward direction suppose $G^*$ has a square $q$-coloring $\chi$.
 We need to show that there are $a_1 \in A_1, \dots, a_k \in A_k$ such that $\sum_{i=1}^k a_i = 0$.
 
 The neighbors of $w_X$ need to be colored using pairwise different colors, meaning the vertices in $\Vin^{(\text{cpy})}$ and $r^{(\text{sel},1)}, g^{(\text{sel},1)}, b^{(\text{sel},1)}$ all have different colors. 
 Hence, we can use the output guarantee of $G^{(\text{cpy})}$ to conclude that $\chi(X_i^{(\text{cpy})}) = \chi(Y_i^{(\text{cpy})})$ for all $i \in [2m]$.
 We also use the output guarantees of all subset gadgets connecting vertex groups of the color class copy gadgets to color class inputs of vector selection gadgets.
 As a result, for every $i \in [k]$ and every $\ell := \ell^{(i)}_j \in L^{(i)}$, we have that $\chi(I_j^{(\text{sel},i)}) = \chi(X_\ell^{(\text{cpy})})$.
 
 Let us call $\chi(r^{(\text{sel},1)}), \chi(g^{(\text{sel},1)}), \chi(b^{(\text{sel},1)})$ \red, \green\ and \blue, respectively.
 For every $i \in [k]$, we have that $\chi((r')^{(\text{sel},i)}) = \chi(r^{(\text{sel},i)})$, $\chi((g')^{(\text{sel},i)}) = \chi(g^{(\text{sel},i)})$ and $\chi((b')^{(\text{sel},i)})) = \chi(b^{(\text{sel},i)})$ using the output guarantee of $G^{(\text{sel},i)}$.
 Furthermore, using the output guarantees of the subset gadgets connecting the logic color outputs of $G^{(\text{sel},i)}$ to the logic color inputs of $G^{(\text{sel},i+1)}$, it holds that
 $\chi(r^{(\text{sel},i+1)}) = \chi((r')^{(\text{sel},i)})$,
 $\chi(g^{(\text{sel},i+1)}) = \chi((g')^{(\text{sel},i)})$ and
 $\chi(b^{(\text{sel},i+1)})) = \chi((b')^{(\text{sel},i)})$ for all $i \in [k-1]$.
 Hence, by induction, all these vertices are colored \red, \green\ and \blue, respectively.
 
 Let $i \in [k]$.
 We use the output guarantees of $G^{(\text{sel},i)}$ again to conclude that
 \[v_x^{(\text{sel},i)} \coloneqq p(\chi|_{\Vx^{(\text{sel},i)}})\]
 is a vector-generated output of $A_i$, i.e., there is some $a_i \in A_i$ such that
 \begin{align*}
  v_x^{(\text{sel}, i)} \coloneqq (&\halftargetval+\NZ(a_i)[1], \halftargetval-\NZ(a_i)[1],\\
                                   &\halftargetval+\NZ(a_i)[2], \halftargetval-\NZ(a_i)[2],\\
                                   &\halftargetval+\NZ(a_i)[3], \halftargetval-\NZ(a_i)[3]).
 \end{align*}
 We claim that $\sum_{i=1}^k a_i = 0$.
 To show this, let us fix some dimension $d \in [m]$ and let $i,i' \in [k]$ with $i < i'$ be the indices of the two vector lists in which $d$ is a non-zero dimension, i.e., $d = z^{(i)}_j$ and $d = z^{(i')}_{j'}$ for some $j,j' \in [3]$.
 
 We know that $\chi(I^{(\text{sel},i)}_{2j-1}) = \chi(I^{(\text{sel},i')}_{2j'-1}) = \chi(X_{2d-1}^{(\text{cpy})})$ and $\chi(I^{(\text{sel},i)}_{2j}) = \chi(I^{(\text{sel},i')}_{2j'}) = \chi(X_{2d}^{(\text{cpy})})$.
 Hence, in the graph $G^{(\text{sel},i)}$, a total of $\halftargetval+\NZ(a_i)[j]$ neighbors of $x$ are colored using colors from $\chi(X_{2d-1}^{(\text{cpy})})$ and a total of $\halftargetval-\NZ(a_i)[j]$ neighbors of $x$ are colored using colors from $\chi(X_{2d}^{(\text{cpy})})$.
 Similarly, in the graph $G^{(\text{sel},i')}$, a total of $\halftargetval+\NZ(a_{i'})[j']$ neighbors of $x$ are colored using colors from $\chi(X_{2d-1}^{(\text{cpy})})$ and a total of $\halftargetval-\NZ(a_{i'})[j']$ neighbors of $x$ are colored using colors from $\chi(X_{2d}^{(\text{cpy})})$.
 So overall, $(\halftargetval + \NZ(a_i)[j]) + (\halftargetval + \NZ(a_{i'})[j'])$ neighbors of $x$ are colored using colors from $\chi(X_{2d-1}^{(\text{cpy})})$,
 and $(\halftargetval - \NZ(a_i)[j]) + (\halftargetval - \NZ(a_{i'})[j'])$ neighbors are colored using colors from $\chi(X_{2d}^{(\text{cpy})})$.
 
 Recall that $|X_r^{(\text{cpy})}| = \targetval$ for all $r \in [2m]$.
 Since $\chi$ has to color neighbors of $x$ using pairwise distinct colors, it follows that
 \[(\halftargetval + \NZ(a_i)[j]) + (\halftargetval + \NZ(a_{i'})[j']) \leq \targetval\] and
 \[(\halftargetval - \NZ(a_i)[j]) + (\halftargetval - \NZ(a_{i'})[j']) \leq \targetval.\]
 Rewriting, we get $\NZ(a_i)[j] + \NZ(a_{i'})[j'] \leq 0$ and $\NZ(a_i)[j] + \NZ(a_{i'})[j'] \geq 0$.
 In combination, $\NZ(a_i)[j] + \NZ(a_{i'})[j'] = 0$ which implies that $a_i[d] + a_{i'}[d] = 0$ as desired.
\end{proof}

\subsubsection{Constructing the Gadgets}

Next, we construct the three gadgets used in the construction described above.
The construction of the subset gadget and color class copy gadget are fairly simple.
However, the construction of the vector selection gadget is more involved and requires us to first build a long sequence of subgadgets.

\paragraph{The subset gadget.}

We start with the construction of the subset gadget.
Let $\alpha \geq \beta \geq 1$.
Recall that an $(\alpha,\beta)$-subset gadget has $\alpha$ input vertices and $\beta$ output vertices, and it ensures that the colors used to color the output vertices are a subset of the colors used to color the input vertices (see Table \ref{tbl:summary-subset}).
A \emph{$\gamma$-equality gadget} is a $(\gamma,\gamma)$-subset gadget.

We construct the gadget $\SubsetGadget(\alpha,\beta)$ as follows.
It has $\alpha$ input vertices $\Vin$ that are all connected to a colorless vertex $a$.
The vertex $a$ is then connected to $q-\alpha$ fresh vertices $C$ to which we refer as the \emph{complement vertices}.
The complement vertices $C$ are all connected to another colorless vertex $b$, which is also connected to all output vertices $\Vout$.

Observe that the subset gadget $\SubsetGadget(\alpha,\beta)$ satisfies Property \ref{prop:only-colorless-vertices-adjacent-to-portals}.

\begin{lemma}
  The gadget $\SubsetGadget(\alpha,\beta)$ satisfies the specification table of the subset gadget (Table~\ref{tbl:summary-subset}).
\end{lemma}

\begin{proof}
 We first show the output guarantee.
 Let $\chi$ a square $q$-coloring of $\SubsetGadget(\alpha,\beta)$ such that $\chi|_{\Vin} \in \Cin$.
 Since the set $\Vin \cup C$ forms a clique in the square of $\SubsetGadget(\alpha,\beta)$, the coloring $\chi$ is injective on $\Vin \cup C$.
 Because there are $\alpha$ input vertices and $q-\alpha$ complement vertices, it follows that the complement vertices are colored using all colors which do not appear on the input vertices.
 Similarly, the output vertices are colored using pairwise distinct colors which do not appear on the complement vertices.
 This imlies that $\chi(\Vout) \subseteq \chi(\Vin)$.
 Also, $|\chi(\Vout)| = |\Vout|$ since all vertices from $\Vout$ are adjacent to $b$.
 So overall, $(\chi|_{\Vin}, \chi|_{\Vout},p(\chi|_{\Vx})) \in \CRel$ as desired.
 
 For the existence guarantee let $(\chiin, \chiout, ()) \in \CRel$.
 We extend $\chiin$ and $\chiout$ to a coloring $\chi$ by coloring the complement vertices $C$ using the $q-\alpha$ colors not appearing on the input.
 It is easy to see this is a valid square $q$-coloring of $\SubsetGadget(\alpha,\beta)$.
\end{proof}

\paragraph{The color class copy gadget.}
Next, we construct the color class copy gadget (see Table~\ref{tbl:summary-color-class-copy}).
As we already mentioned above, there is a very simple way to construct this gadget by connecting $X_i$ to $Y_i$ using a $(\targetval)$-equality gadget for all $i \in [2m]$.
However, this construction is not sufficient for our purposes since the treewidth of the overall graph $G^*$ is too large when using this simple construction.
Instead, we need a gadget for which there is a separator of size $\log m + O(1)$ between $X$ and $Y$, i.e., we can remove $\log m + O(1)$ vertices from the gadget to disconnect $X$ from $Y$.

Towards this end, we construct the color class copy gadget $\CopyGadget$ as follows.
First, we ensure that the colors used to color the vertices in $\Vin = X$ are the same as the colors used to color the vertices in $\Vout = Y$.
This is achieved by connecting the all vertices in $X$ to the vertices in $Y$ via a $(2m \cdot \targetval)$-equality gadget.

Now let
\begin{equation}
 \label{eq:def-r}
 r \coloneqq \left\lceil\frac{(1+\varepsilon)\log m}{2}\right\rceil
\end{equation}
(recall that $\varepsilon$ is part of the input of Theorem \ref{thm:vector-sum-to-distance-two-coloring}).
Then
\begin{equation}
 \label{eq:separator-subsets}
 {2r \choose r} \geq \frac{2^{2r}}{2r+1} \geq \frac{m^{(1+\varepsilon)}}{(1+\varepsilon)\log m + 3} > 2m
\end{equation}
where the last inequality is equivalent to Equation \eqref{eq:m-epsilon-def} which holds by our choice of $m_{\varepsilon}$.

We introduce $2r$ fresh colorless vertices $S$.
For each $0 \leq i \leq 2m$ we choose a distinct subset $S_i \subseteq S$ of size $r$ which is possible by Equation \eqref{eq:separator-subsets}.
We connect every vertex of $X_i$ to all vertices in $S_i$ (i.e., $X_i$ and $S_i$ form a biclique).
Similarly, we connect every vertex of $Y_i$ to all vertices in $S \setminus S_i$ (i.e., $Y_i$ and $S \setminus S_i$ form a biclique).
This completes the construction of the gadget $\CopyGadget$.

\begin{lemma}
 The gadget $\CopyGadget$ satisfies the specification table of the color class copy gadget (Table~\ref{tbl:summary-color-class-copy}).
\end{lemma}

\begin{proof}
 We first show the output guarantee.
 Let $\chi$ be a square $q$-coloring of $\CopyGadget$ such that $\chi|_{\Vin} \in \Cin$, i.e., $\chi(x) \neq \chi(x')$ for all distinct $x,x' \in X$.
 
 The $(2m\cdot\targetval)$-equality gadget ensures that $\chi(Y) = \chi(X)$.
 Now consider some $i \in [2m]$.
 First note that the colors of the vertices in $Y_i$ are pairwise distinct, since they are all adjacent to all vertices in $S\setminus S_i$ and thus within distance $2$ of each other.
 
 Now consider some $i \neq j \in [2m]$.
 Then $S_j \cap (S \setminus S_i) \neq \emptyset$ which implies that all vertices from $Y_i$ are within distance at most $2$ from all vertices in $X_j$.
 So $\chi(Y_i) \cap \chi(X_j) = \emptyset$.
 Since $\chi(Y_i) \subseteq \chi(Y) = \chi(X)$ it follows that $\chi(Y_i) \subseteq \chi(X_i)$.
 Since $|X_i| = |Y_i| = \targetval$ and colors of the vertices in $Y_i$ are pairwise distinct, we conclude that $\chi(Y_i) = \chi(X_i)$.
 In other words, $(\chi|_{\Vin}, \chi|_{\Vout}, ()) \in \CRel$ as desired.
 
 \medskip 
  
 For the existence guarantee let $(\chiin,\chiout, ()) \in \CRel$.
 Clearly, the input and output color sets of the subset gadget connecting $X$ and $Y$ are the same.
 So we can use the existence guarantee for the subset gadget to derive a coloring for it.
 Combining this coloring with $\chiin$ and $\chiout$ gives a $q$-coloring $\chi$ (recall that all vertices in $S$ are colorless).
 
 It remains to show that this $\chi$ is a valid square $q$-coloring.
 Using Property \ref{prop:only-colorless-vertices-adjacent-to-portals} it is easy to see that no color conflict arises from the inner vertices of the subset gadget.
 The only remaining possible conflicts are between vertices of $X_i$ and $Y_i$ for some $i \in [2m]$, since these are the only vertices that share a color.
 However, $N(X_i) \cap N(Y_i) = \emptyset$ for all $i \in [2m]$ by definition.
 It follows that $\chi$ is a square $q$-coloring of $\CopyGadget$.
 \end{proof}

\paragraph{The socket gadget.}
We now come to the series of subgadgets that are needed for the vector selection gadget.
We start with the most low-level ones and use them to build more complex gadgets.

The first fundamental building block is the socket gadget.
There are two types of socket gadgets: constant sockets and switch gadgets.
Both types of gadgets have a control input which is either \red\ or \blue.
The constant sockets always expose a color from one color class $I_i$ (for some $i \in \{1,2\}$) to $x$, irregardless of the color of the control input (recall that exposing a color to $x$ means that a neighbor of $x$ is assigned that color).
We call constant sockets that always expose a color from the coloring of $I_1$ \emph{constant sockets of type $1$}, and similarly constant sockets that always expose a color from the coloring of $I_2$ \emph{constant sockets of type $2$}.
On the other hand, switch sockets output either a color from $I_1$ or $I_2$ depending on the control input.
We describe switch sockets by an abbreviation of their control-to-exposed-color mapping.
Specifically, switch gadgets that show a color from $I_1$ for a \red\ control input and a color from $I_2$ for a \blue\ control input are called \emph{switch sockets of type r1b2}, while the other type where the outputs are switched is called a \emph{switch socket of type r2b1}.

We specify all the different types of socket gadgets in Table~\ref{tbl:summary-socket}.
We use the Kronecker delta defined as $\delta_{i,j} \coloneqq 1$ if $i = j$ and $\delta_{i,j} \coloneqq 0$ otherwise.

\begin{table}[h]
 \centering
 \begin{tabular}{|p{2cm}|p{12cm}|}
  \hline \multicolumn{2}{|p{14cm}|}{Socket Gadget}\\
  
  \hline parameters & $n \in \ZZp$ and the type (constant / switch) and subtype (1 / 2 / r1b2 / r2b1) of the gadget\\
  
  \hline $\Vin$ &
  \vspace{-0.2cm}
  \begin{itemize}[noitemsep, nolistsep]
   \item two groups of $\targetval$ vertices $I_1$ and $I_2$ (the color class inputs)
   \item three ``logic'' input vertices $r,g,b$
   \item a ``control'' input vertex $s$
  \end{itemize}
  \vspace{-0.4cm}
  \\

  \hline

  $\Vout$ &
  \vspace{-0.2cm}
  \begin{itemize}[noitemsep, nolistsep]
   \item two groups of $\targetval$ vertices $I_1', I_2'$
   \item three ``logic'' output vertices $r',g',b'$
   \item a ``control'' output vertex $s'$
  \end{itemize}
  \vspace{-0.4cm}
  \\

  \hline

  $\Vx$ & a single vertex $z$ \\
  
  \hline

  $\Cin$ &
  \vspace{-0.2cm}
  \begin{itemize}[noitemsep, nolistsep]
   \item $|\chiin(I_1 \cup I_2 \cup \{r,g,b\})| = |I_1 \cup I_2 \cup \{r,g,b\}|$
   \item $\chiin(r) = \red$, $\chiin(g) = \green$, $\chiin(b) = \blue$
   \item $\chiin(s) \in \{\red, \blue\}$
  \end{itemize}
  \vspace{-0.4cm}
  \\

  \hline

  $\CRel$ & We divide $\CRel$ into output and vector conditions:
  \begin{enumerate}[noitemsep, nolistsep, label = (\arabic*)]
   \item\label{item:socket-output-condition}
    \begin{itemize}[noitemsep,nolistsep]
     \itemsep0em 
     \item $\chiout(r') = \chiin(r)$, $\chiout(g') = \chiin(g)$, $\chiout(b') = \chiin(b)$
     \item $\chiout(s') = \chiin(s)$
     \item $\chiout(I'_i) = \chiin(I_i)$ for both $i \in \{1,2\}$
    \end{itemize}
   
   \item\label{item:socket-vector-condition}
    \begin{description}[itemsep=0pt,parsep=0pt,topsep=0pt,partopsep=0pt]
     \item[For constant sockets of type $\tau$ ({$\tau \in \{1,2\}$}):]$\hphantom{-}$\newline $v_x = (\delta_{\tau,1},\delta_{\tau,2})$
     \item[For switch sockets of type r$\tau$b$(3-\tau)$ ({$\tau \in \{1,2\}$}):]
      $\hphantom{-}$\newline $v_x = \begin{cases}
                                     (\delta_{\tau,1}, \delta_{\tau,2}) & \text{if } \chiin(s) = \red\\
                                     (\delta_{\tau,2}, \delta_{\tau,1}) & \text{if } \chiin(s) = \blue
                                    \end{cases}$

    \end{description}
  \end{enumerate}
  \\

  \hline
 \end{tabular}
 \caption{The specification table of the socket gadget.}
 \label{tbl:summary-socket}
\end{table}

We first construct the constant socket $\ConstSocket_\tau$ of type $\tau \in \{1,2\}$.
This construction is simple.
First, we copy the inputs to the corresponding outputs via appropriate equality gadgets.
Moreover, the input color class $I_\tau$ is connected to $z$ (the vertex from $\Vx$) via a $(\targetval,1)$-subset gadget.

Next, we construct the switch sockets. 
We describe an r2b1 switch socket $\SwitchSocket_2$ (see Figure~\ref{fig:socket}).
The switch socket of type r1b2 $\SwitchSocket_1$ is constructed analogously by swapping the two color class inputs.

\begin{figure}
 \centering
 \includesvg[width=0.7\linewidth]{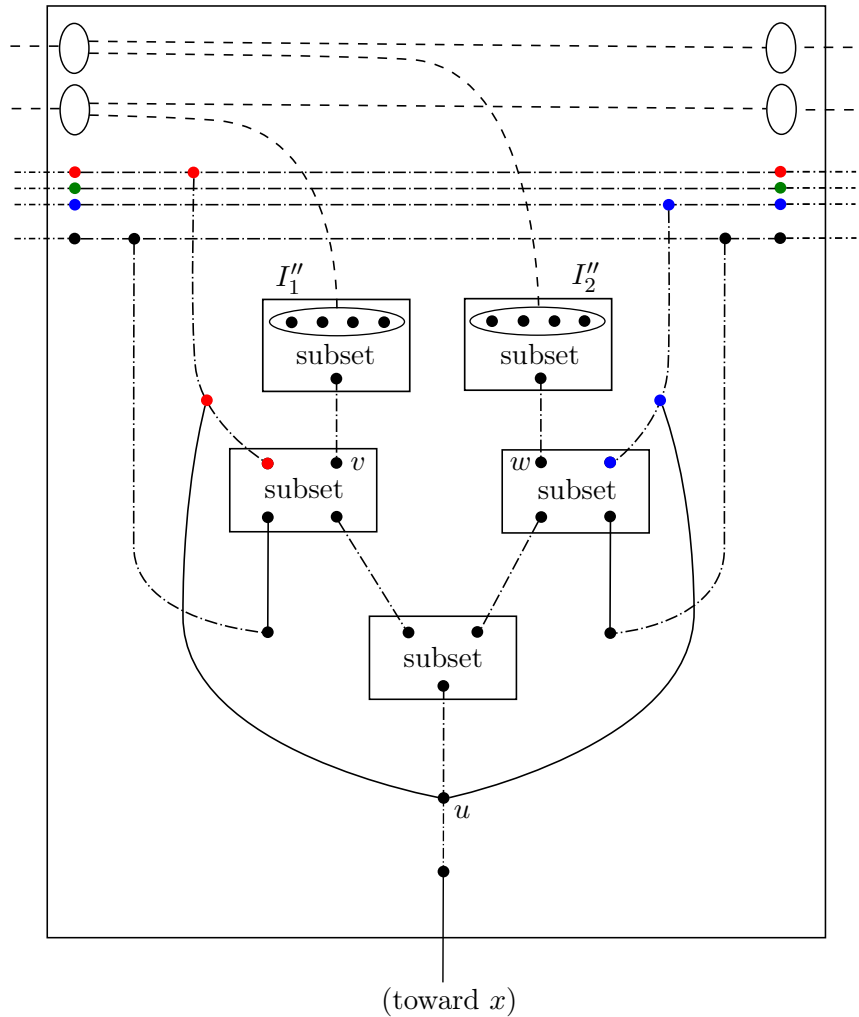}
 \caption{Illustration of a switch socket gadget.
  The lower color class input is $I_1$ in this figure, and hence this socket is of type r2b1.
  Dashed edges represent equality gadgets with more than one input, while dash-dotted edges represent $1$-equality gadgets.}
 \label{fig:socket}
\end{figure}

As before, we first copy the inputs to the corresponding outputs via appropriate equality gadgets.
We create a copy $I_1''$ of the color class input $I_1$ and a copy $I_2''$ of the color class input $I_2$ using $(\targetval)$-equality gadgets. 
The vertices in the copy $I_1''$ serve as input to a $(\targetval,1)$-subset gadget whose output is copied to a fresh vertex $v$ (via a $1$-equality gadget).
Similarly, the vertices in $I_2''$ are the input to another $(\targetval,1)$-subset gadget whose output is copied to a fresh vertex $w$.

The vertex $v$ as well as a \red\ vertex (we create a fresh \red\ vertex using a $1$-equality gadget connected to the input vertex $r$) are the inputs to a $(2,2)$-equality gadget.
We refer to this gadget as the \emph{left equality gadget}.
One of the outputs of the left equality gadget is connected by a normal edge to a vertex that copies the control input.
The other vertex is copied to a $(2,1)$-subset gadget which we refer to as the \emph{final subset gadget}.
On the other side, the vertex $w$ as well as a \blue\ vertex are the inputs to another $(2,2)$-equality gadget which we call the \emph{right equality gadget}.
One of the outputs of the right equality gadget is again connected to a vertex that copies the control input, while the other one is copied to the second input of the final subset gadget.

The output of the final subset gadget is copied to a vertex $u$.
We connect $u$ to a \red\ and a \blue\ vertex (as before, we create a fresh \red\ vertex using a $1$-equality gadget connected to the input vertex $r$; the \blue\ vertex is created analogously).
Finally, we connect $u$ to the vertex $z$ via a $1$-equality gadget.

Observe that all socket gadgets satisfy Property \ref{prop:only-colorless-vertices-adjacent-to-portals}.

\begin{lemma}\label{lem:socket}
  The gadgets $\ConstSocket_\tau$ and $\SwitchSocket_\tau$, $\tau \in \{1,2\}$, satisfy the specification table of the socket gadget (Table~\ref{tbl:summary-socket}).
\end{lemma}

\begin{proof}
 For the constant sockets, it is easy to verify that both the output guarantee and the existence guarantee are satisfied.
 So consider the switch socket $\SwitchSocket_\tau$ and suppose without loss generality that $\tau = 2$.
 
 We first show the output guarantee.
 Let $\chi$ be a square $q$-coloring of $\SwitchSocket_2$ such that $\chi|_{\Vin} \in \Cin$.
 It is easy to verify that $\chi$ satisfies Condition \ref{item:socket-output-condition} of $\CRel$.
 So consider Condition \ref{item:socket-vector-condition}.
 We need to show that $\chi(z) \in \chi(I_2)$ if $\chi(s) = \red$ and $\chi(z) \in \chi(I_1)$ if $\chi(s) = \blue$.
 
 The output guarantee of the first two subset gadgets implies that $\chi(v) \in \chi(I_1)$ and $\chi(w) \in \chi(I_2)$.
 The left equality gadget gets as input $v$ and a \red\ vertex, and one of its outputs is adjacent to a vertex which copies the control input.
 Hence, if the control input is \red, that output cannot be \red, and the left input $v'$ to the final subset gadget is colored \red.
 Similarly, if the control input is \blue, the right input $w'$ to the final subset gadget is colored \blue.
  
 Since the output of the final subset gadget is copied to $u$ and $u$ is adjacent to a \red\ and a \blue\ vertex, it follows that one of the inputs to the final subset gadget is colored neither \red\ nor \blue.
 As argued above, if the control input is \red, the left equality gadget copies \red to the final subset gadget.
 So the left equality gadget needs to copy $\chi(w) \in \chi(I_2)$ to the final equality gadget and we get that $\chi(z) = \chi(u) = \chi(w) \in \chi(I_2)$.
 Similarly, if the control input is \blue, we get that $\chi(z) = \chi(u) = \chi(v) \in \chi(I_1)$.
 
 \medskip
 
 For the existence guarantees let $(\chiin,\chiout, v_x) \in \CRel$ and suppose $S = (S_1,S_2)$ where $S_1 \subseteq \chiin(I_1)$, $S_2 \subseteq \chiin(I_2)$ and $|S_j| = v_x[j]$ for both $j \in \{1,2\}$, i.e., $|S_1| + |S_2| = 1$.
 Without loss of generality suppose $\chiin(s) = \red$ (the other case is analogous).
 This means $v_x = (0,1)$ and hence, $S_1 = \emptyset$ and $S_2 = \{c\}$ for some color $c \in \chiin(I_2)$.
 
 We extend $\chiin$ and $\chiout$ to a square $q$-coloring $\chi$ of $\SwitchSocket_2$ as follows.
 We set
 \begin{itemize}
  \item $\chi(I_1'') \coloneqq \chiin(I_1)$ and $\chi(I_2'') \coloneqq \chiin(I_2)$,
  \item $\chi(v) \coloneqq d$ for some arbitrary $d \in \chi(I_1)$,
  \item $\chi(w) \coloneqq c$
  \item $\chi(v') \coloneqq \red$ where $v'$ is the left input of the final subset gadget,
  \item $\chi(w') \coloneqq c$ where $w'$ is the right input of the final subset gadget, and
  \item $\chi(u) = \chi(z) \coloneqq c$.
 \end{itemize}
 It is easy to check that this coloring extends in a unique way to the remaining vertices taking all equality gadgets into account without creating any color conflicts.
 Finally, we use the existence guarantee of the subset gadgets to extend to the coloring to the inner vertices of all subset gadgets.
 Observe that there are no color conflicts involving inner vertices of subset gadgets by Property~\ref{prop:only-colorless-vertices-adjacent-to-portals}.
\end{proof}

\paragraph{The edge selection gadget.}
We use the low-level socket gadget to build the \emph{edge selection gadget}.
Intuitively, it uses subset gadgets and socket gadgets to simulate the selection of an edge between two groups of vertices.
Its behaviour is specified in Table~\ref{tbl:summary-edge-selection}.

\begin{table}[h]
 \centering
 \begin{tabular}{|p{2cm}|p{12cm}|}
  \hline \multicolumn{2}{| p{14cm} |}{Edge Selection Gadget}\\ \hline
  parameters & $n \in \ZZp$ and $\alpha \in [-n^2, n^2]$\\ \hline
  
  $\Vin$ &
  \vspace{-0.2cm}
  \begin{itemize}[noitemsep, nolistsep]
   \item two groups of $\targetval$ vertices $I_1$ and $I_2$ (the color class inputs)
   \item three ``logic'' input vertices $r,g,b$
   \item a ``control'' input vertex $s$ 
  \end{itemize}
  \vspace{-0.4cm}
  \\
    
  \hline $\Vout$ &
  \vspace{-0.2cm}
  \begin{itemize}[noitemsep, nolistsep]
  \item two groups of $\targetval$ vertices $I'_1, I'_2$
  \item three ``logic'' output vertices $r',g',b'$
  \item a ``control'' output vertex $s'$
  \end{itemize}
  \vspace{-0.4cm}
  \\
  
  \hline $\Vx$ & a total of $2n^2$ vertices \\

  \hline $\Cin$ &
  \vspace{-0.2cm}
  \begin{itemize}[noitemsep, nolistsep]
   \item $|\chiin(I_1 \cup I_2 \cup \{r,g,b\})| = |I_1 \cup I_2 \cup \{r,g,b\}|$
   \item $\chiin(r) = \red, \chiin(g) = \green, \chiin(b) = \blue$
   \item $\chiin(s) \in \{\red,\blue\}$
  \end{itemize}
  \vspace{-0.4cm}
  \\

  \hline $\CRel$ & We divide $\CRel$ into output and vector conditions:
  \begin{enumerate}[noitemsep, nolistsep, label = (\arabic*)]
   \item\label{item:edge-selection-output-condition}
    \begin{itemize}[noitemsep, nolistsep]
     \item
      $\chiout(r') = \chiin(r)$,
      $\chiout(g') = \chiin(g)$ and
      $\chiout(b') = \chiin(b)$
     \item $\chiout(s') = \chiin(s)$
     \item $\chiout(I'_i) = \chiin(I_i)$ for both $i \in \{1,2\}$
    \end{itemize}

  \item\label{item:edge-selection-vector-condition}
   \begin{itemize}[noitemsep, nolistsep]
    \item $v_x = (n^2, n^2)$ if $\chiin(s) = \red$
    \item $v_x = (n^2+\alpha, n^2-\alpha)$ if $\chiin(s) = \blue$
   \end{itemize}
    
  \end{enumerate}
  \vspace{-0.4cm}
  \\ \hline
 \end{tabular}
 \caption{The specification table of the edge selectiongadget}
 \label{tbl:summary-edge-selection}
\end{table}

Let $\alpha \in [-n^2, n^2]$.
We construct the edge selection gadget $\EdgeSelection(\alpha)$ as follows (see Figure~\ref{fig:edge-selection}).
It contains a chain of $2n^2$ socket gadgets which we denote by $G^{(1)},\dots,G^{(2n^2)}$.
For every $j \in [2n^2]$, we identify the elements (e.g., sets and vertices) of $G^{(j)}$ by superscripting them with $(j)$.
For example, the input vertex set $\Vin$ of $G^{(1)}$ and the vertex $z$ of $G^{(2)}$ are denoted by $\Vin^{(1)}$ and $z^{(2)}$, respectively.

\begin{figure}
 \centering
 \includesvg[width=0.95\linewidth]{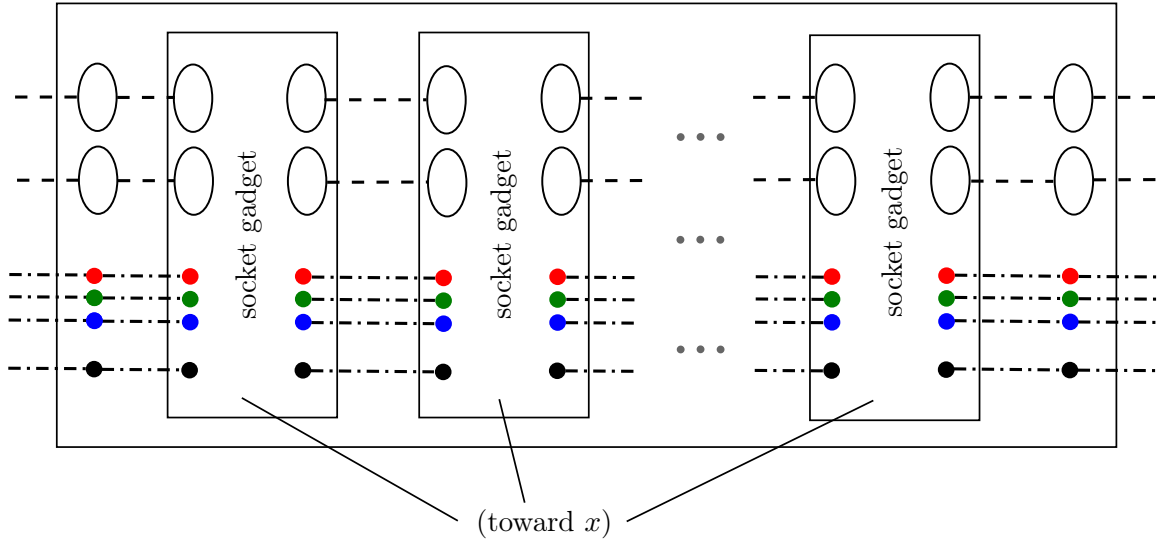}
 \caption{Illustration of the edge selection gadget.}
 \label{fig:edge-selection}
\end{figure}

The inputs $\Vin$ of the edge selection gadget are copied via equality gadgets to the corresponding inputs in the set $\Vin^{(1)}$.
Also, for every $j \in [2n^2-1]$, the outputs $\Vout^{(j)}$ are copied via equality gadgets to the corresponding inputs in $\Vin^{(j+1)}$.
Finally, the outputs of $G^{(2n^2)}$ are copied to the outputs $\Vout$ of the edge selection gadget.

It remains to specify the types of socket gadgets that are used in the chain.
The gadgets $G^{(1)}$ to $G^{(n_1)}$, where $n_1 \coloneqq \min\{n^2, n^2+\alpha\}$, are constant sockets of type $1$.
After that, the gadgets $G^{(n_1+1)}$ to $G^{(n_1+n_2)}$, where $n_2 \coloneqq \min\{n^2,n^2-\alpha\}$, are constant sockets of type $2$.
The remaining gadgets $G^{(n_1+n_2+1)}$ to $G^{(n_1+n_2+n_3)}$, where $n_3 \coloneqq |\alpha|$, are switch sockets.
If $\alpha \geq 0$, they are switch sockets of type r2b1.
If $\alpha < 0$, they are switch sockets of type r1b2.

Finally, we set $\Vx \coloneqq \bigcup_{j \in [2n^2]}\Vx^{(j)} = \{z^{(j)} \mid j \in [2n^2]\}$.

Observe that the edge selection gadget satisfies Property \ref{prop:only-colorless-vertices-adjacent-to-portals}.

\begin{lemma}
 The gadget $\EdgeSelection(\alpha)$ satisfies the specification table of the edge selection gadget (Table~\ref{tbl:summary-edge-selection}).
\end{lemma}

\begin{proof}
 We first show the output guarantee.
 Let $\chi$ be a square $q$-coloring of $\EdgeSelection(\alpha)$ such that $\chi|_{\Vin} \in \Cin$.
 
 Using the output guarantee of the equality gadget and the socket gadget, we first conclude that
 \begin{equation}
  \label{eq:edge-selection-gadget-chain}
  \chi(I_i) = \chi(I_i^{(j)}) = \chi((I_i')^{(j)}) = \chi(I_i')
 \end{equation}
 for all $i \in \{1,2\}$ and $j \in [2n^2]$.
 Similarly, it follows that $\chi(r') = \chi(r)$, $\chi(g') = \chi(g)$, $\chi(b') = \chi(b)$ and $\chi(s') = \chi(s)$.
 In particular, Condition \ref{item:edge-selection-output-condition} is satisfied.
 
 Using Equation \eqref{eq:edge-selection-gadget-chain} and $\Vx = \bigcup_{j \in [2n^2]}\Vx^{(j)}$, we conclude that the vector output of the edge selection gadget is $v_x = \sum_{j=1}^{2n^2}v_x^{(j)}$.
 To explicitly calculate the sum, let us consider the following four cases:
 \begin{enumerate}
  \item\label{item:case-rp} $\chiin(s) = \red$ and $\alpha > 0$
  \item\label{item:case-bp} $\chiin(s) = \red$ and $\alpha < 0$
  \item\label{item:case-rn} $\chiin(s) = \blue$ and $\alpha > 0$
  \item\label{item:case-bn} $\chiin(s) = \blue$ and $\alpha < 0$
 \end{enumerate}
 We have
 \begin{align*}
   n_1 = \min\{n^2, n^2+\alpha\} &= \begin{cases}
                                     n^2        &\text{in case \ref{item:case-rp} or \ref{item:case-rn}}\\
                                     n^2+\alpha &\text{in case \ref{item:case-bp} or \ref{item:case-bn}}
                                    \end{cases}\\
   n_2 = \min\{n^2,n^2-\alpha\}  &= \begin{cases}
                                     n^2        &\text{in case \ref{item:case-bp} or \ref{item:case-bn}}\\
                                     n^2-\alpha &\text{in case \ref{item:case-rp} or \ref{item:case-rn}}
                                    \end{cases}\\
   n_3 = |\alpha|                &= \begin{cases}
                                     \alpha  &\text{in case \ref{item:case-rp} or \ref{item:case-rn}}\\
                                     -\alpha &\text{in case \ref{item:case-bp} or \ref{item:case-bn}}
                                    \end{cases}
 \end{align*}
 Using the output guarantee for the socket gadgets, we get that $v_x^{(1)} = \dots = v_x^{(n_1)} = (1,0)$ and $v_x^{(n_1+1)} = \dots = v_x^{(n_1+n_2)} = (0,1)$.
 Also, for the last $n_3$ gadgets, we have that
 \begin{align*}
  v_x^{(n_1+n_2+1)} = \dots = v_x^{(n_1+n_2+n_3)} =
  \begin{cases}
   (1,0) &\text{in case \ref{item:case-bp} or \ref{item:case-rn}}\\
   (0,1) &\text{in case \ref{item:case-rp} or \ref{item:case-bn}}
  \end{cases}
 \end{align*}
 So overall
 \begin{align*}
  \sum_{j=1}^{2n^2}v_x^{(j)} &=
   \begin{cases}
     n_1(1,0) + n_2(0,1) + n_3(0,1) = (n^2, n^2-\alpha+\alpha) = (n^2,n^2) &\text{in case \ref{item:case-rp}}\\
     n_1(1,0) + n_2(0,1) + n_3(1,0) = (n^2+\alpha-\alpha, n^2) = (n^2,n^2) &\text{in case \ref{item:case-bp}}\\
     n_1(1,0) + n_2(0,1) + n_3(1,0) = (n^2+\alpha, n^2-\alpha)             &\text{in case \ref{item:case-rn}}\\
     n_1(1,0) + n_2(0,1) + n_3(0,1) = (n^2+\alpha, n^2-\alpha)             &\text{in case \ref{item:case-bn}}
   \end{cases}\\
   &= \begin{cases}
       (n^2,n^2)                &\text{if } \chiin(s) = \red\\
       (n^2+\alpha, n^2-\alpha) &\text{if } \chiin(s) = \blue
      \end{cases}
 \end{align*}
 as desired. 
 
 \medskip
 
 For the existence guarantee suppose that $(\chiin, \chiout, v_x) \in \CRel$ and $S = (S_1,S_2)$ with $S_1 \subseteq \chiin(I_1),  S_2 \subseteq \chiin(I_2)$ and $(|S_1|, |S_2|) = v_x$.
 We construct a square $q$-coloring $\chi$ by extending $\chiin$ and $\chiout$.
 First, for each of the socket gadgets $G^{(j)}$, we color both their inputs $\Vin^{(j)}$ and their outputs $\Vout^{(j)}$ using the corresponding colors of $\chiin$.

 We then extend the coloring to the inner vertices of the equality gadgets connecting these inputs and outputs using the existence guarantees for the equality gadgets.
 
 The vector output $v_x^{(j)} \in \{(1,0), (0,1)\}$ of each socket gadget $G^{(j)}$, $j \in [2n^2]$, is determined by the coloring of its input vertices which we already specified.
 Let $U_1 \coloneqq \{ j \in [2n^2] \mid v_x^{(j)} = (1,0)\}$ and $U_2 \coloneqq \{ j \in [2n^2] \mid v_x^{(j)} = (0,1)\}$.
 Using the same calculations as above, we get that $v_x = (|U_1|,|U_2|)$.

 Recall that $v_x = (|S_1|, |S_2|)$.
 Hence, we can choose two arbitrary bijections $f_1\colon U_1 \rightarrow S_1$ and $f_2\colon U_2 \rightarrow S_2$.
 We set $\chi(z^{(j)}) \coloneqq f_1(j) \in \chi(I_1)$ if $j \in U_1$ and $\chi(z^{(j)}) \coloneqq f_2(j) \in \chi(I_2)$ if $j \in U_2$.
 Then $\chi(\Vx) \cap \chiin(I_i) = S_i$ for both $i \in \{1,2\}$.
 We complete the coloring $\chi$ by using the existence guarantee of the socket gadgets to extend $\chi$ to the inner vertices of the socket gadgets.
 
 Note that no colors conflict can arise between inner vertices of different subgadgets by Property~\ref{prop:only-colorless-vertices-adjacent-to-portals}.
\end{proof}

\paragraph{The vector state gadget.}
Next, we construct the \emph{vector state gadget} which combines three edge selection gadgets that act on different color classes, but are controlled by the same control input.
It is described in Table~\ref{tbl:summary-vector-state}.

\begin{table}[h]
 \centering
 \begin{tabular}{ |p{2cm}|p{12cm}| }
  \hline \multicolumn{2}{| p{14cm} |}{Vector State Gadget}\\ \hline
  
  parameters & $n \in \ZZp$ and a vector $y = (y_1, y_2, y_3) \in [-n^2, n^2]^3$\\
  
  \hline
  
  $\Vin$ &
  \vspace{-0.2cm}
  \begin{itemize}[noitemsep, nolistsep]
   \item six groups of $\targetval$ vertices $I_1,I_2,I_3,I_4,I_5,I_6$ (the color class inputs)
   \item three ``logic'' input vertices $r,g,b$
   \item and a ``control'' input vertex $s$
  \end{itemize}
  \vspace{-0.4cm}
  
  \\
  \hline
  
  $\Vout$ &
  \vspace{-0.2cm}
  \begin{itemize}[noitemsep, nolistsep]
   \item six groups of $\targetval$ vertices $I'_{1}, I'_{2}, I'_{3}, I'_{4}, I'_{5},I'_{6}$
   \item three ``logic'' output vertices $r',g',b'$
   \item a ``control'' output vertex $s'$
  \end{itemize}
  \vspace{-0.4cm}
  
  \\ \hline
  
  $\Vx$ & a total of $6n^2$ vertices\\
  
  \hline

  $\Cin$ &
  \vspace{-0.2cm}
  \begin{itemize}[noitemsep, nolistsep]
   \item
    $|\chiin(I_{1} \cup I_{2} \cup I_{3} \cup I_{4} \cup I_{5} \cup I_{6} \cup \{r,g,b\})| = $

    $|I_{1} \cup I_{2} \cup I_{3} \cup I_{4} \cup I_{5} \cup I_{6} \cup \{r,g,b\}|$
   \item $\chiin(r) = \red, \chiin(g) = \green, \chiin(b) = \blue$
   \item $\chiin(s) \in \{\red, \blue\}$
  \end{itemize}
  \vspace{-0.4cm}

  \\ \hline
  
  $\CRel$ & We divide $\CRel$ into output and vector conditions:
  \begin{enumerate}[noitemsep, nolistsep, label = (\arabic*)]
  \item\label{item:vector-state-output-condition}
   \begin{itemize}[noitemsep, nolistsep]
    \item
      $\chiout(r') = \chiin(r)$,
      $\chiout(g') = \chiin(g)$,
      $\chiout(b') = \chiin(b)$
    \item $\chiout(s') = \chiin(s)$
    \item $\chiout(I'_i) = \chiin(I_i)$ for all $i \in [6]$
   \end{itemize}
  \item\label{item:vector-state-vector-condition}
   \begin{itemize}[noitemsep, nolistsep]
    \item $v_x = (n^2, n^2, n^2, n^2, n^2, n^2)$ if $\chiin(s) = \red$
    \item $v_x = (n^2+y_1, n^2-y_1, n^2+y_2, n^2-y_2, n^2+y_3, n^2-y_3)$ if $\chiin(s) = \blue$
   \end{itemize}
  \end{enumerate}
  \vspace{-0.4cm}

  \\ \hline
  
 \end{tabular}
 \caption{The specification table of the vector state gadget}
 \label{tbl:summary-vector-state}
\end{table}

The construction of the vector state gadget $\VectorState(y)$ is depicted in Figure~\ref{fig:vector-state}.
The gadget consists of three edge selection gadgets $G^{(1)}, G^{(2)}, G^{(3)}$.
The parameter of $G^{(j)}$ is $\alpha_j \coloneqq y_j$ for each $j \in [3]$.
As before, we identify the elements (e.g., sets and vertices) of $G^{(j)}$ by superscripting them with $(j)$.

\begin{figure}
 \centering
 \includesvg[width=0.95\linewidth]{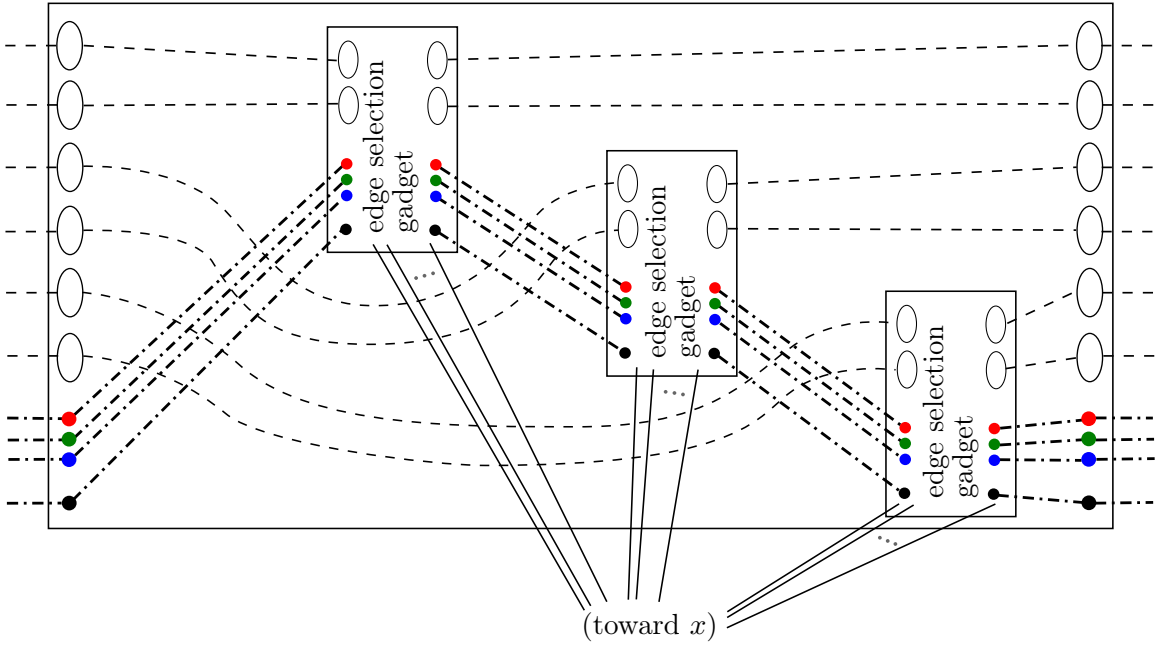}
 \caption{Illustration of a vector state gadget.}
 \label{fig:vector-state}
\end{figure}

The overall wiring of the vector state gadget is simple.
For every $j \in [3]$ we connect $I_{2j-1}$ to $I_1^{(j)}$ and $I_{2j}$ to $I_2^{(j)}$ using $(\targetval)$-equality gadgets.
Similarly, we connect $(I_1')^{(j)}$ to $I_{2j-1}'$ and $(I_2')^{(j)}$ to $I_{2j}'$.
The three logic colors and the control input are copied from the inputs of $\VectorState(y)$ to the inputs of $G^{(1)}$,
then from the outputs of $G^{(1)}$ to the inputs of $G^{(2)}$,
from the outputs of $G^{(2)}$ to the inputs of $G^{(3)}$, and finally from the outputs of $G^{(3)}$ to the outputs of $\VectorState(y)$.

Finally, we set $\Vx = \bigcup_{j \in [2n^2]}\Vx^{(j)}$.
Observe that the vector state gadget satisfies Property \ref{prop:only-colorless-vertices-adjacent-to-portals} since the subset gadget and the edge selection gadget satisfy Property \ref{prop:only-colorless-vertices-adjacent-to-portals}.

\begin{lemma}
 The gadget $\VectorState(y)$ satisfies the specification table of the vector state gadget (Table~\ref{tbl:summary-vector-state}).
\end{lemma}

\begin{proof}
 We first show the output guarantee.
 Let $\chi$ be a square $q$-coloring of $\VectorState(y)$ such that $\chi|_{\Vin} \in \Cin$.
 
 Using the output guarantee of the equality gadgets and edge selection gadgets, we get that
 $\chi(r^{(j)}) = \chi(r') = \red$, $\chi(b^{(j)}) = \chi(b') = \blue$, $\chi(g^{(j)}) = \chi(g') = \green$, $\chi(s^{(j)}) =  \chi(s') = \chi(s)$,
 $\chi(I_1^{(j)}) = \chi(I_{2j-1}') = \chi(I_{2j-1})$ and $\chi(I_2^{(j)}) = \chi(I_{2j}') = \chi(I_{2j})$ for all $j \in [3]$.
 In particular, Condition \ref{item:vector-state-output-condition} is satisfied.
 
 Since $\Vx = \bigcup_{j \in [3]}\Vx^{(j)}$ it follows that vector output of the vector state gadget is $v_x = \sum_{j=1}^{3}v_x^{(j)}$.
 Using the output guarantee of the edge selection gadget, we conclude that
 $v_x = \sum_{j=1}^{3}v_x^{(j)} = (n^2,n^2,0,0,0,0) + (0,0,n^2,n^2,0,0) + (0,0,0,0,n^2,n^2) = (n^2,n^2,n^2,n^2,n^2,n^2)$ if $\chi(s) = \red$, and
 $v_x = \sum_{j=1}^{3}v_x^{(j)} = (n^2+y_1,n^2-y_1,0,0,0,0) + (0,0,n^2+y_2,n^2-y_2,0,0) + (0,0,0,0,n^2+y_3,n^2-y_3) = (n^2+y_1, n^2-y_1, n^2+y_2, n^2-y_2, n^2+y_3, n^2-y_3)$ if $\chi(s) = \blue$.
 
 So overall, $(\chi|_{\Vin}, \chi|_{\Vout}, p(\chi|_{\Vx})) \in \CRel$ as desired.
 
 \medskip
 
 For the existence guarantee suppose that $(\chiin, \chiout, v_x) \in \CRel$ and $S=(S_1,S_2,S_3,S_4,S_5,S_6)$ with $S_1 \subseteq \chiin(I_1),\dots,S_6 \subseteq \chiin(I_6)$ and $(|S_1|,|S_2|,|S_3|,|S_4|,|S_5|,|S_6|) = v_x$.
 We construct a square $q$-coloring $\chi$ of $\VectorState(y)$ by extending $\chiin$ and $\chiout$ as follows.

 For all $j \in [3]$, we color the inputs and outputs $\Vin^{(j)}, \Vout^{(j)}$ of $G^{(j)}$ in the natural way.
 Next, we use the existence guarantee of each of the equality gadgets contained in the vector state gadget to obtain a coloring of their inner vertices.
 All that remains is to find a coloring for each of the edge selection gadgets.
 To do this, we use the existence guarantees of the gadget $G^{(j)}$ for $j \in [3]$ with $S^{(j)} = (S_{2j-1},S_{2j})$ and use these colorings to complete $\chi$.
 It follows that $\chi(\Vx) \cap \chiin(I_j) = S_j$ for all $j \in [6]$.
 
 Observe that no color conflict can arise between inner vertices of different subgadgets due to Property \ref{prop:only-colorless-vertices-adjacent-to-portals}.
\end{proof}

\paragraph{The one-way switch gadget.}
We now build another low-level gadget that is needed to build the vector selection gadget.
The \emph{one-way switch gadget} receives and outputs logic colors and a control color and allows the control color to either stay constant or to switch from \blue\ to \red, but not the other way around.
The gadget is specified in Table~\ref{tbl:summary-one-way-switch}.

\begin{table}[h]
\centering
\begin{tabular}{ |p{2cm}|p{12cm}| }
  \hline \multicolumn{2}{| p{14cm} |}{One-Way Switch Gadget}

  \\ \hline

  parameters & none \\

  \hline

  $\Vin$ &
  \vspace{-0.2cm}
  \begin{itemize}[noitemsep, nolistsep]
   \item three ``logic'' input vertices $r,g,b$
   \item a ``control'' input vertex $s$
  \end{itemize}
  \vspace{-0.4cm}
  
  \\ \hline

  $\Vout$ &
  \vspace{-0.2cm}
  \begin{itemize}[noitemsep, nolistsep]
   \item three ``logic'' output vertices $r',g',b'$
   \item a ``control'' output vertex $s'$
  \end{itemize}
  \vspace{-0.4cm}
  
  \\ \hline

  $\Vx$ & $\emptyset$\\

  \hline

  $\Cin$ &
  \vspace{-0.2cm}
  \begin{itemize}[noitemsep, nolistsep]
   \item $\chiin(r) = \red, \chiin(g) = \green, \chiin(b) = \blue$
   \item $\chiin(s) \in \{\red, \blue\}$
  \end{itemize}
  \vspace{-0.4cm}

  \\ \hline

  $\CRel$ & We group the conditions of $\CRel$ into two parts:
  \begin{enumerate}[noitemsep, nolistsep, label = (\arabic*)]
   \item\label{item:ows-copy-condition}
    $\chiout(r') = \chiin(r)$,
    $\chiout(g') = \chiin(g)$ and
    $\chiout(b') = \chiin(b)$
   \item\label{item:ows-switch-condition} the control output $\chiout(s')$ satisfies
    \begin{itemize}[noitemsep, nolistsep]
     \item $\chiout(s') = \red$ if $\chiin(s) = \red$
     \item $\chiout(s') \in \{\red, \blue\}$ if $\chiin(s) = \blue$
    \end{itemize}
  \end{enumerate}
  \vspace{-0.4cm}

  \\ \hline
  
\end{tabular}
\caption{The specification table of the one-way-switch
gadget\label{tbl:summary-one-way-switch}}
\end{table}

We construct the one-way switch gadget $\OWSGadget$ as follows (see Figure~\ref{fig:one-way-switch}).
As usual, we connect the input vertices $r,g,b$ to the corresponding output vertices $r',g',b'$ via $1$-equality gadgets.
Next, we create four fresh vertices $t,u,v$ and $w$.
We restrict $t,u,v,w$ to only use subsets of the logic colors by connecting them via a $1$-equality gadget to the output of subset gadgets copying colors from the input logic colors.
Specifically, $t$ and $w$ are restricted to \red\ and \blue, $u$ is restricted to \green\ and \red\ and $v$ is restricted to \red, \blue\ or \green.

\begin{figure}
 \centering
 \includesvg[width=0.6\linewidth]{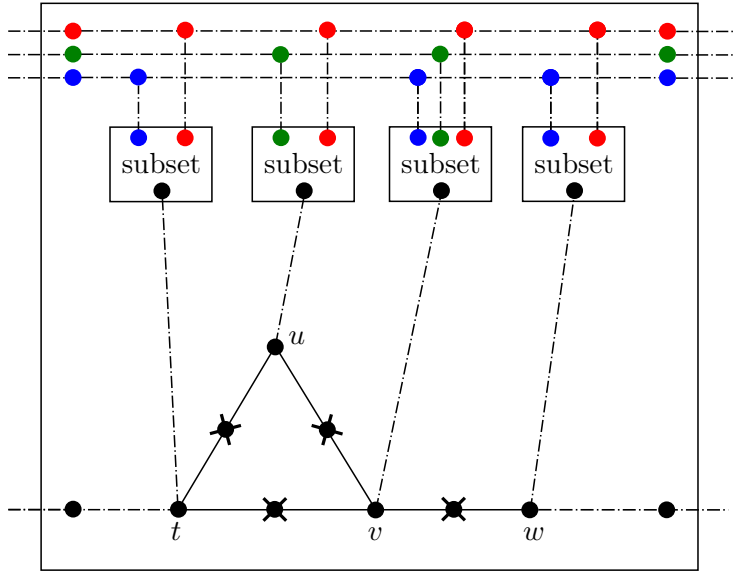}
 \caption{Illustration of a one-way switch gadget where crossed-out vertices are colorless.}
 \label{fig:one-way-switch}
\end{figure}

The vertex $t$ is connected to the input vertex $s$ via a $1$-equality gadget, while the control output vertex $s'$ copies the color of $v$ via a $1$-equality gadget.
We connect the vertices $t,u,v,w$ via subdivided edges (i.e., edges that have a colorless vertex in the middle) so that $t$, $u$ and $v$ form a triangle, and $v$ is connected to $w$.

As usual, observe that the one-way switch gadget satisfies Property \ref{prop:only-colorless-vertices-adjacent-to-portals} since the subset gadget satisfies Property \ref{prop:only-colorless-vertices-adjacent-to-portals}.

\begin{lemma}
 The gadget $\OWSGadget$ satisfies the specification table of the one-way switch gadget (Table~\ref{tbl:summary-one-way-switch}).
\end{lemma}

\begin{proof}
 We first show the output guarantee.
 Let $\chi$ be a square $q$-coloring of $\OWSGadget$ such that $\chi|_{\Vin} \in \Cin$.
 
 Condition \ref{item:ows-copy-condition} of $\CRel$ is clearly satisfied due to the equality gadgets.
 Also, $\chi(s') = \chi(w) \in \{\red,\blue\}$ due to the subset constraint imposed on $w$.
 So if $\chi(s) = \blue$ we are done.
 
 Suppose $\chi(s) = \red$. Then $\chi(t) = \red$.
 Since $\dist(t,u) \leq 2$ and $u$ is restricted to \red\ or \green, we get that $\chi(u) = \green$.
 Also, $\dist(v,u) \leq 2$ and $\dist(v,t) \leq 2$ and $v$ is restricted to colors \red, \green\ or \blue.
 So $\chi(v) = \blue$.
 Finally, $\chi(w) \in \{\red,\blue\}$ and $\chi(w) \neq \chi(v)$ which implies that $\chi(s') = \chi(w) = \red$ as desired.
 
 \medskip
 
 For the existence guarantee suppose that $(\chiin, \chiout, ()) \in \CRel$.
 We construct the square $q$-coloring $\chi$ of $\OWSGadget$ as follows.
 
 First suppose that $\chiin(s) = \red$ which implies that $\chiout(s') = \red$.
 We set $\chi(s) = \chi(t) = \chi(w) = \chi(s') \coloneqq \red$, $\chi(u) \coloneqq \green$ and $\chi(u) \coloneqq \blue$.
 We use the existence guarantees of the subset gadgets to extend the coloring in the natural way.

 Next, suppose $\chiin(s) = \blue$.
 This means $\chiout(s') \in \{\red, \blue\}$.
 First assume $\chiout(s') = \red$.
 We set $\chi(u) = \chi(w) = \chi(s') \coloneqq \red$, $\chi(v) \coloneqq \green$ and $\chi(s) = \chi(t) \coloneqq \blue$.
 
 Otherwise, $\chiout(s') = \blue$.
 We set $\chi(u) \coloneqq \red$, $\chi(v) \coloneqq \green$ and $\chi(s) = \chi(t) = \chi(w) = \chi(s') \coloneqq \blue$.
 
 In both subcases, we again use the existence guarantees of the subset gadgets to extend the coloring in the natural way.
 Observe that no color conflict can arise between inner vertices of different subgadgets due to Property~\ref{prop:only-colorless-vertices-adjacent-to-portals}.
\end{proof}

\paragraph{The vector selection gadget.}
Finally, we are ready to construct the vector selection gadget.
Recall the definition of a vector-generated output and the specification of the vector selection gadget in Table~\ref{tbl:summary-vector-selection}.
Suppose $A = \{a_1,\dots,a_{n^4}\}$ is the input vector list.
Recall that $A$ is a node-representing vector list.
Suppose $D = \{z_1,z_2,z_3\}$ denotes the set of non-zero dimensions of $A$.

We construct the vector selection gadget $\VectorSelection(A)$ as follows (see also Figure~\ref{fig:vector-selection}).
The overall structure of the vector selection gadget consists of a chain of vector state gadgets with one-way switch gadgets in-between.
More precisely, we create $n^4$ vector state gadgets $G^{(\text{vst},1)},\dots,G^{(\text{vst},n^4)}$ (the parameters are specified later) and $n^4-1$ one-way switch gadgets $G^{(\text{ows},1)},\dots,G^{(\text{ows},n^4-1)}$.
As before, we identify the elements (e.g., sets and vertices) of $G^{(\text{vst},j)}$ and $G^{(\text{ows},j)}$ by superscripting them with $(\text{vst},j)$ and $(\text{ows},j)$, respectively.

\begin{figure}
 \centering
 \includesvg[width=.95\linewidth]{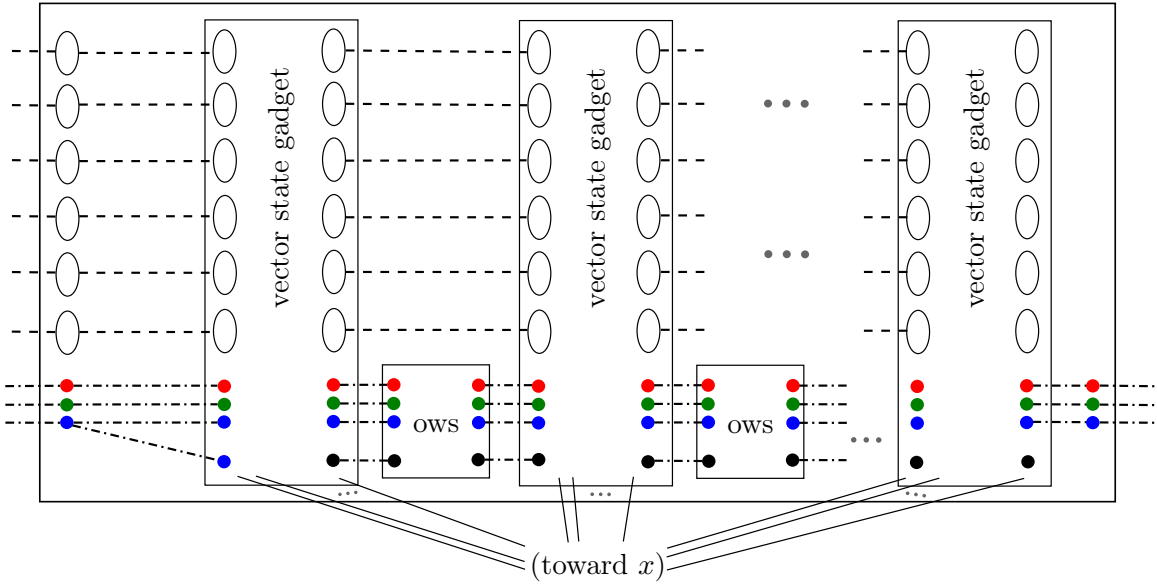}
 \caption{Illustration of the vector selection gadget.}
 \label{fig:vector-selection}
\end{figure}

Now, we connect the following sets and vertices using appropriate equality gadgets.
First, we connect the input vertex $b$ to $s^{(\text{vst},1)}$.
For all $i \in [6]$, we connect $I_i$ to $I_i^{(\text{vst},1)}$.
Also, we connect $r,g$ and $b$ to $r^{(\text{vst},1)},g^{(\text{vst},1)}$ and $b^{(\text{vst},1)}$, respectively.
For every $j \in [n^4-1]$, we connect $I_i^{(\text{vst},j)}$ to $I_i^{(\text{vst},j+1)}$ for all $i \in [6]$.
Furthermore, we connect $(r')^{(\text{vst},j)}, (g')^{(\text{vst},j)}, (b')^{(\text{vst},j)}$ and $(s')^{(\text{vst},i)}$ to $r^{(\text{ows},j)}, g^{(\text{ows},j)}, b^{(\text{ows},j)}$ and $s^{(\text{ows},j)}$, respectively.
Next, we connect $(r')^{(\text{ows},j)}, (g')^{(\text{ows},j)}, (b')^{(\text{ows},j)}$ and $(s')^{(\text{ows},j)}$ to $r^{(\text{vst},j+1)}, g^{(\text{vst},j+1)},b^{(\text{vst},j+1)}$ and $s^{(\text{vst},j+1)}$, respectively.
Finally, we connect $(r')^{(\text{vst},n^4)},(g')^{(\text{vst},n^4)}$ and $(b')^{(\text{vst},n^4)}$ to the logic color outputs $r',g',b'$, respectively.
The control output $(s')^{(\text{vst},n^4)}$ and the color class outputs $I_j^{(\text{vst},n^4)}$, $j \in [6]$, are left unconnected.

Note that the output guarantees of the one-way switch and vector state gadgets ensure that each one-way switch gadget outputs either \red\ or \blue\ on its control output,
and that the vector state gadgets output the same color on their control output as they receive on their control input.
For every square $q$-coloring $\chi$ of the gadget, let us call $\chi(s^{(\text{vst},j)})$ (for $j \in [n^4]$) the \emph{chain state at position $j$}.
Due to the behaviour of the one-way switch gadgets, we conclude that there exists some $t \in [n^4]$ such that for all $t' \leq t$ the chain state at position $t'$ is \blue, and for all $t' > t$ the chain state at position $t'$ is \red.
Observe that there are $n^4$ possible positions $t$ for this chain state switch from \blue\ to \red\ to occur in the chain.
These represent the $n^4$ possible choices for $a_t \in A$.
More precisely, if the chain state switch happens at position $t$, the vector selection gadget should output the vector
\[v_x = (\halftargetval+\NZ(a_t)[1], \halftargetval-\NZ(a_t)[1],\halftargetval+\NZ(a_t)[2], \halftargetval-\NZ(a_t)[2],\halftargetval+\NZ(a_t)[3], \halftargetval-\NZ(a_t)[3]),\]
i.e., the vector-generated output generated by $a_t$.

To complete the construction, we specify the parameters of the vector state gadgets.
Let
\[R_i \coloneqq \begin{cases}
                 [1,n^2]   &\text{if } z_i \in D^+ \\
                 [-n^2,-1] &\text{if } z_i \in D^-
                \end{cases}\]
for $i \in [3]$.
Recall that $\NZ(a_j) \in R_1 \times R_2 \times R_3$ for all $j \in [n^4]$.
We define the parameter $y^{(\text{vst},j)}$ of the vector state gadget $G^{(\text{vst},j)}$ as follows.
We set
\[y^{(\text{vst},j)} \coloneqq \NZ(a_1)\] and
\[y^{(\text{vst},j)} \coloneqq \NZ(a_j) - \NZ(a_{j-1})\]
for all $j \in \{2,\dots,n^4\}$.
Since $\NZ(a_j)[i] \in R_i$ for all $j \in [n^4]$ and $i \in [3]$, it follows that $y^{(\text{vst},j)} \in [-n^2, n^2]^3$ for all $j \in [n^4]$.
The reason for this choice of parameters is that $\sum_{j=1}^t y^{(\text{vst},j)} = \NZ(a_t)$ for all $t \in [n^4]$.

Finally, we set $\Vx \coloneqq \bigcup_{j \in [n^4]} \Vx^{(\text{vst},j)}$.
This completes the construction of $\VectorSelection(A)$.

As usual, observe that the vector selection gadget satisfies Property \ref{prop:only-colorless-vertices-adjacent-to-portals} since the subset gadget satisfies Property \ref{prop:only-colorless-vertices-adjacent-to-portals}.

\begin{lemma}
 The gadget $\VectorSelection(A)$ satisfies the specification table of the vector selection gadget (Table~\ref{tbl:summary-vector-selection}).
\end{lemma}

\begin{proof}
 Let $\Vin^{(\text{vst})} \coloneqq \bigcup_{j=1}^{n^4}\Vin^{(\text{vst},j)}$.
 We say that a coloring $\chiin^{\text{vst}}\colon \Vin^{(\text{vst})} \to [q]$ is \emph{compliant} if the following conditions are satisfied:
 \begin{enumerate}[label = (\Roman*)]
  \item\label{item:compliant-1} $\chiin^{(\text{vst})}(r^{(\text{vst},j)}) = \red$,
   $\chiin^{(\text{vst})}(b^{(\text{vst},j)}) = \blue$ and
   $\chiin^{(\text{vst})}(g^{(\text{vst},j)}) = \green$ for all $j \in [n^4]$,
  \item\label{item:compliant-2} $\chiin^{(\text{vst})}(I_i^{(\text{vst},j)}) = \chiin^{(\text{vst})}(I_i^{(\text{vst},j')})$ for all $j,j' \in [n^4]$ and $i \in [6]$, and
  \item\label{item:compliant-3} there is some $t \in [n^4]$ such that $\chiin^{(\text{vst})}(s^{(\text{vst},j)}) = \blue$ for all $j \leq t$ and
   $\chiin^{(\text{vst})}(s^{(\text{vst},t')}) = \red$ for all $j > t$.
 \end{enumerate}
 Furthermore, we say that it \emph{inherits} a coloring $\chiin\colon \Vin \to [q]$ with $\chiin \in \Cin$ if
 \[\chiin^{(\text{vst})}(I_i^{(\text{vst},1)}) = \chiin(I_i)\]
 for all $i \in [6]$.
 
 \medskip
 
 We first show the output guarantee.
 Let $\chi$ be a square $q$-coloring of $\VectorSelection(A)$ such that $\chi|_{\Vin} \in \Cin$.
 It is easy to see that $\chi(r') = \chi(r)$, $\chi(g') = \chi(g)$ and $\chi(b') = \chi(b)$.
 
 Let $v_x \coloneqq p(\chi_|{\Vx})$.
 We need to show that $v_x$ is a vector-generated output of $A$.
 We first argue that $\chi|_{\Vin^{(\text{vst})}}$ is compliant.
 Using the output guarantees of all involved subgadgets it is easy to see that Conditions \ref{item:compliant-1} and \ref{item:compliant-2} are satisfied.
 Also, Condition \ref{item:compliant-3} follows from the output guarantee of the one-way switch gadget.
 Moreover, it is also immediately clear that $\chi|_{\Vin^{(\text{vst})}}$ inherits $\chi|_{\Vin}$.
 
 Now let $v_x^{(\text{vst},j)}$ denote the vector output of the vector state gadget $G^{(\text{vst},j)}$.
 Since $V_x = \bigcup_{j \in [n^4]}V_x^{(\text{vst},j)}$, we conclude that
 \[v_x = \sum_{j=1}^{n^4}v_x^{(\text{vst},j)}.\]
 Let us write $y_j \coloneqq y^{(\text{vst},j)}$ for simplicity.
 Using the output guarantee of the vector state gadget and defining $t$ to be the position of the chain state switch, we get that
 \begin{align*}
  v_x &= \sum_{j=1}^{n^4}v_x^{(\text{vst},j)}\\
      &= \sum_{j=1}^{t}v_x^{(\text{vst},j)} + \sum_{j=t+1}^{n^4}v_x^{(\text{vst},j)}\\
      &= \sum_{j=1}^{t}(n^2+y_j[1],n^2-y_j[1],
                        n^2+y_j[2],n^2-y_j[2], 
                        n^2+y_j[3],n^2-y_j[3])\\
      & \;\;\;\;+\sum_{j=t+1}^{n^4}(n^2,n^2,n^2,n^2,n^2,n^2)\\
      &= \left(\sum_{j=1}^{t}y_j[1],
               -\sum_{j=1}^{t}y_j[1],
               \sum_{j=1}^{t}y_j[2],
               -\sum_{j=1}^{t}y_j[2],
               \sum_{j=1}^{t}y_j[3],
               -\sum_{j=1}^{t}y_j[3]\right)\\
      & \;\;\;\;+ t(n^2,n^2,n^2,n^2,n^2,n^2) + (n^4-t)(n^2,n^2,n^2,n^2,n^2,n^2)\\
      &= \left(\NZ(a_{t})[1],-\NZ(a_{t})[1],\NZ(a_{t})[2],-\NZ(a_{t})[2],\NZ(a_{t})[3],-\NZ(a_{t})[3]\right)\\
      & \;\;\;\;+ (n^6,n^6,n^6,n^6,n^6,n^6)
 \end{align*}
 So $v_x$ is vector-generated output of $A$ which is generated by $a_t$.
 Overall, this means that $(\chi|_{\Vin}, \chi|_{\Vout}, p(\chi|_{\Vx})) \in \CRel$.
 
 \medskip
 
 For the existence guarantee suppose that $(\chiin, \chiout, v_x) \in \CRel$ and $S=(S_1,\dots,S_6)$ with $S_1 \subseteq \chiin(I_1),\dots,S_6 \subseteq \chiin(I_6)$ and $(|S_1|,\dots,|S_6|) = v_x$.
 Assume that $v_x$ is generated by $a_t$.
 We construct a square $q$-coloring $\chi$ of $\VectorSelection(A)$ by extending $\chiin$ and $\chiout$ as follows.
 
 First, there is a unique coloring $\chiin^{(\text{vst})}$ that is compliant, has the chain state switch at position $t$ and inherits $\chiin$.
 We define $\chi|_{\Vin^{(\text{vst})}} \coloneqq \chiin^{(\text{vst})}$.
 We use the existence guarantees of all subgadgets except the vector state gadgets to extend the coloring $\chi$ constructed so far.
 It is easy to see that this is always possible.

 Now, consider the vector state gadget $G^{(\text{vst},j)}$.
 Let us denote by $v_x^{(\text{vst},j)}$ the vector output of $G^{(\text{vst},j)}$ which is uniquely determined by $\chi|_{\Vin^{(\text{vst})}}$.
 We wish to find colorings for the vector state gadgets resulting in this vector output.

 Towards this end, we arbitrarily choose sets of colors $S^{(j)}_i$ for all $j \in [n^4]$ and $i \in [6]$ such that
 \begin{itemize}
  \item $\left(|S^{(j)}_1|,\dots,|S^{(j)}_6|\right) = v_x^{(\text{vst},j)}$ for all $j \in n^4$, and
  \item $\bigcup_{j \in [n^4]} S^{(j)}_i = S_i$ for all $i \in [6]$.
 \end{itemize}
 Now we can use the existence guarantees of the gadget $G^{(\text{vst},j)}$ on the tuple $S^{(j)} = (S^{(j)}_1,\dots,S^{(j)}_6)$ to color the subgadget $G^{(\text{vst},j)}$.
 We have that $\chi(V_x) \cap \chiin(I_i) = \bigcup_{j\in[n^4]}\chi(V_x^{(j)}) \cap \chiin(I_i) = \bigcup_{j\in[n^4]}S^{(j)}_i = S_i$ as desired.

 As usual, observe that no color conflict can arise between inner vertices of different subgadgets due to Property~\ref{prop:only-colorless-vertices-adjacent-to-portals}.
\end{proof}

\subsubsection{Removing Colorless Vertices}
So far, the reduction relies on colorless vertices.
We argue how to modify the construction to replace all colorless vertices by normal vertices.
Suppose $G^*$ denotes the output graph of the reduction and let $V_{\text{colorless}}$ denote the set of colorless vertices in $G^*$.
Also let $n_{\text{colorless}} \coloneqq V_{\text{colorless}}$ denote the number of colorless vertices.
We set $q_{\text{colorless}} \coloneqq n_{\text{colorless}}$, i.e., the number of neutral colors equals the total number of colorless vertices.

Now, in order to replace all colorless vertices by normal vertices, we only need to modify the subset gadget. 
We reduce the number of complement vertices by two --
meaning that the total number of complement vertices is now $q-\alpha-2$ instead of $q-\alpha$ --
and add an edge between the two colorless vertices $a$ and $b$.

It is easy to see that the modified gadget without colorless vertices still enforces a subset constraint between the input and output colors.
Indeed, all complement vertices, as well as both previously colorless vertices, are now within distance $2$ of the $\alpha$ input vertices.
They are also pairwise within distance $2$ of each other.
Hence, these $q-\alpha-2+2$ vertices are colored using exactly the $q-\alpha$ colors which do not appear on the input vertices.
Similarly, the output vertices are within distance $2$ of all complement vertices and the two previously colorless vertices,
So they can not be colored using any of those colors.
Hence, their colors are a subset of the colors used for the input vertices (and all their colors are pairwise distinct).
So overall, the subset gadget still behaves in the same way.

Now, given a coloring $\chi$ of the constructed graph with colorless, we obtain a coloring $\chi'$ of the updated graph by coloring each vertex $v \in V_{\text{colorless}}$ using a distinct neutral color.
It is easy to see that $\chi'$ of the updated instance since colors of vertices in $V_{\text{colorless}}$ appear only once, and hence, they can not introduce any color conflicts.

In the other direction, it can be checked that for all gadgets, increasing the number of colors does not change the possible colorings (assuming appropriate subset gadgets are used) since all normal vertices are forced to be colored with non-neutral colors via appropriate subset gadgets.
Hence, the neutral colors can only appear on vertices from $V_{\text{colorless}}$ which implies a coloring $\chi'$ of the updated graph restricts to a coloring of the original graph $G^*$ with colorless vertices.

\subsubsection{Properties of the Reduction}

We complete the proof of Theorem \ref{thm:vector-sum-to-distance-two-coloring} by analysing the number of vertices and the treewidth of the constructed instance $G^*$.

\paragraph{Number of vertices and running time.}
We first bound the number of vertices of $G^*$.
The result depends on $q_{\text{colorless}}$, since each subset gadget has $q-\alpha-2$ complement vertices.
Recall that in our construction, $q_{\text{colorless}}$ is the number of colorless vertices in the construction.
Hence, by also bounding the number of colorless vertices, we can get a bound on the total number of vertices of the graph that does not depend on $q_{\text{colorless}}$.

Discounting its inputs and outputs, a subset gadget has $O(q)$ vertices, of which $2$ are colorless.
The color class copy gadget, which consists of the vertices $\Vin^{(\text{cpy})}$, $\Vout^{(\text{cpy})}$, the subset gadget connecting the two as well as the vertex set $S$,
has $O(m\cdot n^6) + O(m\cdot n^6) + O(q) + O(\log m) = O(m\cdot n^6 + q)$ vertices in total.
Also, $O(\log m)$ vertices are colorless.

Both socket gadgets consist of a constant number of subset gadgets where $\alpha + \beta$ is always bounded by $O(n^6)$.
Hence, the total number of vertices is $O(q+n^6)$, and the total number of colorless vertices is $O(1)$.

The edge selection gadget consists of $2n^2$ socket gadgets and $O(n^2)$ subset gadgets connecting them.
So the total number of vertices is $O(n^2(q+n^6))$, and the total number of colorless vertices is $O(n^2)$.

Similarly, the vector state gadget consists of three edge selection gadgets and a constant number of subset gadgets connecting them.
So as before, the total number of vertices is $O(n^2(q+n^6))$, and the total number of colorless vertices is $O(n^2)$.

The one-way switch gadget has a constant number of subset gadgets with $\alpha + \beta = O(1)$ and a constant number of other vertices.
So it has $O(q)$ vertices and $O(1)$ colorless vertices.

The vector selection gadget consists of $O(n^4)$ vector state gadgets and one-way switch gadgets, as well as $O(n^4)$ subset gadgets connecting them.
Hence, we have $O(n^4 \cdot n^2(q+n^6)) = O(n^6(q+n^6))$ vertices total, of which $O(n^4 \cdot n^2) = O(n^6)$ are colorless.

Finally, the graph $G^*$ consists of $k$ vector selection gadgets, one color class copy gadget, $O(k)$ subset gadgets connecting them, and a constant number of additional vertices.
So the total number of vertices of the reduction output is $O(k \cdot n^6(q+n^6) + m \cdot n^6 + q + k \cdot q) = O(kn^6(q+n^6)+mn^6)$, of which
$O(k \cdot n^6 + m\cdot n^6 + k) = O((k+m)n^6)$ are colorless.

Recall that $q = 2m\cdot \targetval + 3 + q_\text{colorless} = 2m \cdot \targetval + 3 + O((k+m)n^6) = O((k+m)n^6)$.
Hence, the total number of vertices is $O(kn^6((k+m)n^6+n^6)+mn^6) = O(k(k+m)n^{12}))$.

Certainly, the construction of the graph can be done in time polynomial in the number of vertices.

\paragraph{Treewidth.}
It remains to bound the treewidth of $G^*$.
We show that 
\[\tw(G^*) = 2r + O(1) = (1+\varepsilon)\log(m) + O(1)\]
where $r$ is the number defined in Equation \eqref{eq:def-r}.
Using Theorem \ref{thm:cops-and-robbers-characterization-of-treewidth}, it suffices to provide a winning strategy for $2r + O(1)$ cops in cops-and-robbers game on $G^*$.

First, we place $2r+4$ cops on the central vertex $x$, on $w_X$, on the $2r$ vertices of $S$ in $G^{(\text{cpy})}$, and on the two formerly colorless vertices in the $(2m\cdot \targetval)$-equality gadget connecting $\Vin^{(\text{cpy})}$ and $\Vout^{(\text{cpy})}$.
These cops never move throughout the play.
Also let $U$ denote the set of those vertices.
If the robber is on one of the complement vertices of the $(2m\cdot \targetval)$-equality gadget, we can catch them immediately with another cop.

We now show that if the robber is not on one of the complement vertices, we can catch them using only a constant number of additional cops.

Note that each $X_j^{(\text{cpy})}$ or $Y_j^{(\text{cpy})}$ is only connected to one vector selection gadget.
Hence, in the graph $G^* - U$, after removing the logic inputs $r^{(\text{sel},i)},g^{(\text{sel},i)},b^{(\text{sel},i)}$ as well as the logic outputs $(r')^{(\text{sel},i)}, (g')^{(\text{sel},i)}, (b')^{(\text{sel},i)}$,
the vector selection gadget $G^{(\text{sel},i)}$ along with the six vertex groups $X_j^{(\text{cpy})}$ or $Y_j^{(\text{cpy})}$ that $G^{(\text{vst},i)}$ is connected to via equality gadgets form a connected component.

Hence, following the chain of vector selection gadgets, we can use $6$ additional cops to trap the robber in one of the vector selection gadgets $G^{(\text{sel},i)}$ (or in an equality connecting the logic inputs and outputs in which case the cops easily win).
The $6$ cops blocking the logic inputs and outputs of the vector selection gadget $G^{(\text{sel},i)}$ remain in position, and we now catch the robber within this vector selection gadget or its connected $X_j^{(\text{cpy})}$ and $Y_j^{(\text{cpy})}$ using a constant number of additional cops.

The vector selection gadget has its $6$ color class inputs (connected to the corresponding six $X_j^{(\text{cpy})}$ and $Y_j^{(\text{cpy})}$) and its $3$ logic color inputs and outputs.
The main body of the gadget is made of the chain of vector state gadgets $G^{(\text{vst},i)}$ and one-way switch gadgets $G^{(\text{ows},i)}$.

We can use $2$ cops to disconnect the inputs and outputs of a subset gadget by placing them on both of the formerly colorless vertices (note that $1$ cop would suffice, but we use $2$ to make the argument below cleaner).
Hence, for any vector state gadget, we can block the equality gadgets that externally connect to its inputs and outputs using a constant number of cops.

We do this input-output blocking for the first vector selection gadget.
Specifically, this blocks the color class inputs $I_j^{(\text{sel},i)}$ from the rest of the vector selection gadget.
The robber is now either on one of the equality gadgets leading to an $X_j^{(\text{cpy})}$ or $Y_j^{(\text{cpy})}$ (which also includes the vertices of the $X_j^{(\text{cpy})}$ or $Y_j^{(\text{cpy})}$), or they are within the vector selection gadget itself.

Assume they are on one of the equality gadgets going to the $X_j^{(\text{cpy})}$ or $Y_j^{(\text{cpy})}$.
We can block the formerly colorless vertices using two cops, and then catch the robber -- who is either on one of the complement vertices or on vertices of the $X_j^{(\text{cpy})}$ or $Y_j^{(\text{cpy})}$ -- using a third cop.

Now assume the robber is within the vector selection gadget.
As before, we can go through the chain of vector state and one-way switch gadgets, always blocking off the inputs and outputs of consecutive vector state gadgets.
This way we trap the robber in one of the subgadgets.
If the robber is trapped in a subgadget, we proceed as before.
We keep the cops in place to keep the robber trapped throughout the remainder of the play, and we follow the same strategy for the next subgadget until we eventually trap the robber in one of the low-level socket gadgets or one-way switch gadgets.
At this point, it is easy to see that the robber can be caught using a constant number of additional cops. 

Overall, this strategy requires $2r + O(1)$ cops which implies the desired bound on the treewidth.
This concludes the proof of Theorem~\ref{thm:vector-sum-to-distance-two-coloring}.

\section{Algorithms for Planar Graphs}
\label{sec:planar-alg}
We now turn to designing algorithms for \SqCol\ on planar graphs.
Let us write \textsc{Planar} \SqCol\ (resp.\ \textsc{Planar} \qSqCol) to refer to the variant of \SqCol\ (resp.\ \qSqCol) where the input graph is planar.
The main result of this section is an algorithm that solves \textsc{Planar} \SqCol\ in subexponential time $2^{O(n^{2/3}\log n)}$.

The algorithm consists of two subroutines, one of which covers the case that the number $q$ of available colors is small, and the second subroutine is used for large numbers of colors.
The first subroutine relies on the fact that we can bound the treewidth of $G^2$ by $O(\sqrt{nq})$ at which point we can rely on standard dynamic programming algorithms for \textsc{$q$-Coloring}.
The second subroutine is more intricate and here, the idea is to construct a $(O(n/q),O(n/q),O(1))$-protrusion decomposition (see, e.g., \cite[Chapter 15]{FominLSZ19}) and then rely on dynamic programming ideas that are similar to those already used in Section \ref{sec:tw-alg}.
This results in an algorithm for \textsc{Planar} \SqCol\ running in time $n^{O(n/q)}$.

Let us start with the first subroutine.
Here, we use the fact that planar graphs of bounded diameter have bounded treewidth.
Let $G$ be a graph.
A \emph{spanning tree} of $G$ is a tree $T$ with vertex $V(G)$ and edge set $E(T) \subseteq E(G)$.
Here, we consider rooted spanning trees where an arbitrary vertex of $T$ is declared to be the root of $T$.
The \emph{height} of $T$ is the maximum distance between the root of $T$ and any other vertex of $T$.

\begin{lemma}[see, e.g., {\cite[Lemma 12.10]{FlumG06}}]
 \label{la:tw-planar-diameter}
 Let $G$ be a planar graph that has a spanning tree of height $\ell$.
 Then
 \[\tw(G) \leq 3\ell.\]
 Moreover, given a planar graph $G$ and a spanning tree of $G$ of height $\ell$, a tree decomposition of $G$ of width at most $3\ell$ can be computed in time $O(\ell\cdot n)$.
\end{lemma}

The next lemma is the key insight for the first subroutine.

\begin{lemma}
 \label{la:tw-planar-square}
 Let $G$ be a planar graph of maximum degree $\Delta$.
 Then
 \[\tw(G^2) = O\!\left(\sqrt{n\Delta}\right).\]
 Moreover, given a planar graph $G$, a tree decomposition of $G^2$ of width $O\!\left(\sqrt{n\Delta}\right)$ can be computed in polynomial time.
\end{lemma}

\begin{proof}
 Since the treewidth of a graph $G$ equals the maximum treewidth of its connected components we may assume without loss of generality that $G$ is connected.
 Fix $r \in V(G)$ to be an arbitrary vertex of $G$.
 We define
 \[D_i \coloneqq \left\{v \in V(G) \;\middle|\; \left\lfloor\frac{\dist_G(v,r)}{2}\right\rfloor = i\right\}\]
 for all $i \in \ZZ$.
 We define 
 \[M \coloneqq \left\lceil\frac{\sqrt{n\Delta}}{\Delta}\right\rceil\]
 and let
 \[L_j \coloneqq \bigcup_{i \equiv j \mod M} D_i\]
 for all $j \in \{0,\dots,M-1\}$.
 Now let $j^* \in \{0,\dots,M-1\}$ such that $|L_{j^*}|$ is minimal.
 Clearly,
 \[|L_{j^*}| \leq \frac{n}{M} \leq \frac{n}{\frac{\sqrt{n\Delta}}{\Delta}} = \sqrt{n\Delta}.\]
 
 \begin{claim}
  Let $i \in \ZZ$ such that $i \equiv j^* \mod M$.
  Then
  \[\tw(G^2[D_{i+1} \cup \dots \cup D_{i + M - 1}]) \leq (\Delta + 1)(4 + 6(M+1)) - 1.\]
  Moreover, a tree decomposition of $G^2[D_{i+1} \cup \dots \cup D_{i + M - 1}]$ of width at most $(\Delta + 1)(4 + 6(M+1)) - 1$ can be computed in polynomial time.
 \end{claim}
 \begin{claimproof}
  For ease of notation let us define $C \coloneqq D_{i+1} \cup \dots \cup D_{i + M - 1}$.
  Consider the sets $R' \coloneqq D_0 \cup \dots \cup D_{i-1}$ and $C' \coloneqq D_i \cup \dots \cup D_{i + M}$.
  Note that $G[R']$ is connected (if $R' \neq \emptyset$).
  Now let $G'$ be the graph obtained from $G$ by deleting all vertices outside of $R' \cup C'$ and contracting $R'$ to a single vertex.
  Performing breath-first search on $G'$ starting at $R'$ results in a spanning tree of height at most $\ell \coloneqq 1 + 2(M+1)$.
  
  So $\tw(G') \leq 3\ell$ by Lemma \ref{la:tw-planar-diameter}.
  Let $(T,\beta')$ be a tree decomposition of $G'$ of width at most $3\ell$.
  We construct a tree decomposition $(T,\beta)$ of $G^2[C]$ as follows (observe that the tree $T$ remains unchanged).
  For each $t \in V(T)$ we define
  \[\beta(t) \coloneqq \left(\bigcup_{w \in \beta'(t) \cap C'} N_G[w]\right) \cap C.\]
  Clearly, $(T,\beta)$ can be computed in polynomial time using the algorithm from Lemma \ref{la:tw-planar-diameter}.
  
  We show that $(T,\beta)$ is indeed a tree decomposition of $G^2[C]$.
  First suppose that $uv \in E(G^2[C])$.
  Then there is some vertex $w \in V(G)$ such that $uw,wv \in E(G)$.
  Since $w \in N_G[C]$ and $N_G[C] \subseteq C'$ we conclude that $w \in C'$.
  In particular, there is some node $t \in V(T)$ such that $w \in \beta'(t)$.
  But then $u,v \in \beta(t)$ by definition.
  
  So let $v \in C$ and let $T_v \coloneqq \{t \in V(T) \mid v \in \beta(t)\}$.
  We have that
  \[T_v = \{t \in V(T) \mid N_G[v] \cap \beta'(t) \neq \emptyset\}.\]
  Since $N_G[v]$ induces a connected subgraph in $G'$ and $(T,\beta')$ is a tree decomposition, we conclude that $T_v$ induces a connected subtree of $T$ (see Observation \ref{obs:td-connected-subtree-from-connected-subgraph}).
  So overall, $(T,\beta)$ is a tree decomposition of $G^2[C]$.
  
  To complete the proof, we bound the width of $(T,\beta)$.
  Let $t \in V(T)$.
  We have that
  \[|\beta(t)| \leq (\Delta + 1) \cdot |\beta'(t)| \leq (\Delta + 1)(3\ell + 1) = (\Delta + 1)(4 + 6(M+1)).\qedhere\]
 \end{claimproof}

 Now, we construct a tree decomposition $(T,\beta)$ of $G$ as follows.
 Let $i_1,\dots,i_k$ be the list of all indices $i \in \ZZ$ such that $i \equiv j^* \mod M$ and $C_i \coloneqq D_{i+1} \cup \dots \cup D_{i + M - 1} \neq \emptyset$.
 For every $p \in [k]$ we construct a tree decomposition $(T_p,\beta_p)$ of $G^2[C_{i_p}]$ via the last claim.
 Moreover, let $r_p \in V(T_p)$ denote an arbitrary node which is designated as the root of $T_p$.
 We define $T$ to be the tree obtained from the disjoint union of all trees $T_p$ and a fresh root node $r$ that is connected to the root node $r_p$ for every $p \in [k]$.
 Formally,
 \[V(T) \coloneqq \{r\} \cup \bigcup_{p \in [k]} (\{p\} \times V(T_p))\]
 and
 \[E(T) \coloneqq \{r(p,r_p) \mid p \in [k]\} \cup \bigcup_{p \in [k]} \{(p,t)(p,t') \mid tt' \in E(T_p)\}.\]
 We define
 \[\beta(r) \coloneqq L_{j^*}\]
 and
 \[\beta(p,t) \coloneqq L_{j^*} \cup \beta_p(t)\]
 for all $p \in [k]$ and $t \in V(T_p)$.
 Clearly, $(T,\beta)$ can be computed in polynomial time.
 
 We show that $(T,\beta)$ is indeed a tree decomposition of $G^2$.
 Let $uv \in G^2$.
 We distinguish several cases.
 If $u,v \in L_{j^*}$ then $u,v \in \beta(r)$.
 Otherwise, by symmetry, we may assume that $u \in C_{i_p}$ for some $p \in [k]$.
 Since $\dist_G(u,v) \leq 2$ and $u \in D_{i_p+1} \cup \dots \cup D_{i_p + M - 1}$ it follows that $v \in D_{i_p} \cup \dots \cup D_{i_p + M} \subseteq C_{i_p} \cup L_{j^*}$.
 If $v \in C_{i_p}$ then there is some $t \in V(T_p)$ such that $u,v \in \beta_p(t) \subseteq \beta(p,t)$.
 Otherwise $v \in L_{j^*}$ and $u,v \in \beta(p,t)$ for any $t \in V(T_p)$ such that $u \in \beta_p(t)$.
 
 Also, it is easy to see that for every $v \in V(G)$ the set $T_v \coloneqq \{t \in V(T) \mid v \in \beta(t)\}$ induces a connected subtree of $T$.
 So overall, we conclude that $(T,\beta)$ is a tree decomposition of $G^2$.
 
 Finally, for $t \in V(T)$, we have that
 \[|\beta(t)| \leq |L_{j^*}| + (\Delta + 1)(4 + 6(M+1)) = O\!\left(\sqrt{n\Delta}\right).\qedhere\]
\end{proof}

Now, to obtain an algorithm for \qSqCol, we can combine the last lemma with well-known algorithms for \textsc{$q$-Coloring} on graphs of bounded treewidth.
The following theorem can for example be obtained from \cite[Theorem 7.9 \& 7.18]{CyganFKLMPPS15}

\begin{theorem}
 \label{thm:q-coloring-tw}
 For $q \geq 3$, there is an algorithm that solves \textsc{$q$-Coloring} in time $q^{O(\ttw)} \cdot n^{O(1)}$ on graphs of treewidth at most $\ttw$.
\end{theorem}

\begin{theorem}
 \label{thm:alg-planar-small-degree}
 There is an algorithm that solves \qSqCol\ on planar graphs of maximum degree $\Delta$ in time $q^{O(\sqrt{n\Delta})}$.
\end{theorem}

\begin{proof}
 Let $G$ denote the input graph.
 We compute the graph $G^2$ and apply the algorithm from Theorem \ref{thm:alg-planar-small-degree}.
 By Lemma \ref{la:tw-planar-square}, we have that $\tw(G^2) = O(\sqrt{n\Delta})$.
 So the algorithm from Theorem \ref{thm:alg-planar-small-degree} runs in time $q^{O(\sqrt{n\Delta})}$.
\end{proof}

\begin{corollary}[Lemma \ref{lem:alg-planar-small-degree-intro} restated]
 \label{cor:alg-planar-few-colors}
 There is an algorithm that solves \textsc{Planar} \qSqCol\ in time $q^{O(\sqrt{qn})}$.
\end{corollary}

\begin{proof}
 Let $G$ denote the input graph.
 If $G$ has maximum degree at least $q$ then $G^2$ contains a clique of size at least $q+1$ and the we return NO.
 Otherwise, we run the algorithm from Theorem \ref{thm:alg-planar-small-degree}.
\end{proof}

This completes the first subroutine.
Next, we turn to the second subroutine which covers the case that the number $q$ of available colors is large.
The basic strategy for this case is as follows.
First, we apply some simple reduction rules to remove vertices whose second neighborhood is small (those vertices can always be colored in the end using a greedy strategy).
Afterwards, there need to be many vertices of large degree which allows us to identify a distance-$3$ dominating set $D$ of small size, i.e., every vertex of $G$ is at distance at most $3$ from some vertex in $D$.
Using known techniques \cite{BodlaenderFLPST16}, we use this dominating set to construct a $(O(n/q),O(n/q),O(1))$-protrusion decomposition.
Finally, we use dynamic programming approaches to design an algorithm for \qSqCol\ given the protrusion decomposition.

We start by describing a simple reduction rule.
We say that a graph $G$ is \emph{$q$-irreducible} if $V(G) = U \cup N_G(U)$ where $U \coloneqq \{u \in V(G) \mid |N_{G^2}[u]| > q\}$.
If $G$ is not $q$-irreducible then we can identify a strict induced subgraph of $G$ that is equivalent to $G$ with respect to the problem \qSqCol\ as the next lemma shows.

\begin{lemma}
 \label{la:reduce-graph}
 Let $G$ be a graph and let $q \geq 1$ be an integer.
 Let $U \coloneqq \{u \in V(G) \mid |N_{G^2}[u]| > q\}$ and $W \coloneqq U \cup N_G(U)$.
 Then $G$ has a square $q$-coloring if and only if $G[W]$ has a square $q$-coloring.
\end{lemma}

\begin{proof}
 The forward direction is trivial since any square $q$-coloring of $G$ immediately restricts to a square $q$-coloring of $G[W]$.
 
 So assume that $G[W]$ has a square $q$-coloring $\chi'$, i.e., $\chi'$ is a $q$-coloring of $(G[W])^2$.
 We define a coloring $\chi\colon V(G) \rightarrow [q]$ as follows.
 First, we set $\chi(u) \coloneqq \chi'(u)$ for all $u \in U$.
 Since $N_G(U) \subseteq W$, every path of length at most $2$ between two vertices $u,u' \in U$ in $G$ also exists in $G[W]$.
 This implies that $\chi(u) \neq \chi(u')$ for all distinct $u,u' \in U$ such that $\dist_G(u,u') \leq 2$.
 Now, let $\{v_1,\dots,v_\ell\} = V(G) \setminus U$ and pick $i \in [\ell]$.
 Since $v_i \notin U$ we have that $|N_{G^2}[v_i]| \leq q$ by definition.
 This means there is some color $c_i \in [q]$ that is not used so far by any vertex in $N_{G^2}[v_i]$.
 We set $\chi(v_i) \coloneqq c_i$.
 Doing this for all $i \in [\ell]$, we obtain a square $q$-coloring $\chi$ of $G$.
\end{proof}

Now, as the next intermediate target, we construct a distance-$3$ dominating set of a planar, $q$-irreducible graph.

\begin{lemma} 
 \label{la:dist-2-domset-for-high-degree}
 Let $\ell \in \ZZp$ be a positive integer.
 Also let $G$ be a graph and suppose $U \subseteq V(G)$ such that $|N_{G^2}[u]| \geq \ell$ for all $u \in U$.
 Then there is a set $D \subseteq U$ such that
 \[|D| \leq \frac{2\cdot|E(G)|}{\ell}\]
 and for every $u \in U$ there is some $w \in D$ such that $\dist_G(w,u) \leq 2$.
 Moreover, given the graph $G$ and the set $U$, one can compute such a set $D$ in polynomial time.
\end{lemma}

\begin{proof}
 We compute the set $D$ using a greedy algorithm.
 We initialize $D \coloneqq \emptyset$ and, as long as $D$ does not dominate every vertex from $U$ in the graph $G^2$, we pick an arbitrary undominated vertex and add it to $D$.
 Clearly, this can be done in polynomial time.
 We need to prove the desired bound on the size of $D$.
 Suppose $D = \{v_1,\dots,v_s\}$ where we list vertices in the order in which they are added to $D$ by the greedy algorithm described above.
 Observe that $N_G[v_i] \cap N_G[v_j] = \emptyset$ for all $i < j \in [s]$ since otherwise $v_j$ would be dominated by $v_i$ and the greedy algorithm would not have added it to $D$.
 Then
 \[s\cdot\ell \leq \sum_{i = 1}^{s} \sum_{w \in N_G[v_i]} \deg_G(w) \leq \sum_{w \in V(G)} \deg_G(w) \leq 2\cdot |E(G)|.\]
 Rearranging the terms gives the desired bound.
\end{proof}

\begin{corollary}
 \label{cor:dist-3-domset}
 There is a polynomial-time algorithm that, given an integer $q \geq 1$ and a $q$-irreducible graph $G$, computes a set $D \subseteq V(G)$ such that
 \begin{enumerate}
  \item for every $v \in V(G)$ there is some $w \in D$ such that $\dist_{G}(v,w) \leq 3$, and
  \item $|D| \leq \frac{2\cdot|E(G)|}{q}$.
 \end{enumerate}
\end{corollary}

\begin{proof}
 Let $U \coloneqq \{u \in V(G) \mid |N_{G^2}[u]| > q\}$.
 By Lemma \ref{la:dist-2-domset-for-high-degree} there is some set $D \subseteq U$ such that
 \[|D| \leq \frac{2\cdot|E(G)|}{q}\]
 and for every $u \in U$ there is some $w \in D$ such that $\dist_G(w,u) \leq 2$.
 Also, for every $v \in V(G)$ there is some $u \in U$ such that $\dist_G(v,u) \leq 1$.
 In combination, gives the desired properties of the set $D$.
 Finally, the set $U$ can clearly be computed in polynomial time and afterwards, $D$ is computed using the algorithm from Lemma \ref{la:dist-2-domset-for-high-degree}.
\end{proof}

Now, we are able to use the dominating set constructed in the last corollary to obtain a $(O(n/q), O(n/q), O(1))$-protrusion decomposition.
Let us start by formally defining protrusion decompositions.

\begin{definition}
 Let $G$ be a graph and let $\alpha,\delta,k \geq 2$ be integers.
 An \emph{$(\alpha,\delta,k)$-protrusion decomposition of $G$} is a rooted tree decomposition $(T,\beta)$ of $G$ such that
 \begin{enumerate}
  \item $|\beta(r)| \leq \alpha$ where $r$ is the root of $T$,
  \item $|\beta(t)| \leq k$ for every other node $t \in V(T) \setminus \{r\}$, and
  \item $\deg_T(r) \leq \delta$.
 \end{enumerate}
\end{definition}

\begin{lemma}[{\cite[Lemma 6.2]{BodlaenderFLPST16}}]
 \label{la:protrusion-decomposition}
 Let $r \geq 1$ be a fixed integer.
 Let $G$ be a planar graph and suppose that there is a set $D \subseteq V(G)$ such that for every $v \in V(G)$ there is some $w \in D$ such that $\dist_{G}(v,w) \leq r$.
 Then there is a $(c|D|,c|D|,c)$-protrusion decomposition of $G$ where $c$ is some constant that only depends on $r$.
 
 Moreover, given $G$ and $D$, a $(c|D|,c|D|,c)$-protrusion decomposition of $G$ can be computed in polynomial time. 
\end{lemma}

We remark at this point that our definition of a protrusion decomposition, which follows \cite{FominLSZ19}, is slightly different from the definition from \cite{BodlaenderFLPST16}.
However, up to constant factors in the parameters, the precise definition does not matter for the existence of a protrusion decomposition (see also \cite[Chapter 15]{FominLSZ19}).
Let us also point out that \cite[Lemma 6.2]{BodlaenderFLPST16} does not directly state the existence of a polynomial-time algorithm that computes the protrusion decomposition.
However, the algorithm is imminent from the proof.
Alternatively, an algorithm can also be obtained via methods described in \cite[Chapter 15]{FominLSZ19}.

By combining Corollary \ref{cor:dist-3-domset} and Lemma \ref{la:protrusion-decomposition} we obtain the following result.

\begin{corollary}
 \label{cor:protrusion-decomposition}
 There is a polynomial-time algorithm that, given an integer $q \geq 1$ and a $q$-irreducible, planar graph $G$, computes a $(\frac{cn}{q},\frac{cn}{q},c)$-protrusion decomposition of $G$ for some absolute constant $c$.
\end{corollary}

The next lemma implements the last step of the second subroutine.

\begin{lemma}
 \label{la:dp-protrusion-decomposition}
 There is an algorithm that, given a graph $G$, an integer $q$ and an $(\alpha,\delta,k)$-protrusion decomposition $(T,\beta)$ of $G$,
 decides whether $G$ has a square $q$-coloring in time $|V(T)| \cdot n^{O(\alpha + \delta \cdot 2^{k})}$.
\end{lemma}

\begin{proof}
 Let $X \coloneqq \beta(r)$ where $r$ denotes the root of $T$.
 Also, let $t_1,\dots,t_\delta$ denote the children of $r$ in $T$, and define $Y_i \coloneqq \beta(r) \cap \beta(t_i)$ for all $i \in [\delta]$.
 Moreover, let $V_i$ denote the set of vertices contained in bags below $t_i$ (including $t_i$ itself).
 Finally, we define
 \[\CY \coloneqq \bigcup_{i \in \delta} (\{i\} \times Y_i)\]
 and
 \[\CZ \coloneqq \{(i,A) \mid i \in [\delta], A \subseteq Y_i\}.\]
 Observe that $|\CY| \leq \delta \cdot k$ and $|\CZ| \leq \delta \cdot 2^k$.
 
 The algorithm iterates over all triples of functions $\chi\colon X \rightarrow [q]$, $\xi \colon \CY \rightarrow 2^X$ such that $\xi(i,v) \subseteq Y_i$ for all $i \in [\delta]$ and $v \in Y_i$, and $\rho\colon \CZ \to [q]_0$.
 Observe that the number of such triples is upper bounded by
 \begin{equation}
  \label{eq:number-triples}
  q^{|X|} \cdot (2^k)^{\delta k} \cdot (q+1)^{\delta2^k} = q^{O(\alpha + \delta \cdot 2^{k})}.
 \end{equation}
 
 So let us fix such a triple $(\chi,\xi,\rho)$.
 We say that $(\chi,\xi,\rho)$ is \emph{locally valid} if
 \begin{enumerate}[label = (PL.\arabic*)]
  \item\label{item:planar-local-1} $\chi(u) \neq \chi(v)$ for all distinct $u,v \in X$ such that $\dist_{G[X]}(u,v) \leq 2$,
  \item\label{item:planar-local-2} there are no vertices $u,w \in Y_i$ and $v \in X$ such that $uv \in E(G)$,  $u \in \xi(i,w)$ and $\chi(v) = \chi(w)$, and
  \item\label{item:planar-local-3} there are no distinct $i,i' \in [\delta]$ and $v \in Y_i$, $w \in Y_{i'}$ such that $\chi(v) = \chi(w)$ and $\xi(i,v) \cap \xi(i',w) \neq \emptyset$.
 \end{enumerate}
 Note that these conditions extend the conditions from Definition \ref{def:local-validity}.
 
 We also say that $(\chi,\xi,\rho)$ is \emph{extendable} if, for every $i \in [\delta]$, there exists a square $q$-coloring $\echi_i$ of $G[V_i]$ such that
 \begin{enumerate}[label = (PD.\arabic*)]
  \item\label{item:dp-chi-pl} $\chi(v) = \echi_i(v)$ for all $v \in Y_i$,
  \item\label{item:dp-xi-pl} $\xi(i,v) = A_i(\echi_i,\chi(v))$ where 
   \[A_i(\echi_i,c) \coloneqq \{v \in Y_i \mid \exists w \in V_i \setminus Y_i \colon \echi_i(w) = c \wedge vw \in E(G)\},\] and
  \item\label{item:dp-rho-pl} \[\rho(i,A) = |\{c \in [q] \mid A_i(\echi_i,c) = A\} \setminus \{\chi(v) \mid v \in Y_i\}|\] for all $A \in 2^{Y_i}$.
 \end{enumerate}
 Observe that these conditions precisely reflect the conditions from Definition \ref{def:dp-table} and hence, it can be checked in time $|V(T)| \cdot n^{O(2^k)}$ whether $(\chi,\xi,\rho)$ is extendable by Lemma \ref{la:dp-table}.
 If $(\chi,\xi,\rho)$ is not locally valid and extendable, then the algorithm rejects the current triple $(\chi,\xi,\rho)$ and moves to the next iteration.
 
 Otherwise, similar to join operation in the proof Lemma \ref{la:dp-table}, we need to verify that the partial square $q$-colorings are compatible to one another and can indeed be combined into a global square $q$-coloring.
 Let us say a color $c \in [q]$ \emph{$i$-free} (with respect to $(\chi,\xi,\rho)$) if $c \notin \{\chi(v) \mid v \in Y_i\}$.
 Here, the remaining task is to assign, to every $(i,A) \in \CZ$, a set of $\rho(i,A)$ of $i$-free colors in a consistent way.
 We check this by implementing another dynamic programming algorithm.
 
 For a set $C \subseteq [q]$ of colors and a function $\varrho\colon \CZ \rightarrow [q]_0$ we define $\gamma(C,\varrho)$ to be true if there is a mapping $\eta\colon \CZ \rightarrow 2^{C}$ satisfying the following conditions:
 \begin{enumerate}[label = (\roman*)]
  \item\label{item:dp-free-colors-1} $|\eta(i,A)| = \varrho(i,A)$,
  \item\label{item:dp-free-colors-2} $\eta(i,A) \cap \{\chi(v) \mid v \in Y_i\} = \emptyset$,
  \item\label{item:dp-free-colors-3} $\eta(i,A) \cap \eta(i,A') = \emptyset$ for all distinct $A,A' \subseteq Y_i$,
  \item\label{item:dp-free-colors-4} $\eta(i,A) \cap \eta(i',A') = \emptyset$ for all distinct $(i,A),(i',A') \in \CZ$ such that $A \cap A' = \emptyset$, and
  \item\label{item:dp-free-colors-5} there are no $(i,A) \in \CZ$, $u \in A$ and $v \in X$ such that $uv \in E(G)$ and $\chi(v) \in \eta(i,A)$.
  \item\label{item:dp-free-colors-6} there are no distinct $i,i' \in [\delta]$, $A \subseteq Y_i$, $u \in Y_i \cap Y_{i'}$, $v \in Y_{i'}$ such that $u \in \xi(i',v)$, $u \in A$ and $\chi(v) \in \eta(i,A)$.
 \end{enumerate}
 More intuitively, $\gamma(C,\varrho)$ is set to true if we can assign $\varrho(i,A)$ many free colors to every $(i,A) \in \CZ$ in a consistent way by only using colors from the set $C$.
 In particular, our goal is to determine whether $\gamma([q],\rho)$ is true.
 We achieve this by computing a suitable subset of the entries $\gamma(C,\varrho)$ using a dynamic programming approach.
 More precisely, let $\CC_{{\sf sing}} \coloneqq \{\{c\} \mid c \in [q]\}$ be the set of  all singleton subsets of $[q]$.
 Also, let $\CC_{{\sf seg}} \coloneqq \{[q'] \mid q' \in [q]\}$ be the set of all initial segments of $[q]$.
 We compute $\gamma(C,\varrho)$ for all $C \in \CC_{{\sf sing}} \cup \CC_{{\sf seg}}$ and all possible mappings $\varrho$ using a dynamic programming approach.
 The next claim shows how to compute the entries $\gamma(C,\varrho)$ for $C \in \CC_{{\sf sing}}$.
 
 \begin{claim}
  \label{claim:dp-base-case}
  There is an algorithm that, given $c \in [q]$ and $\varrho\colon \CZ \rightarrow [q]_0$, determines whether $\gamma(\{c\},\varrho)$ is true and runs in time polynomial in $n$ and $|\CZ|$.
 \end{claim}
 \begin{claimproof}
  Since $|\{c\}| = 1$ there is at most one candidate function $\eta\colon \CZ \rightarrow 2^{C}$ which satisfies Condition \ref{item:dp-free-colors-1} (i.e., $\eta(i,A) = \emptyset$ if $\varrho(i,A) = 0$, and $\eta(i,A) = \{c\}$ if $\varrho(i,A) = 1$).
  Now, it can be easily checked in time polynomial in $n$ and $|\CZ|$ whether all other Conditions \ref{item:dp-free-colors-2} - \ref{item:dp-free-colors-6} are satisfied.
 \end{claimproof}

 Now, we argue how to compute the remaining entries in a dynamic programming fashion.
 
 \begin{claim}
  \label{claim:dp-step}
  Let $C \subseteq [q]$ and $\varrho\colon \CZ \rightarrow [q]_0$ be some mapping.
  Also suppose $C = C_1 \uplus C_2$.
  Then $\gamma(C,\varrho)$ is true if and only if there are functions $\varrho_1,\varrho_2\colon \CZ \rightarrow [q]_0$ such that
  \begin{enumerate}[label = (\Alph*)]
   \item\label{item:dp-step-1} $\gamma(C_1,\varrho_1)$ and $\gamma(C_2,\varrho_2)$ are true, and
   \item\label{item:dp-step-2} $\varrho(i,A) = \varrho_1(i,A) + \varrho_2(i,A)$ for all $(i,A) \in \CZ$.
  \end{enumerate}
 \end{claim}
 \begin{claimproof}
  For the forward direction, first suppose that $\gamma(C,\varrho)$ is true, i.e., there is a function $\eta\colon \CZ \rightarrow 2^{C}$ satisfying Conditions \ref{item:dp-free-colors-1} - \ref{item:dp-free-colors-6}.
  For $j \in \{1,2\}$, we define $\eta_j\colon \CZ \rightarrow 2^{C_j}$ via $\eta_j(i,A) \coloneqq \eta(i,A) \cap C_j$.
  Moreover, we define $\varrho_j(i,A) \coloneqq |\eta_j(i,A)|$. We have
  \[\varrho(i,A) = |\eta(i,A)| = |\eta(i,A) \cap C_1| + |\eta(i,A) \cap C_2| = |\eta_1(i,A)| + |\eta_2(i,A)| = |\varrho_1(i,A)| + |\varrho_2(i,A)|\]
  which shows Property \ref{item:dp-step-2}.
  Also, it is easy to verify that Property \ref{item:dp-step-1} holds via the witnessing mappings $\eta_1$ and $\eta_2$.
  Indeed, for $j \in \{1,2\}$, the function $\eta_j$ satisfies Condition \ref{item:dp-free-colors-1} by definition of $\varrho_j$, and Conditions \ref{item:dp-free-colors-2} - \ref{item:dp-free-colors-6} for $\eta_j$ follow directly from the corresponding properties of $\eta$.
  
  For the backward direction let $\varrho_1,\varrho_2\colon \CZ \rightarrow [q]_0$ be functions satisfying \ref{item:dp-step-1} and \ref{item:dp-step-2}.
  By Property \ref{item:dp-step-1}, there are functions $\eta_j\colon \CZ \rightarrow 2^{C_j}$, $j \in \{1,2\}$, satisfying Conditions \ref{item:dp-free-colors-1} - \ref{item:dp-free-colors-6} with respect to $(C_j,\varrho_j)$.
  We define $\eta\colon \CZ \to 2^C$ via $\eta(i,A) \coloneqq \eta_1(i,A) \cup \eta_2(i,A)$.
  It is again easy to check that $\eta$ witnesses that $\gamma(C,\varrho)$ is true.
  Indeed, Condition \ref{item:dp-free-colors-1} follows directly from Property \ref{item:dp-step-2} and the definition of $\eta$.
  Also, as before, Conditions \ref{item:dp-free-colors-2} - \ref{item:dp-free-colors-6} for $\eta$ follow directly from the corresponding properties of $\eta_1$ and $\eta_2$.
 \end{claimproof}
 
 Now, we can iteratively compute all entries $\gamma(C,\varrho)$ where $C \in \CC_{{\sf sing}} \cup \CC_{{\sf seg}}$ as follows.
 First, we iterate over all $C \in \CC_{{\sf sing}}$ and all choices of $\varrho$ and compute $\gamma(C,\varrho)$ using Claim \ref{claim:dp-base-case}.
 Afterwards, we iterate over all values $q' \in \{2,\dots,q\}$ (starting at $q' = 2$) as well as over all $\varrho\colon \CZ \rightarrow [q]_0$.
 To determine whether $\gamma([q'],\varrho)$ is true, we set $C_1 \coloneqq \{1,\dots,q'-1\}$ and $C_2 \coloneqq \{q'\}$ and iterate over all function pairs $\varrho_1,\varrho_2\colon \CZ \rightarrow [q]_0$ such that $\varrho(i,A) = \varrho_1(i,A) + \varrho_2(i,A)$ for all $(i,A) \in \CZ$.
 If there is such a pair such that $\gamma(C_1,\varrho_1)$ and $\gamma(C_2,\varrho_2)$ are true, then we set $\gamma(C,\varrho)$ to true.
 This correctly computes $\gamma(C,\varrho)$ by Claim \ref{claim:dp-step}.
 
 Finally, the main algorithm returns YES (i.e., the algorithm concludes that $G$ has a square $q$-coloring) if $\gamma([q],\rho)$ is true.
 Otherwise, the algorithm moves to the next iteration of the outer loop (i.e., the algorithm chooses the next triple $(\chi,\xi,\rho)$).
 
 If the algorithm has checked all triples $(\chi,\xi,\rho)$ and never returned YES, then it returns NO.
 
 \begin{correctness}
  We argue that the algorithm described above correctly decides whether $G$ has a square $q$-coloring.
  If $G$ is not $q$-irreducible, this follows directly from Lemma \ref{la:reduce-graph}.
  So suppose that $G$ is $q$-irreducible.
  
  First suppose there is some square $q$-coloring $\echi$ of $G$.
  We define $\chi\coloneqq \echi|_X$ and $\echi_i\coloneqq \echi|_{V_i}$ for all $i \in [\delta]$.
  Moreover, we define
  \[\xi(i,v) = A_i(\echi_i,\chi(v))\]
  for all $(i,v) \in \CY$, and
  \[\rho(i,A) = |\{c \in [q] \mid A_i(\echi_i,c) = A\} \setminus \{\chi(v) \mid v \in Y_i\}|\]
  for all $(i,A) \in \CZ$.
  We claim the algorithm return YES in the iteration $(\chi,\xi,\rho)$, i.e., we need to show that $(\chi,\xi,\rho)$ is locally valid, extendable and $\gamma([q],\rho)$ is true.
  
  Condition \ref{item:planar-local-1} is clearly satisfied since $\echi$ is a square coloring of $G$.
  For Condition \ref{item:planar-local-2} suppose $uv \in E(G)$, $\chi(v) = \chi(w)$, and  $u \in \xi(i,w) = A_i(\echi_i,\chi(w))$.
  By definition, there is some $v' \in V_i \setminus Y_i$ such that $\chi(v') = \chi(w) = \chi(v)$ and $uv' \in E(G)$.
  Note that $v \neq v'$ since $v \in X$.
  This contradicts $\echi$ being a square coloring of $G$.
  Similarly, for Condition \ref{item:planar-local-3}, suppose there are distinct $i,i' \in [\delta]$ and $v \in Y_i$, $w \in Y_{i'}$ such that $\chi(v) = \chi(w)$ and $\xi(i,v) \cap \xi(i,w)$.
  Let $u \in \xi(i,v) \cap \xi(i,w)$.
  By definition, there are $v' \in V_i \setminus Y_i$ such that $\chi(v') = \chi(v)$ and $uv' \in E(G)$, and $w' \in V_{i'} \setminus Y_{i'}$ such that $\chi(w') = \chi(w)$ and $uv' \in E(G)$.
  Overall, we obtain distinct vertices $v',u,w' \in V(G)$ such that $uv',uw' \in E(G)$ and $\chi(v') = \chi(w')$.
  Again, this contradicts $\echi$ being a square coloring of $G$.
  So $(\chi,\xi,\rho)$ is locally valid.
  
  Also, $(\chi,\xi,\rho)$ is extendable since Conditions \ref{item:dp-chi-pl} - \ref{item:dp-rho-pl} are satisfied by definition.
  So it remains to argue that there is some witness $\eta\colon \CZ \rightarrow 2^{[q]}$ showing that $\gamma([q],\rho)$ is true.
  We define
  \[\eta(i,A) \coloneqq \{c \in [q] \mid A_i(\echi_i,c) = A\} \setminus \{\chi(v) \mid v \in Y_i\}.\]
  We have that Conditions \ref{item:dp-free-colors-1} - \ref{item:dp-free-colors-3} are trivially satisfied.
  Also, Conditions \ref{item:dp-free-colors-4} and \ref{item:dp-free-colors-6} can be verified using the same arguments as for Condition \ref{item:planar-local-3} (note that, using Condition \ref{item:dp-free-colors-3}, we only need to consider the case $i \neq i'$).
  Similarly, Condition \ref{item:dp-free-colors-5} follows by the same arguments as Condition \ref{item:planar-local-2}.
  
  Finally, observe that the algorithm correctly computes $\gamma([q],\rho)$ using Claims \ref{claim:dp-base-case} and \ref{claim:dp-step}.
  So overall, the algorithm returns YES.
  
  \medskip
  
  In the other direction, suppose that the algorithm outputs YES, i.e., there is a triple $(\chi,\xi,\rho)$ that is locally valid, extendable and $\gamma([q],\rho)$ is true.
  Since $(\chi,\xi,\rho)$ is extendable, for every $i \in [\delta]$, there is square $q$-coloring $\echi_i$ of $G[V_i]$ satisfying \ref{item:dp-chi-pl} - \ref{item:dp-rho-pl}.
  Also, $\gamma([q],\rho)$ is true meaning that there is a mapping $\eta\colon\CZ \to 2^{[q]}$ satisfying \ref{item:dp-free-colors-1} - \ref{item:dp-free-colors-4} (with respect to the tuple $([q],\rho)$).
  Since
  \[|\eta(i,A)| = \rho(i,A) = |\{c \in [q] \mid A_i(\echi_i,c) = A\} \setminus \{\chi(v) \mid v \in Y_i\}|\]
  by Conditions \ref{item:dp-rho-pl} and \ref{item:dp-free-colors-1}, we can rename the colors in the coloring $\echi_i$ (using Conditions \ref{item:dp-free-colors-2} and \ref{item:dp-free-colors-3}) so that
  \begin{equation}
   \label{eq:rename-colors-planar}
   \eta(i,A) = \{c \in [q] \mid A_i(\echi_i,c) = A\} \setminus \{\chi(v) \mid v \in Y_i\}
  \end{equation}
  for all $(i,A) \in \CZ$ while preserving Conditions \ref{item:dp-chi-pl} - \ref{item:dp-rho-pl} (since we only rename free colors for every $i \in [\delta]$).
  
  Now, we define the coloring $\echi\colon V(G) \to [q]$ via
  \[\echi(v) \coloneqq \begin{cases}
                        \chi(v) &\text{if } v \in X\\
                        \echi_i(v) &\text{if } v \in V_i \text{ for some } i \in [\delta]
                       \end{cases}\]
  Note that $\echi$ is well-defined by Condition \ref{item:dp-chi-pl}.
  We claim that $\echi$ is a square coloring of $G$.
  Let $v,w \in V(G)$ be distinct vertices such that $\dist(v,w) \leq 2$.
  We need to argue that $\echi(v) \neq \echi(w)$.
  
  First suppose that $v \in X$.
  If $vw \in E(G)$ then either $w \in X$ and $\echi(v) \neq \echi(w)$ by \ref{item:planar-local-1}, or $w \in V_i \setminus X$ for some $i \in [\delta]$ which implies that $v \in V_i$ and $\echi(v) \neq \echi(w)$ since $\echi_i$ is a square coloring of $G[V_i]$.
  So suppose that $\dist(v,w) = 2$ and let $u \in N_G(v) \cap N_G(w)$.
  If $u,w \in X$ then $\echi(v) \neq \echi(w)$ by \ref{item:planar-local-1}.
  If $w \in X$, but $u \in V_i \setminus X$ for some $i \in [\delta]$ then $v,w \in V_i$, and $\echi(v) \neq \echi(w)$ follows as before.
  So suppose that $w \notin X$, i.e., $w \in V_i \setminus X$ for some $i \in [\delta]$.
  This means that $u \in V_i$.
  If $v \in V_i$ then again $\echi(v) \neq \echi(w)$ follows as before.
  So suppose that $v \in X \setminus V_i$ which means that $u \in Y_i$.
  If $\echi(w)$ is not $i$-free then there is some $w' \in Y_i$ such that $\echi(w) = \echi(w') = \chi(w')$, and $u \in \xi(i,w')$ by \ref{item:dp-xi-pl}.
  So $\chi(w') \neq \chi(v)$ by \ref{item:planar-local-2} which implies that $\echi(v) \neq \echi(w)$.
  In the other subcase $\echi(w)$ is $i$-free which means that $\echi(w) \in \eta(i,A)$ where $A = A_i(\echi_i,\echi(w))$ by Equation \eqref{eq:rename-colors-planar}.
  Observe that $u \in A$ by definition.
  So $\echi(v) = \chi(v) \notin \eta(i,A)$ by Condition \ref{item:dp-free-colors-5}.
  It follows that $\echi(v) \neq \echi(w)$.
  This completes the case that $v \in X$.
  
  By symmetry, we also obtain that $\echi(v) \neq \echi(w)$ if $w \in X$.
  So suppose that $v,w \notin X$.
  So there are $i,i' \in [\delta]$ such that $v \in V_i \setminus Y_i$ and $w \in V_{i'} \setminus Y_{i'}$.
  If $i = i'$ then $\echi(v) \neq \echi(w)$ since $\echi_i$ is a square coloring of $G[V_i]$.
  So suppose that $i \neq i'$.
  Observe that $vw \notin E(G)$.
  Since $\dist(v,w) \leq 2$ there is some $u \in N_G(v) \cap N_G(w)$.
  We have that $u \in Y_i \cap Y_{i'}$.
  
  If $\chi(v)$ is $i$-free then $\echi(v) \in \eta(i,A_v)$ where $A_v = A_i(\echi_i,\echi(v))$ by Equation \eqref{eq:rename-colors-planar}, and $u \in A_v$.
  Similarly, if $\chi(w)$ is $i'$-free then $\echi(w) \in \eta(i',A_w)$ where $A_w = A_{i'}(\echi_{i'},\echi(v))$, and $u \in A_w$.
  
  So if $\chi(v)$ is $i$-free and $\chi(w)$ is $i'$-free, then $A_v \cap A_w \neq \emptyset$ which implies that $\eta(i,A_v) \cap \eta(i',A_w) = \emptyset$ by Condition \ref{item:dp-free-colors-4}.
  It follows that $\echi(v) \neq \echi(w)$.
  
  Next, suppose that $\chi(v)$ is $i$-free and $\chi(w)$ is not $i'$-free.
  Then there is some $w' \in Y_{i'}$ such that $\echi(w) = \echi(w')$.
  Also, $u \in \xi(i',w')$ by \ref{item:dp-xi-pl}.
  So $\chi(w') \notin \eta(i,A_v)$ by Condition \ref{item:dp-free-colors-6}.
  Since $\echi(v) \in \eta(i,A_v)$, we conclude that $\echi(v) \neq \echi(w)$.
  
  The case where $\chi(v)$ is not $i$-free and $\chi(w)$ is $i'$-free can be handled symmetrically.\
  
  So suppose $\chi(v)$ is not $i$-free and $\chi(w)$ is not $i'$-free.
  Then there is some $v' \in Y_{i}$ such that $\echi(v) = \echi(v')$, and there is some $w' \in Y_{i'}$ such that $\echi(w) = \echi(w')$.
  Also, $u \in \xi(i,v') \cap \xi(i',w')$ by \ref{item:dp-xi-pl}.
  So $\chi(v') \neq \chi(w')$ by \ref{item:planar-local-3}.
  It follows that $\echi(v) \neq \echi(w)$.
  
  So overall, $\echi$ is a square coloring of $G$ which completes the correctness proof of the algorithm.
 \end{correctness}
 
 \begin{runningtime}
  By Equation \ref{eq:number-triples}, the outer loop of the algorithm considers at most $q^{O(\alpha + \delta \cdot 2^{k})}$ many triples $(\chi,\xi,\rho)$.
  So let us fix such a triple.
  Clearly, it can be checked in polynomial time whether $(\chi,\xi,\rho)$ is locally valid.
  By turning the subtree rooted at $t_i$ into a nice tree decomposition of $G[V_i]$ for all $i \in [\delta]$, we can use Lemma \ref{la:dp-table} to decide whether $(\chi,\xi,\rho)$ is extendable in time $|V(T)| \cdot n^{O(2^k)}$.
  So it only remains to analyze the time required to compute the entries $\gamma(C,\varrho)$ for all $C \in \CC_{{\sf sing}} \cup \CC_{{\sf seg}}$ and all possible mappings $\varrho$.
  First observe that the number of such pairs $(C,\varrho)$ is bounded by
  \[|\CC_{{\sf sing}} \cup \CC_{{\sf seg}}| \cdot (q+1)^{|\CZ|} \leq 2q \cdot (q+1)^{\delta2^k} = q^{O(\delta \cdot 2^{k})}.\]
  If $|C| = 1$ then we can compute $\gamma(C,\varrho)$ in time polynomial in $n$ and $|\CZ| \leq \delta \cdot 2^{k}$ by Claim \ref{claim:dp-base-case}.
  Otherwise, we use Claim \ref{claim:dp-step} and computing an entry $\gamma(C,\varrho)$ takes time
  \[\left((q+1)^{\delta2^k}\right)^2(n + |\CZ|)^{O(1)}.\]
  So overall, computing all entries takes time
  \[q^{O(\delta \cdot 2^{k})} \cdot \left((q+1)^{\delta2^k}\right)^2(n + |\CZ|)^{O(1)} = q^{O(\alpha + \delta \cdot 2^{k})} \cdot (n + |\CZ|)^{O(1)}.\]
  So overall, the algorithm takes time
  \[q^{O(\alpha + \delta \cdot 2^{k})} \cdot \Big(|V(T)| \cdot n^{O(2^k)} + q^{O(\delta \cdot 2^{k})} \cdot (n + |\CZ|)^{O(1)}\Big) = |V(T)| \cdot n^{O(\alpha + \delta \cdot 2^{k})}.\]
  as desired.
 \end{runningtime}
\end{proof}

Now, we can combine all the parts to obtain the following theorem.

\begin{theorem}[Lemma \ref{lem:alg-planar-many-colors-intro} restated]
 \label{thm:alg-planar-many-colors}
 There is an algorithm that solves \textsc{Planar} \qSqCol\ in time $n^{O(n/q)}$.
\end{theorem}

\begin{proof}
 Let $G$ denote the input graph.
 The algorithm computes $U \coloneqq \{u \in V(G) \mid |N_{G^2}[u]| > q\}$ and $W \coloneqq U \cup N_G(U)$.
 If $W \subsetneq V(G)$ the algorithm deletes all vertices from $V(G) \setminus W$ and recursively decides whether $G[W]$ has a square $q$-coloring.
 Observe that this is correct by Lemma \ref{la:reduce-graph}.
 
 Otherwise, $G$ is $q$-irreducible.
 Using Corollary \ref{cor:protrusion-decomposition}, we compute a $(\alpha,\delta,c)$-protrusion decomposition $(T,\beta)$ of $G$ where $\alpha,\delta = O(\frac{n}{q})$ and $c$ is some absolute constant.
 Afterwards, we decide whether $G$ has a square $q$-coloring using Lemma \ref{la:dp-protrusion-decomposition}.
 
 For the running time, observe that the algorithm always arrives at a $q$-irreducible graph after polynomially many steps.
 Also, given a $q$-irreducible graph, the protrusion decomposition $(T,\beta)$ is computed in polynomial time by Corollary \ref{cor:protrusion-decomposition}.
 Finally, the application of the algorithm from Lemma \ref{la:dp-protrusion-decomposition} takes time $n^{O(n/q)}$.
\end{proof}

Combining Corollary \ref{cor:alg-planar-few-colors} and Theorem \ref{thm:alg-planar-many-colors}, we obtain a subexponential algorithm for \textsc{Planar} \SqCol.

\begin{corollary}[Theorem \ref{thm:planar-alg-intro} restated]
  There is an algorithm that solves \textsc{Planar} \SqCol\ in time $2^{O(n^{2/3} \log n)}$.
\end{corollary}

\begin{proof}
 Let $G$ be the input graph and $q$ the number of colors.
 If
 \[q \leq n^{1/3}\]
 then we apply Corollary \ref{cor:alg-planar-few-colors} giving a running time of
 \[q^{O(\sqrt{qn})} = n^{O(\sqrt{n^{4/3}})} = 2^{O(n^{2/3} \log n)}.\]
 Otherwise, we apply Theorem \ref{thm:alg-planar-many-colors} resulting in a running time
 \[n^{O(\sqrt{n/q})} = n^{O(n^{2/3})} = 2^{O(n^{2/3} \log n)}.\qedhere\]
\end{proof}

\section{Lower Bounds for Planar Graphs}
\label{sec:planar-lower-bound}
In this section, we provide hardness results for \textsc{Planar} \SqCol.
More precisely, we prove the following theorem which implies Theorems \ref{thm:planar-lb-intro} and \ref{thm:planar-eth-lb-intro}.

\begin{theorem}
 \label{thm:planar-hardness-q}
 For every fixed $q \geq 4$, the problem \textsc{Planar} \qSqCol\ is \NP-hard. Moreover, assuming \ETH, it cannot be solved in time $2^{o(\sqrt{n})}$.
\end{theorem}

Let us start with two basic remarks.
First observe that the bound on $q$ is optimal since \qSqCol\ is polynomial-time solvable (on general graphs) for all $q \leq 3$.
Indeed, if $q \leq 3$, then every graph $G$ of maximum degree at least $3$ is a trivial NO-instance since $G^2$ contains a clique of size at least $4$.
So it suffices to consider graphs of maximum degree at most $2$ for which there is an easy algorithm.

Also note that the theorem immediately implies the same hardness results for \SqCol.

\begin{corollary}
 \textsc{Planar} \SqCol\ is \NP-hard. Moreover, assuming \ETH, it cannot be solved in time $2^{o(\sqrt{n})}$.
\end{corollary}

To prove Theorem \ref{thm:planar-hardness-q} we give a reduction from \textsc{Planar $3$-Coloring} problem to \textsc{Planar} \qSqCol\ for all $q \geq 4$ and then exploit known hardness results for \textsc{Planar $3$-Coloring}.

\begin{theorem}
 \label{thm:planar-3-col-hardness}
 \textsc{Planar $3$-Coloring} is \NP-hard. Moreover, assuming \ETH, it cannot be solved in time $2^{o(\sqrt{n})}$.
\end{theorem}

The \NP-hardness of \textsc{Planar $3$-Coloring} was first proved in \cite{Stockmeyer73}.
For the second part of theorem we refer to \cite[Theorem 14.9]{CyganFKLMPPS15}.

To describe the reduction, we split it into two steps.
We first consider a variant where we also allow special \emph{equality edges} where endpoints have to be assigned the same color.
Then, in a second step, we replace the equality edges by a simple gadget to complete the intended reduction.
We further split the first step depending on whether $q = 4$ or $q \geq 5$ as different constructions are needed for these two cases.

\subsection{Four Colors}

\begin{figure}
 \centering
 \begin{tikzpicture}
  \path[use as bounding box] (-11.4,-3.6) rectangle (3.6,3.6);
  
  \node[normalvertex,fill=blue!80] (1) at (-7,2) {};
  \node[normalvertex,fill=red!80] (2) at (-7,-2) {};
  \node[normalvertex,fill=red!80] (3) at (-11,2) {};
  \node[normalvertex,fill=Green!80] (4) at (-11,-2) {};
  
  \foreach \i/\j in {1/2,1/3,1/4,2/4,3/4}{
   \draw[thick] (\i) edge (\j);
  }
  
  \draw[white,fill=gray!80] (-5.5,0.2) -- (-4.9,0.2) -- (-4.9,0.4) -- (-4.5,0) -- (-4.9,-0.4) -- (-4.9,-0.2) -- (-5.5,-0.2) -- cycle;
 
  \node[normalvertex,fill=blue!80] (12) at ($(2,2)+(0:0.6)$) {};
  \node[normalvertex,fill=blue!80] (13) at ($(2,2)+(90:0.6)$) {};
  \node[normalvertex,fill=blue!80] (14) at ($(2,2)+(225:0.6)$) {};
  
  \node[normalvertex,fill=red!80] (21) at ($(2,-2)+(0:0.6)$) {};
  \node[normalvertex,fill=red!80] (24) at ($(2,-2)+(270:0.6)$) {};
  
  \node[normalvertex,fill=red!80] (31) at ($(-2,2)+(90:0.6)$) {};
  \node[normalvertex,fill=red!80] (34) at ($(-2,2)+(180:0.6)$) {};
  
  \node[normalvertex,fill=Green!80] (41) at ($(-2,-2)+(45:0.6)$) {};
  \node[normalvertex,fill=Green!80] (42) at ($(-2,-2)+(270:0.6)$) {};
  \node[normalvertex,fill=Green!80] (43) at ($(-2,-2)+(180:0.6)$) {};
  
  \node[normalvertex,fill=yellow!80] (f1) at (0,-2) {};
  \node[normalvertex,fill=yellow!80] (f2) at (0.4,-0.4) {};
  \node[normalvertex,fill=yellow!80] (f3) at (2,0) {};
  
  \node[normalvertex,fill=yellow!80] (g1) at (0,2) {};
  \node[normalvertex,fill=yellow!80] (g2) at (-0.4,0.4) {};
  \node[normalvertex,fill=yellow!80] (g3) at (-2,0) {};
  
  \node[normalvertex,fill=yellow!80] (h1) at (0,3.2) {};
  \node[normalvertex,fill=yellow!80] (h2) at (0,-3.2) {};
  \node[normalvertex,fill=yellow!80] (h3) at (3.2,0) {};
  \node[normalvertex,fill=yellow!80] (h4) at (-3.2,0) {};
  
  \scoped[on background layer]{
  \draw[color = gray!50, fill = gray!50, line width = 18pt, rounded corners] (12.center) -- (h3.center) -- (21.center) -- (f3.center) -- cycle;
  \draw[color = gray!50, fill = gray!50, line width = 18pt, rounded corners] (13.center) -- (h1.center) -- (31.center) -- (g1.center) -- cycle;
  \draw[color = gray!50, fill = gray!50, line width = 18pt, rounded corners] (14.center) -- (g2.center) -- (41.center) -- (f2.center) -- cycle;
  \draw[color = gray!50, fill = gray!50, line width = 18pt, rounded corners] (24.center) -- (h2.center) -- (42.center) -- (f1.center) -- cycle;
  \draw[color = gray!50, fill = gray!50, line width = 18pt, rounded corners] (34.center) -- (h4.center) -- (43.center) -- (g3.center) -- cycle;
  
  \draw[color = gray!80, line width = 8pt] (h3.center) --  (f3.center);
  \draw[color = gray!80, line width = 8pt] (h1.center) --  (g1.center);
  \draw[color = gray!80, line width = 8pt] (g2.center) --  (f2.center);
  \draw[color = gray!80, line width = 8pt] (h2.center) --  (f1.center);
  \draw[color = gray!80, line width = 8pt] (h4.center) --  (g3.center);
  }
  
  \foreach \i/\j in {12/13,12/14,13/14,21/24,31/34,41/42,41/43,42/43,f1/f2,f1/f3,f2/f3,g1/g2,g1/g3,g2/g3}{
   \draw[gray, dotted,line width = 2pt] (\i) edge (\j);
  }
  
  \draw[gray, dotted,line width = 2pt, bend left=90] (h1) edge (h3);
  \draw[gray, dotted,line width = 2pt, bend right=90] (h1) edge (h4);
  \draw[gray, dotted,line width = 2pt, bend right=90] (h2) edge (h3);
  \draw[gray, dotted,line width = 2pt, bend left=90] (h2) edge (h4);
  
 \end{tikzpicture}
 \caption{A graph $G$ is shown on the left and the graph constructed by Lemma \ref{la:lower-bound-4-colors-with-eq} is displayed on the right.
  All dashed edges represent equality edges contained in $\eqRel$.
  Each gray region is replaced by the gadget from Figure \ref{fig:cross-gadget-4-colors} which ensures that all colors of outgoing vertices are pairwise distinct except those connected by a thick edge which are forced to receive the same color.}
 \label{fig:planar-hardness-4-colors}
\end{figure}
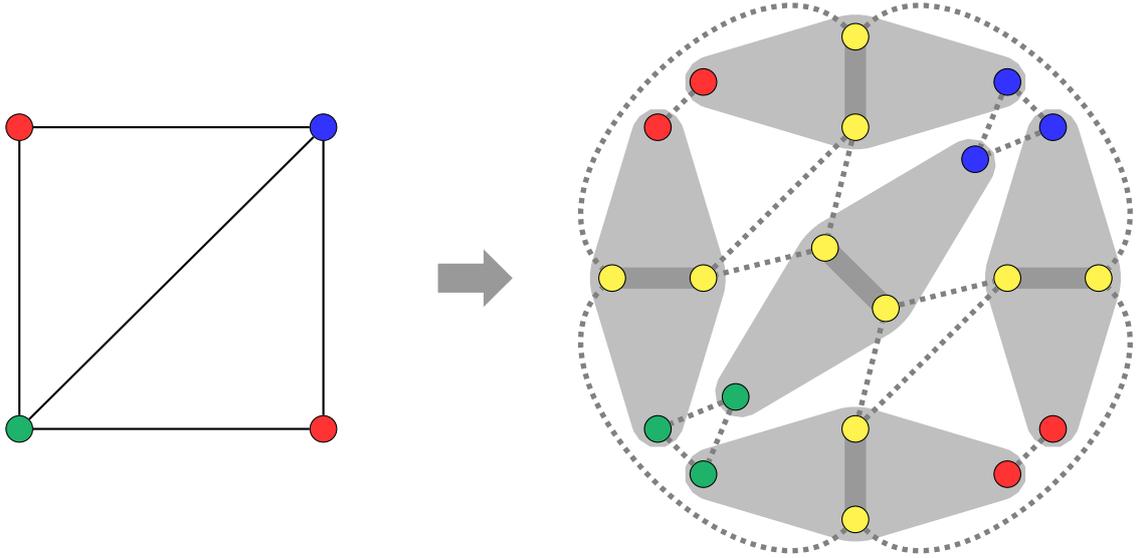

We start by implementing the first step for $q = 4$.
This is achieved by the next lemma.

\begin{lemma}
 \label{la:lower-bound-4-colors-with-eq}
 Let $G$ be a connected planar graph. Then there is a graph $H$ and a set $\eqRel \subseteq \binom{V(H)}{2}$ such that the following conditions are satisfied:
 \begin{enumerate}
  \item $|V(H)| = 14|E(G)|$,
  \item $H^{+\eqRel} \coloneqq (V(H),E(H) \cup \eqRel)$ is planar and has maximum degree $3$, and
  \item $G$ is $3$-colorable if and only if there is coloring $\chi\colon V(H) \rightarrow \{1,2,3,4\}$ such that
   \begin{enumerate}
    \item $\chi(u) = \chi(v)$ for all $uv \in \eqRel$, and
    \item $\chi(u) \neq \chi(v)$ for all distinct $u,v \in V(H)$ such that $\dist_H(u,v) \leq 2$.
   \end{enumerate}
 \end{enumerate}
 Moreover, given the graph $G$, the pair $(H,\eqRel)$ can be computed in polynomial time.
\end{lemma}

\begin{proof}
 A visualization of the construction is given in Figure \ref{fig:planar-hardness-4-colors}.
 Fix a planar embedding of $G$ and let $F$ denote the set of faces.
 We set
 \begin{align*}
  V(H) \coloneqq\;\;\; &\{(v,e) \mid v \in V(G), e \in E(G), v \in e\}\\
       \cup\;          &\{(f,e) \mid f \in F, e \in E(G), e \text{ is incident to } f\}\\
       \cup\;          &E(G) \times \{0,\dots,9\}.
 \end{align*}
 Since every edge is incident to two vertices as well as two faces we get that $|V(H)| = 14|E(G)|$.
 
 Intuitively, the idea is that, for every $v \in V(G)$, all vertices $(v,e)$ receive the same color which represents the color of $v$.
 Also, all vertices $(f,e)$ receive the same color which corresponds to the ``fourth'' color not used in the coloring of $G$ (in Figure \ref{fig:planar-hardness-4-colors}, all vertices $(f,e)$ are colored \yellow).
 Finally, the vertices $E(G) \times \{0,\dots,9\}$ are used to construct gadgets that are placed on every edge of $G$ and which (together with suitable equality edges) ensure that the remaining vertices can only be colored in the desired way.
 
 Let $v \in V(G)$ and let $e_1,\dots,e_d$ denote the incident edges of $v$ ordered cyclically according to the fixed embedding of $G$.
 We define
 \begin{equation}
  \label{eq:rel-vertices}
  \eqRel(v) \coloneqq \{(v,e_i)(v,e_{i+1}) \mid i \in [d]\} \cup \{(v,e_d)(v,e_1)\}.
 \end{equation}
 
 Next, let $f \in F$ be a face of $G$ and let $e_1,\dots,e_k$ denote the incident edges of $f$ ordered cyclically according to the fixed embedding of $G$.
 We define
 \begin{equation}
  \label{eq:rel-faces}
  \eqRel(f) \coloneqq \{(f,e_i)(f,e_{i+1}) \mid i \in [k]\} \cup \{(f,e_k)(f,e_1)\}.
 \end{equation}
 
 Finally, let $e \in E(G)$ and suppose that $e = uv$.
 Also let $f_1,f_2 \in F$ denote the two faces incident to $e$.
 We set
 \begin{equation}
  \label{eq:rel-edges}
  \eqRel(e) \coloneqq \{(e,0)(u,e),\;(e,0)(e,4),\;(e,1)(e,2),\;(e,1)(f_2,e),\;(e,3)(v,e),\;(e,7)(f_1,e)\}.
 \end{equation}
 We also set
 \begin{equation}
  M(e) \coloneqq \{01,23,29,37,45,56,68,69,78,89\}.
 \end{equation}
 A visualization is given in Figure \ref{fig:cross-gadget-4-colors}.
 Overall, we now define
 \[E(H) \coloneqq \bigcup_{e \in E(G)} \{(e,i)(e,j) \mid ij \in M(e)\}\]
 and
 \[\eqRel \coloneqq \bigcup_{v \in V(H)} \eqRel(v) \cup \bigcup_{f \in F} \eqRel(f) \cup \bigcup_{e \in E(G)} \eqRel(e).\]
 
 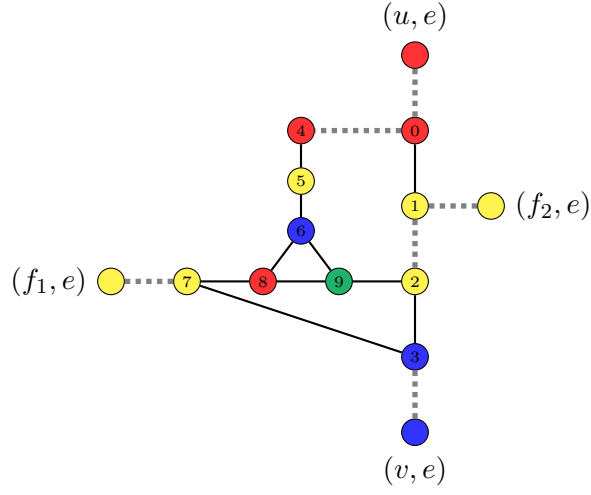
\begin{figure}
  \centering
  \begin{tikzpicture}
   \node[normalvertex,label=left:{$(f_1,e)$},fill=yellow!80] (f1) at (0,2) {};
   \node[normalvertex,fill=yellow!80] (7) at (1,2) {{\tiny $7$}};
   \node[normalvertex,fill=red!80] (8) at (2,2) {{\tiny $8$}};
   \node[normalvertex,fill=Green!80] (9) at (3,2) {{\tiny $9$}};
   \node[normalvertex,fill=yellow!80] (2) at (4,2) {{\tiny $2$}};
   \node[normalvertex,label=below:{$(v,e)$},fill=blue!80] (v) at (4,0) {};
   \node[normalvertex,fill=blue!80] (3) at (4,1) {{\tiny $3$}};
   \node[normalvertex,fill=yellow!80] (1) at (4,3) {{\tiny $1$}};
   \node[normalvertex,fill=red!80] (0) at (4,4) {{\tiny $0$}};
   \node[normalvertex,label=above:{$(u,e)$},fill=red!80] (u) at (4,5) {};
   \node[normalvertex,label=right:{$(f_2,e)$},fill=yellow!80] (f2) at (5,3) {};
   \node[normalvertex,fill=red!80] (4) at (2.5,4) {{\tiny $4$}};
   \node[normalvertex,fill=yellow!80] (5) at (2.5,3.33) {{\tiny $5$}};
   \node[normalvertex,fill=blue!80] (6) at (2.5,2.67) {{\tiny $6$}};
   
   \foreach \i/\j in {0/1,2/3,2/9,3/7,4/5,5/6,6/8,6/9,7/8,8/9}{
    \draw[thick] (\i) edge (\j);
   }
   \foreach \i/\j in {u/0,v/3,f1/7,f2/1,1/2,0/4}{
    \draw[gray, dotted,line width = 2pt] (\i) edge (\j);
   }
  \end{tikzpicture}
  \caption{The gadget used in Lemma \ref{la:lower-bound-4-colors-with-eq} where equality edges are represented by dashed edges.
   It has four outgoing vertices $(u,e)$, $(v,e)$, $(f_1,e)$ and $(f_2,e)$.
   The gadget enforces that $(f_1,e)$ and $(f_2,e)$ receive the same color, and $(u,e)$, $(v,e)$ and $(f_1,e)$ receive pairwise distinct colors.}
  \label{fig:cross-gadget-4-colors}
 \end{figure}
 Clearly, $H^{+\eqRel}$ is planar and has maximum degree $3$.
 Also, it is easy to see that the pair $(H,\eqRel)$ can be computed in polynomial time.
 
 Suppose that $G$ is $3$-colorable via a coloring $\mu\colon V(G) \rightarrow \{1,2,3\}$.
 We define a coloring $\chi\colon V(H) \rightarrow \{1,2,3,4\}$ as follows.
 Let $e \in E(G)$ and suppose that $e = uv$.
 Also let $f_1,f_2 \in F$ denote the two faces incident to $e$.
 Suppose that $a \coloneqq \mu(u)$, $b \coloneqq \mu(v)$ and $\{a,b,c\} = \{1,2,3\}$.
 We set
 \begin{itemize}
  \item $\chi(u,e) = \chi(e,0) = \chi(e,4) = \chi(e,8) \coloneqq a$,
  \item $\chi(v,e) = \chi(e,3) = \chi(e,6) \coloneqq b$,
  \item $\chi(e,9) \coloneqq c$, and
  \item $\chi(f_1,e) = \chi(f_2,e) = \chi(e,1) = \chi(e,2) = \chi(e,5) = \chi(e,7) \coloneqq 4$.
 \end{itemize}
 It is easy to verify that the coloring $\chi$ satisfies the desired properties (see also Figures \ref{fig:planar-hardness-4-colors} and \ref{fig:cross-gadget-4-colors}).
 
 In the other direction, suppose that $\chi \colon V(H) \rightarrow \{1,2,3,4\}$ is a coloring with the desired properties.
 
 \begin{claim}
  Let $e \in E(G)$ and let $f_1,f_2$ denote the two incident faces.
  Then $\chi(f_1,e) = \chi(f_2,e)$.
 \end{claim}
 \begin{claimproof}
  By Equation \eqref{eq:rel-edges} we get that that $\chi(f_1,e) = \chi(e,7)$ and $\chi(f_2,e) = \chi(e,1) = \chi(e,2)$.
  Also, $\dist_H((e,i),(e,j)) \leq 2$ for all $i,j \in \{6,7,8,9\}$.
  This means that
  \[\{\chi(e,6),\chi(e,7),\chi(e,8),\chi(e,9)\} = \{1,2,3,4\}.\]
  By the same argument
  \[\{\chi(e,2),\chi(e,6),\chi(e,8),\chi(e,9)\} = \{1,2,3,4\}.\]
  Combining both equations implies that $\chi(e,7) = \chi(e,2)$.
  Overall, we get that $\chi(f_1,e) = \chi(f_2,e)$.
 \end{claimproof}

 Combining the last claim and Equation \eqref{eq:rel-faces} we obtain that all vertices from set $\{(f,e) \mid f \in F, e \in E(G), e \text{ is incident to } f\}$ receive the same color under $\chi$.
 Without loss of generality suppose that $\chi(f,e) = 4$ for all $e \in E(G)$ and incident faces $f$.
 
 Using Equation \eqref{eq:rel-vertices}, for every $v \in V(G)$, there is some color $\mu(v) \in \{1,2,3,4\}$ such that $\chi(v,e) = \mu(v)$ for all incident edges $e \in E(G)$.
 Also $\chi(v,e) \neq 4$ since $\chi(e,1) = 4$ by Equation \eqref{eq:rel-edges}.
 So $\mu(v) \in \{1,2,3\}$ for every $v \in V(G)$.
 It remains to show that adjacent vertices receive different colors.
 So suppose that $e = uv \in E(G)$.
 We have that $\chi(e,3) = \mu(v)$ using Equation \eqref{eq:rel-edges}.
 Also, $\dist_H((e,i),(e,j)) \leq 2$ for all $i,j \in \{6,7,8,9\}$.
 This means that
 \[\{\chi(e,6),\chi(e,7),\chi(e,8),\chi(e,9)\} = \{1,2,3,4\}.\]
 Since $\dist_H((e,i),(e,3)) \leq 2$ for all $i \in \{7,8,9\}$, we conclude that $\chi(e,6) = \chi(e,3) = \mu(v)$.
 Finally $\mu(u) = \chi(e,4)$ using Equation \eqref{eq:rel-edges}.
 Since $\dist_H((e,3),(e,4)) \leq 2$ it follows that $\mu(u) \neq \mu(v)$.
\end{proof}

\subsection{More Than Four Colors}

Next, we implement the first step for $q \geq 5$, i.e., we prove a varaint of Lemma \ref{la:lower-bound-4-colors-with-eq} for $q \geq 5$.
Before getting to the actual reduction, let us briefly explain why a different construction is required.
Intuitively speaking, the main idea for all reductions is to enforce that the ``original'' vertices of the input graph can only be colored with one of three ``candidate'' colors.
Hence, the information about which three colors are those candidates needs to distributed over the entire graph.
Alternatively, it suffices to distribute the information which $q-3$ colors are the ``auxiliary'' colors.
For $q = 4$ this means that we only need to distribute one auxiliary color which turns out to be fairly easy.
However, for larger numbers of colors, we require a more intricate distribution strategy where the idea is to forward the set of ``auxiliary'' colors along a cycle.
In order to ensure that the set of ``auxiliary'' colors is available at every edge of the input graph, we start by showing the following lemma.
Intuitively speaking, given a planar graph $G$ together with an embedding of $G$, it constructs a cycle in the plane that crosses every edge of $G$ exactly twice (see also Figure \ref{fig:planar-tour-step-6}).

\begin{lemma}
 \label{la:edge-covering-tour-planar}
 Let $G$ be a connected planar graph of minimum degree $3$.
 Let $G'$ be the graph defined via $V(G') \coloneqq V(G) \cup \vec{E}(G)$, where $\vec{E}(G) \coloneqq \{(u,v) \mid uv \in E(G)\}$, and 
 \[E(G') \coloneqq \{u(u,v),\;(u,v)(v,u),\;(v,u)v \mid uv \in E(G)\}.\]
 Then there is a set $E^* \subseteq \binom{\vec{E}(G)}{2}$ such that
 \begin{enumerate}
  \item the graph $C = (\vec{E}(G),E^*)$ is a cycle, and
  \item the multigraph $G^+ = (V(G'),E(G') \cup E^*)$ is planar via an embedding where, for every $(u,v) \in \vec{E}(G)$, its four incident edges $e_1,\dots,e_4$, listed in the cyclic order of the embedding, are alternately contained in the sets $E(G')$ and $E^*$.
 \end{enumerate}
 Moreover, given the graph $G$, the set $E^*$ together with the desired embedding of $G^+$ can be computed in polynomial time.
\end{lemma}

Before diving into the proof, let us briefly clarify the last part of lemma.
Consider the multigraph $G^+ = (V(G'),E(G') \cup E^*)$.
Every vertex $(u,v) \in \vec{E}(G)$ in this graph has four incident edges, two of which are contained in the set $E(G')$ and the other two are being contained in $E^*$.
The lemma guarantees that there is a planar embedding of the multigraph $G^+$ such that these two types of edges always alternate when looking at the cyclic order $e_1,\dots,e_4$ induced by the embedding (see also Figure \ref{fig:planar-tour-step-6}).

\begin{figure}
 \centering
 \begin{subfigure}[t]{0.32\linewidth}
  \begin{tikzpicture}[scale=0.8]
   \path[use as bounding box] (90:4) -- (210:4) -- (330:4) -- cycle;
   
   \node[normalvertex,fill=red!80] (1) at (0,0) {};
   \node[normalvertex,fill=red!80] (2) at (90:3) {};
   \node[normalvertex,fill=red!80] (3) at (210:3) {};
   \node[normalvertex,fill=red!80] (4) at (330:3) {};
   
   \foreach \i/\j in {1/2,1/3,1/4,2/3,2/4,3/4}{
    \draw[thick] (\i) edge (\j);
   }
   
   \node (12) at (90:1.5) {};
   \node (13) at (210:1.5) {};
   \node (14) at (330:1.5) {};
   \node (23) at (150:1.5) {};
   \node (24) at (30:1.5) {};
   \node (34) at (270:1.5) {};
   \draw[line width = 2pt,white] (23) ..controls (105:4.5) and (75:4.5).. node[shift=(90:0.2)] {{\small\color{white} $1$}} (24);
   \draw[line width = 2pt,white] (23) ..controls (195:4.5) and (225:4.5).. node[shift=(210:0.2)] {{\small\color{white} $3$}} (34);
   \draw[line width = 2pt,white] (24) ..controls (345:4.5) and (315:4.5).. node[shift=(330:0.2)] {{\small\color{white} $2$}} (34);
  \end{tikzpicture}
  \caption{The input graph $G$.}
  \label{fig:planar-tour-step-1}
 \end{subfigure}
 \hfill
 \begin{subfigure}[t]{0.32\linewidth}
  \begin{tikzpicture}[scale=0.8]
   \path[use as bounding box] (90:4) -- (210:4) -- (330:4) -- cycle;
   
   \node[normalvertex,fill=red!80] (1) at (0,0) {};
   \node[normalvertex,fill=red!80] (2) at (90:3) {};
   \node[normalvertex,fill=red!80] (3) at (210:3) {};
   \node[normalvertex,fill=red!80] (4) at (330:3) {};
   
   \foreach \i/\j in {1/2,1/3,1/4,2/3,2/4,3/4}{
    \draw[thick] (\i) edge (\j);
   }
   
   \node[normalvertex,fill=Green!80] (12) at (90:1.5) {};
   \node[normalvertex,fill=Green!80] (13) at (210:1.5) {};
   \node[normalvertex,fill=Green!80] (14) at (330:1.5) {};
   \node[normalvertex,fill=Green!80] (23) at (150:1.5) {};
   \node[normalvertex,fill=Green!80] (24) at (30:1.5) {};
   \node[normalvertex,fill=Green!80] (34) at (270:1.5) {};
   
   \draw[red,line width = 2pt] ($(210:1.5) + (45:0.3)$) -- ($(210:1.5) + (225:0.3)$);
   \draw[red,line width = 2pt] ($(210:1.5) + (135:0.3)$) -- ($(210:1.5) + (315:0.3)$);
   
   \draw[red,line width = 2pt] ($(150:1.5) + (45:0.3)$) -- ($(150:1.5) + (225:0.3)$);
   \draw[red,line width = 2pt] ($(150:1.5) + (135:0.3)$) -- ($(150:1.5) + (315:0.3)$);
   
   \draw[line width = 2pt,blue] (12) edge node[shift=(330:0.15)] {{\small $9$}} (13);
   \draw[line width = 2pt,blue] (12) edge node[shift=(120:0.15)] {{\small $4$}} (23);
   \draw[line width = 2pt,blue] (13) edge node[shift=(180:0.2)] {{\small $12$}} (23);
   \draw[line width = 2pt,blue] (12) edge node[shift=(210:0.2)] {{\small $10$}} (14);
   \draw[line width = 2pt,blue] (12) edge node[shift=(60:0.15)] {{\small $5$}} (24);
   \draw[line width = 2pt,blue] (14) edge node[shift=(0:0.15)] {{\small $6$}} (24);
   \draw[line width = 2pt,blue] (13) edge node[shift=(90:0.15)] {{\small $11$}} (14);
   \draw[line width = 2pt,blue] (13) edge node[shift=(240:0.15)] {{\small $8$}} (34);
   \draw[line width = 2pt,blue] (14) edge node[shift=(300:0.15)] {{\small $7$}} (34);
   \draw[line width = 2pt,blue] (23) ..controls (105:4.5) and (75:4.5).. node[shift=(90:0.2)] {{\small $1$}} (24);
   \draw[line width = 2pt,blue] (23) ..controls (195:4.5) and (225:4.5).. node[shift=(210:0.2)] {{\small $3$}} (34);
   \draw[line width = 2pt,blue] (24) ..controls (345:4.5) and (315:4.5).. node[shift=(330:0.2)] {{\small $2$}} (34);
   
  \end{tikzpicture}
  \caption{We obtain $G''$ by subdividing every edge. The set $\widetilde{E}$ is marked in blue where the numbers $1,\dots,12$ describe an Euler tour of $\widetilde{G}$. The Euler tour has two ``crossings'' at the marked vertices.}
  \label{fig:planar-tour-step-2}
 \end{subfigure}
 \hfill
 \begin{subfigure}[t]{0.32\linewidth}
  \begin{tikzpicture}[scale=0.8]
   \path[use as bounding box] (90:4) -- (210:4) -- (330:4) -- cycle;
   
   \node[normalvertex,fill=red!80] (1) at (0,0) {};
   \node[normalvertex,fill=red!80] (2) at (90:3) {};
   \node[normalvertex,fill=red!80] (3) at (210:3) {};
   \node[normalvertex,fill=red!80] (4) at (330:3) {};
   
   \foreach \i/\j in {1/2,1/3,1/4,2/3,2/4,3/4}{
    \draw[thick] (\i) edge (\j);
   }
   
   \node[normalvertex,fill=Green!80] (12) at (90:1.5) {};
   \node[normalvertex,fill=Green!80] (13) at (210:1.5) {};
   \node[normalvertex,fill=Green!80] (14) at (330:1.5) {};
   \node[normalvertex,fill=Green!80] (23) at (150:1.5) {};
   \node[normalvertex,fill=Green!80] (24) at (30:1.5) {};
   \node[normalvertex,fill=Green!80] (34) at (270:1.5) {};
   
   \draw[line width = 2pt,blue] (12) edge node[shift=(330:0.15)] {{\small $11$}} (13);
   \draw[line width = 2pt,blue] (12) edge node[shift=(120:0.15)] {{\small $4$}} (23);
   \draw[line width = 2pt,blue] (13) edge node[shift=(180:0.2)] {{\small $12$}} (23);
   \draw[line width = 2pt,blue] (12) edge node[shift=(210:0.2)] {{\small $10$}} (14);
   \draw[line width = 2pt,blue] (12) edge node[shift=(60:0.15)] {{\small $5$}} (24);
   \draw[line width = 2pt,blue] (14) edge node[shift=(0:0.15)] {{\small $6$}} (24);
   \draw[line width = 2pt,blue] (13) edge node[shift=(90:0.15)] {{\small $9$}} (14);
   \draw[line width = 2pt,blue] (13) edge node[shift=(240:0.15)] {{\small $8$}} (34);
   \draw[line width = 2pt,blue] (14) edge node[shift=(300:0.15)] {{\small $7$}} (34);
   \draw[line width = 2pt,blue] (23) ..controls (105:4.5) and (75:4.5).. node[shift=(90:0.2)] {{\small $3$}} (24);
   \draw[line width = 2pt,blue] (23) ..controls (195:4.5) and (225:4.5).. node[shift=(210:0.2)] {{\small $1$}} (34);
   \draw[line width = 2pt,blue] (24) ..controls (345:4.5) and (315:4.5).. node[shift=(330:0.2)] {{\small $2$}} (34);
   
  \end{tikzpicture}
  \caption{By reordering the edges on the Euler tour, we can always obtain an Euler tour without any ``crossings''.}
  \label{fig:planar-tour-step-3}
 \end{subfigure}
 
 \vspace{10pt}
 
 \begin{subfigure}[t]{0.32\linewidth}
  \begin{tikzpicture}[scale=0.8]
   \path[use as bounding box] (90:4) -- (210:4) -- (330:4) -- cycle;
   
   \node[normalvertex,fill=red!80] (1) at (0,0) {};
   \node[normalvertex,fill=red!80] (2) at (90:3) {};
   \node[normalvertex,fill=red!80] (3) at (210:3) {};
   \node[normalvertex,fill=red!80] (4) at (330:3) {};
   
   \foreach \i/\j in {1/2,1/3,1/4,2/3,2/4,3/4}{
    \draw[thick] (\i) edge (\j);
   }
   
   \node[normalvertex,fill=Green!80] (12) at (90:1) {};
   \node[normalvertex,fill=Green!80] (13) at (210:1) {};
   \node[normalvertex,fill=Green!80] (14) at (330:1) {};
   \node[normalvertex,fill=Green!80] (21) at (90:2) {};
   \node[normalvertex,fill=Green!80] (31) at (210:2) {};
   \node[normalvertex,fill=Green!80] (41) at (330:2) {};
   
   \node[normalvertex,fill=Green!80] (23) at ($(90:3) + (240:1.732)$) {};
   \node[normalvertex,fill=Green!80] (32) at ($(90:3) + (240:3.464)$) {};
   \node[normalvertex,fill=Green!80] (24) at ($(90:3) + (300:1.732)$) {};
   \node[normalvertex,fill=Green!80] (42) at ($(90:3) + (300:3.464)$) {};
   \node[normalvertex,fill=Green!80] (34) at ($(210:3) + (0:1.732)$) {};
   \node[normalvertex,fill=Green!80] (43) at ($(210:3) + (0:3.464)$) {};
   
   \foreach \v in {13,31,24,42,34,43}{
    \draw[red,line width = 2pt] ($(\v) + (75:0.3)$) -- ($(\v) + (285:0.3)$);
    \draw[red,line width = 2pt] ($(\v) + (105:0.3)$) -- ($(\v) + (255:0.3)$);
   }
   
   \foreach \i/\j in {23/21,21/42,42/41,41/34,34/31,31/14,14/12,12/13,13/32}{
    \draw[line width = 2pt,blue] (\i) edge (\j);
   }
   \draw[line width = 2pt,blue] (32) ..controls (195:4.5) and (225:4.5).. (43);
   \draw[line width = 2pt,blue] (43) ..controls (315:4.5) and (345:4.5).. (24);
   \draw[line width = 2pt,blue] (24) ..controls (75:4.5) and (105:4.5).. (23);
   
  \end{tikzpicture}
  \caption{By splitting every subdivision vertex (the green vertices), we obtain the graph $G'$ and a cycle $\widehat{E}$ (the blue edges) on $\vec{E}(G)$.
   For the marked vertices, the incident edges are not alternating between $\widehat{E}$ and $E(G')$.}
  \label{fig:planar-tour-step-4}
 \end{subfigure}
 \hfill
 \begin{subfigure}[t]{0.32\linewidth}
  \begin{tikzpicture}[scale=0.8]
   \path[use as bounding box] (90:4) -- (210:4) -- (330:4) -- cycle;
   
   \node[normalvertex,fill=red!80] (1) at (0,0) {};
   \node[normalvertex,fill=red!80] (2) at (90:3) {};
   \node[normalvertex,fill=red!80] (3) at (210:3) {};
   \node[normalvertex,fill=red!80] (4) at (330:3) {};
   
   \foreach \i/\j in {1/2,1/3,1/4,2/3,2/4,3/4}{
    \draw[thick] (\i) edge (\j);
   }
   
   \node[normalvertex,fill=Green!80] (12) at (90:1) {};
   \node[normalvertex,fill=Green!80] (13) at (210:1) {};
   \node[normalvertex,fill=Green!80] (14) at (330:1) {};
   \node[normalvertex,fill=Green!80] (21) at (90:2) {};
   \node[normalvertex,fill=Green!80] (31) at (210:2) {};
   \node[normalvertex,fill=Green!80] (41) at (330:2) {};
   
   \node[normalvertex,fill=Green!80] (23) at ($(90:3) + (240:1.732)$) {};
   \node[normalvertex,fill=Green!80] (32) at ($(90:3) + (240:3.464)$) {};
   \node[normalvertex,fill=Green!80] (24) at ($(90:3) + (300:1.732)$) {};
   \node[normalvertex,fill=Green!80] (42) at ($(90:3) + (300:3.464)$) {};
   \node[normalvertex,fill=Green!80] (34) at ($(210:3) + (0:1.732)$) {};
   \node[normalvertex,fill=Green!80] (43) at ($(210:3) + (0:3.464)$) {};
   
   \foreach \v in {13,31,34,43}{
    \draw[red,line width = 2pt] ($(\v) + (75:0.3)$) -- ($(\v) + (285:0.3)$);
    \draw[red,line width = 2pt] ($(\v) + (105:0.3)$) -- ($(\v) + (255:0.3)$);
   }
   
   \foreach \i/\j in {23/21,21/41,41/34,34/31,31/14,14/12,12/13,13/32}{
    \draw[line width = 2pt,blue] (\i) edge (\j);
   }
   \draw[line width = 2pt,blue] (32) ..controls (195:4.5) and (225:4.5).. (43);
   \draw[line width = 2pt,blue] (43) ..controls (315:4.5) and (345:4.5).. (42);
   \draw[line width = 2pt,blue] (24) ..controls (75:4.5) and (105:4.5).. (23);
   
   \draw[line width = 2pt,blue,bend left] (42) edge (24);
  \end{tikzpicture}
  \caption{By locally modifying the blue edges, the constructed cycle is fixed step by step to obtain the desired outcome.}
  \label{fig:planar-tour-step-5}
 \end{subfigure}
 \hfill
 \begin{subfigure}[t]{0.32\linewidth}
  \begin{tikzpicture}[scale=0.8]
   \path[use as bounding box] (90:4) -- (210:4) -- (330:4) -- cycle;
   
   \node[normalvertex,fill=red!80] (1) at (0,0) {};
   \node[normalvertex,fill=red!80] (2) at (90:3) {};
   \node[normalvertex,fill=red!80] (3) at (210:3) {};
   \node[normalvertex,fill=red!80] (4) at (330:3) {};
   
   \foreach \i/\j in {1/2,1/3,1/4,2/3,2/4,3/4}{
    \draw[thick] (\i) edge (\j);
   }
   
   \node[normalvertex,fill=Green!80] (12) at (90:1) {};
   \node[normalvertex,fill=Green!80] (13) at (210:1) {};
   \node[normalvertex,fill=Green!80] (14) at (330:1) {};
   \node[normalvertex,fill=Green!80] (21) at (90:2) {};
   \node[normalvertex,fill=Green!80] (31) at (210:2) {};
   \node[normalvertex,fill=Green!80] (41) at (330:2) {};
   
   \node[normalvertex,fill=Green!80] (23) at ($(90:3) + (240:1.732)$) {};
   \node[normalvertex,fill=Green!80] (32) at ($(90:3) + (240:3.464)$) {};
   \node[normalvertex,fill=Green!80] (24) at ($(90:3) + (300:1.732)$) {};
   \node[normalvertex,fill=Green!80] (42) at ($(90:3) + (300:3.464)$) {};
   \node[normalvertex,fill=Green!80] (34) at ($(210:3) + (0:1.732)$) {};
   \node[normalvertex,fill=Green!80] (43) at ($(210:3) + (0:3.464)$) {};
   
   \foreach \i/\j in {23/21,21/41,41/31,13/14,14/12,12/32}{
    \draw[line width = 2pt,blue] (\i) edge (\j);
   }
   \draw[line width = 2pt,blue] (32) ..controls (195:4.5) and (225:4.5).. (34);
   \draw[line width = 2pt,blue] (43) ..controls (315:4.5) and (345:4.5).. (42);
   \draw[line width = 2pt,blue] (24) ..controls (75:4.5) and (105:4.5).. (23);
   
   \draw[line width = 2pt,blue,bend left] (42) edge (24);
   \draw[line width = 2pt,blue,bend left] (34) edge (43);
   \draw[line width = 2pt,blue,bend left] (31) edge (13);
  \end{tikzpicture}
  \caption{Finally, we obtain the desired multigraph $G^+$. The edges from the set $E^*$ are blue.}
  \label{fig:planar-tour-step-6}
 \end{subfigure}
 
 \caption{Visualization of the steps involved in the proof of Lemma \ref{la:edge-covering-tour-planar}. Given a planar graph of minimum degree $3$, we eventually obtain a cycle (shown by the blue edges) that crosses every edge of the input graph exactly twice.}
 \label{fig:planar-tour}
\end{figure}
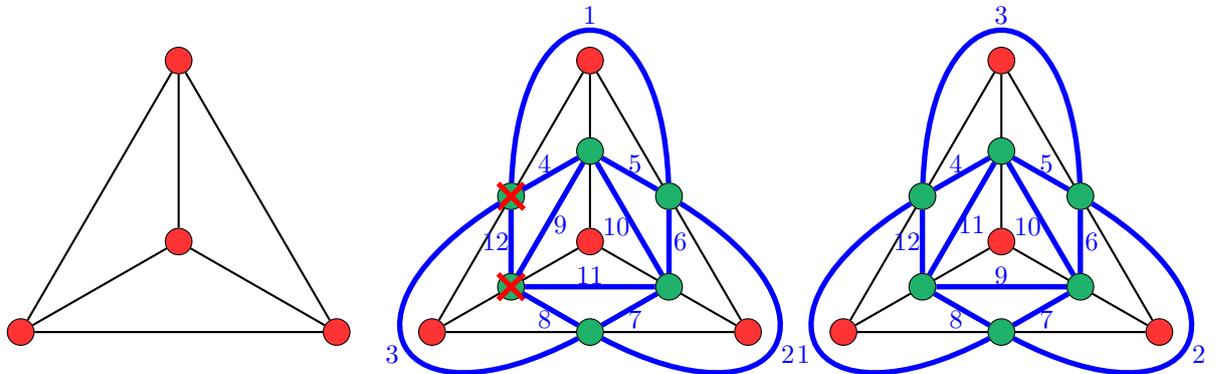

\begin{proof}
 Fix a planar embedding of $G$ and let $F$ denote the set of faces.
 Consider the graph $G''$ that is obtained from $G$ by subdividing every edge once, i.e., $V(G'') \coloneqq V(G) \cup E(G)$ and
 \[E(G'') \coloneqq \{ve \mid e \in E(G), v \in e\}.\]
 We define the set $\widetilde{E} \subseteq \binom{E(G)}{2}$ as follows.
 For every face $f \in F$ let $e_1,\dots,e_k$ denote the incident edges of $f$ ordered cyclically according to the fixed embedding of $G$.
 We add pairs $e_ie_{i+1}$, $i \in [k]$, and $e_1e_k$ to the set $\widetilde{E}$.
 Since $G$ has minimum degree $3$ the graph $\widetilde{G} = (E(G),\widetilde{E})$ is $4$-regular (see also Figure \ref{fig:planar-tour-step-2}).
 Hence, $\widetilde{G}$ has an Euler tour $\widetilde{e}_1,\dots,\widetilde{e}_m$ where $\widetilde{e}_i \in \widetilde{E}$ for every $i \in [m]$.
 
 We can naturally extend the embedding of $G''$ (obtained from the fixed embedding of $G$) to an embedding of $\widetilde{G}$.
 Now, consider again the Euler tour $\widetilde{e}_1,\dots,\widetilde{e}_m$.
 Since $\widetilde{G}$ is $4$-regular, each vertex $v \in V(\widetilde{G})$ is visited twice by the Euler,
 say via edges $\widetilde{e}_i,\widetilde{e}_{i+1}$ and $\widetilde{e}_j,\widetilde{e}_{j+1}$.
 We say that the Euler tour $\widetilde{e}_1,\dots,\widetilde{e}_m$ is \emph{crossing} at $v$ if $\widetilde{e}_i,\widetilde{e}_{i+1}$ are not adjacent in the cyclic order of edges incident to $v$ in $\widetilde{G}$ with respect to the fixed embedding of $\widetilde{G}$ (see Figure \ref{fig:planar-tour-step-2}).
 By reordering the edges of the Euler tour, we may assume without loss of generality that $\widetilde{e}_1,\dots,\widetilde{e}_m$ is not crossing at any vertex $v \in V(\widetilde{G})$ (see Figure \ref{fig:planar-tour-step-3}).
 
 Now, since each element of $uv \in E(G)$ corresponds to two elements $(u,v),(v,u) \in \vec{E}(G)$, we can transform the Euler tour $\widetilde{e}_1,\dots,\widetilde{e}_m$ into a cycle $\widehat{E} \subseteq \binom{\vec{E}(G)}{2}$ on $\vec{E}(G)$ in the natural way which provides a planar embedding of $\widehat{G} = (V(G'),E(G') \cup \widehat{E})$ (see Figure \ref{fig:planar-tour-step-4}).
 Note that every vertex $\vec{e} \in \vec{E}(G) \subseteq V(\widehat{G})$ has degree $4$ in $\widehat{G}$ with two incident edges coming from $E(G')$ and the other two coming from $\widehat{E}$.
 
 To complete the proof, it only remains to ensure that these edges appear alternately (see Figure \ref{fig:planar-tour-step-4}).
 This can be achieved as follows.
 Consider a pair $(u,v),(v,u) \in \vec{E}(G)$.
 Either edges from $E(G')$ and $\widehat{E}$ already appear alternately along the cyclic order associated with the embedding for both vertices $(u,v)$ and $(v,u)$, or this condition is violated for both of them.
 In the latter case, we locally modify the set $\widehat{E}$ as follows.
 First, we omit $(u,v)$ from the cycle and connect its two adjacent vertices via a new edge that is added to $\widehat{E}$.
 Afterwards, we ``pull the cycle defined by $\widehat{E}$ over the edge $uv \in E(G)$'', i.e., we replace the occurrence of $(v,u)$ by the pair $(u,v),(v,u)$.
 In particular, this adds $(u,v)(v,u)$ to the set $\widehat{E}$, creating a multiedge in the graph $\widehat{G}$.
 This multiedge can be placed on either side of the same edge $(u,v)(v,u) \in E(G')$ which means that there is always an embedding such that incident edges of $(u,v)$ alternately come from the sets $E(G')$ and $\widehat{E}$ as desired (see Figure \ref{fig:planar-tour-step-5}).
 
 We perform this modification for all pairs $(u,v),(v,u) \in \vec{E}(G)$ to obtain the desired final outcome $E^*$ together with an embedding of $G^+$ (see Figure \ref{fig:planar-tour-step-6}).
 
 We complete the proof by observing that all steps can be performed in polynomial time.
\end{proof}

\begin{lemma}
 \label{la:lower-bound-many-colors-with-eq}
 Let $q \geq 5$.
 Let $G$ be a connected planar graph of minimum degree $3$. Then there is a graph $H$ and a set $\eqRel \subseteq \binom{V(H)}{2}$ such that the following conditions are satisfied:
 \begin{enumerate}
  \item $|V(H)| = O(q \cdot |V(G)|)$,
  \item $|\eqRel| = O(|V(G)|)$,
  \item $H^{+\eqRel} \coloneqq (V(H),E(H) \cup \eqRel)$ is planar and has maximum degree $q-1$, and
  \item $G$ is $3$-colorable if and only if there is coloring $\chi\colon V(H) \rightarrow [q]$ such that
   \begin{enumerate}
    \item $\chi(u) = \chi(v)$ for all $uv \in \eqRel$, and
    \item $\chi(u) \neq \chi(v)$ for all distinct $u,v \in V(H)$ such that $\dist_H(u,v) \leq 2$.
   \end{enumerate}
 \end{enumerate}
 Moreover, given the graph $G$, the pair $(H,\eqRel)$ can be computed in polynomial time.
\end{lemma}

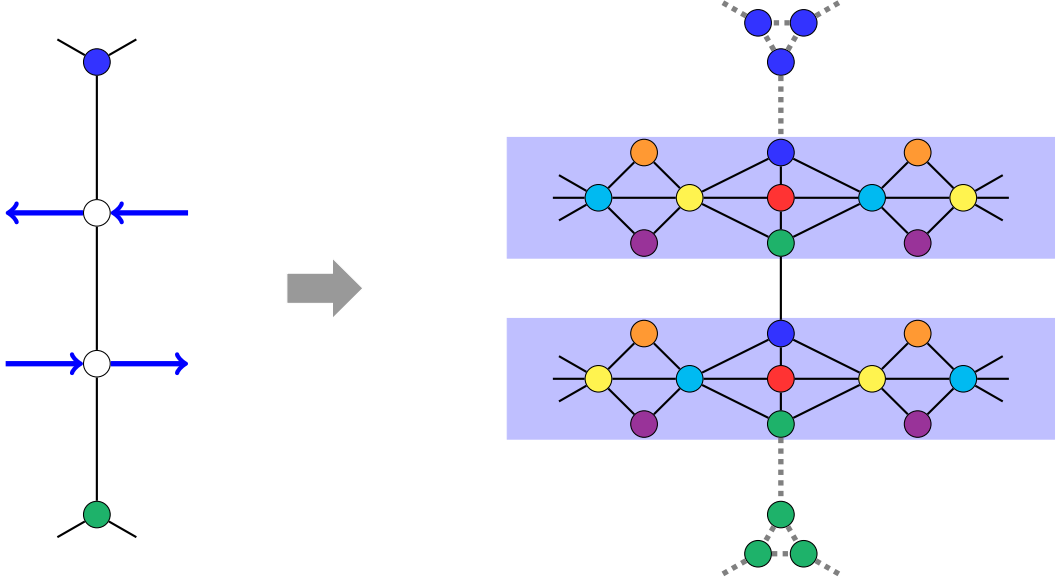
\begin{figure}
 \centering
 \begin{tikzpicture}
  
  \node[normalvertex,fill=blue!80] (a1) at (-8,6) {};
  \node[normalvertex] (a2) at (-8,4) {};
  \node[normalvertex] (a3) at (-8,2) {};
  \node[normalvertex,fill=Green!80] (a4) at (-8,0) {};

  \foreach \i/\j in {1/2,2/3,3/4}{
   \draw[thick] (a\i) edge (a\j);
  }
  \draw[thick] (a1) edge ($(a1.center)+(30:0.6)$);
  \draw[thick] (a1) edge ($(a1.center)+(150:0.6)$);
  \draw[<-,line width = 2pt,blue] (a2) edge ($(a2.center)+(1.2,0)$);
  \draw[->,line width = 2pt,blue] (a2) edge ($(a2.center)+(-1.2,0)$);
  \draw[->,line width = 2pt,blue] (a3) edge ($(a3.center)+(1.2,0)$);
  \draw[<-,line width = 2pt,blue] (a3) edge ($(a3.center)+(-1.2,0)$);
  \draw[thick] (a4) edge ($(a4.center)+(210:0.6)$);
  \draw[thick] (a4) edge ($(a4.center)+(330:0.6)$);
  
  \draw[white,fill=gray!80] (-5.5,3.2) -- (-4.9,3.2) -- (-4.9,3.4) -- (-4.5,3) -- (-4.9,2.6) -- (-4.9,2.8) -- (-5.5,2.8) -- cycle;

  \draw[blue!25,fill=blue!25] (-2.6,3.4) rectangle (4.6,5);
  \draw[blue!25,fill=blue!25] (-2.6,1) rectangle (4.6,2.6);
  
  \node[normalvertex,fill=blue!80] (12) at ($(1,6)+(0:0)$) {};
  \node[normalvertex,fill=blue!80] (13) at ($(1,6)+(60:0.6)$) {};
  \node[normalvertex,fill=blue!80] (14) at ($(1,6)+(120:0.6)$) {};
  
  \node[normalvertex,fill=blue!80] (22) at ($(1,4)+(0,0.8)$) {};
  \node[normalvertex,fill=red!80] (23) at ($(1,4)+(0,0.2)$) {};
  \node[normalvertex,fill=Green!80] (24) at ($(1,4)+(0,-0.4)$) {};
  
  \node[normalvertex,fill=blue!80] (32) at ($(1,2)+(0,0.4)$) {};
  \node[normalvertex,fill=red!80] (33) at ($(1,2)+(0,-0.2)$) {};
  \node[normalvertex,fill=Green!80] (34) at ($(1,2)+(0,-0.8)$) {};
  
  \node[normalvertex,fill=Green!80] (41) at ($(1,0)+(0:0)$) {};
  \node[normalvertex,fill=Green!80] (42) at ($(1,0)+(240:0.6)$) {};
  \node[normalvertex,fill=Green!80] (43) at ($(1,0)+(300:0.6)$) {};
  
  \foreach[count=\i] \x/\y in {-0.8/4.2,2.8/4.2}{
   \node[normalvertex,fill=yellow!80] (e\i1) at ($(\x,\y)+(0.6,0)$) {};
   \node[normalvertex,fill=cyan!80] (e\i2) at ($(\x,\y)+(-0.6,0)$) {};
   \node[normalvertex,fill=orange!80] (e\i3) at ($(\x,\y)+(0,0.6)$) {};
   \node[normalvertex,fill=violet!80] (e\i4) at ($(\x,\y)+(0,-0.6)$) {};
   \foreach \v/\w in {1/2,1/3,1/4,2/3,2/4}{
    \draw[thick] (e\i\v) edge (e\i\w);
   }
  }
  \foreach \i/\x/\y in {3/-0.8/1.8,4/2.8/1.8}{
   \node[normalvertex,fill=cyan!80] (e\i1) at ($(\x,\y)+(0.6,0)$) {};
   \node[normalvertex,fill=yellow!80] (e\i2) at ($(\x,\y)+(-0.6,0)$) {};
   \node[normalvertex,fill=orange!80] (e\i3) at ($(\x,\y)+(0,0.6)$) {};
   \node[normalvertex,fill=violet!80] (e\i4) at ($(\x,\y)+(0,-0.6)$) {};
   \foreach \v/\w in {1/2,1/3,1/4,2/3,2/4}{
    \draw[thick] (e\i\v) edge (e\i\w);
   }
  }
  
  \draw[thick] (e12) edge ($(e12.center)+(150:0.6)$);
  \draw[thick] (e12) edge ($(e12.center)+(180:0.6)$);
  \draw[thick] (e12) edge ($(e12.center)+(210:0.6)$);
  \draw[thick] (e32) edge ($(e32.center)+(150:0.6)$);
  \draw[thick] (e32) edge ($(e32.center)+(180:0.6)$);
  \draw[thick] (e32) edge ($(e32.center)+(210:0.6)$);
  \draw[thick] (e21) edge ($(e21.center)+(30:0.6)$);
  \draw[thick] (e21) edge ($(e21.center)+(0:0.6)$);
  \draw[thick] (e21) edge ($(e21.center)+(-30:0.6)$);
  \draw[thick] (e41) edge ($(e41.center)+(30:0.6)$);
  \draw[thick] (e41) edge ($(e41.center)+(0:0.6)$);
  \draw[thick] (e41) edge ($(e41.center)+(-30:0.6)$);
  
  \foreach \i/\j in {22/23,23/24,24/32,32/33,33/34,e11/22,e11/23,e11/24,e22/22,e22/23,e22/24,e31/32,e31/33,e31/34,e42/32,e42/33,e42/34}{
   \draw[thick] (\i) edge (\j);
  }
  \foreach \i/\j in {12/13,12/14,13/14,12/22,34/41,41/42,41/43,42/43}{
   \draw[gray, dotted,line width = 2pt] (\i) edge (\j);
  }
  
  \draw[gray, dotted,line width = 2pt] (13) edge ($(13.center)+(30:0.6)$);
  \draw[gray, dotted,line width = 2pt] (14) edge ($(14.center)+(150:0.6)$);
  \draw[gray, dotted,line width = 2pt] (42) edge ($(42.center)+(210:0.6)$);
  \draw[gray, dotted,line width = 2pt] (43) edge ($(43.center)+(330:0.6)$);
 \end{tikzpicture}
 \caption{Visualization of the construction from Lemma \ref{la:lower-bound-many-colors-with-eq} on a single edge of $G$. The directed cycle given by $E^*$ is shown in blue and crosses the edge exactly twice.}
 \label{fig:planar-hardness-many-colors}
\end{figure}

\begin{proof}
 Let $G'$ be the graph defined via $V(G') \coloneqq V(G) \cup \vec{E}(G)$, where $\vec{E}(G) \coloneqq \{(u,v) \mid uv \in E(G)\}$, and 
 \[E(G') \coloneqq \{u(u,v),\;(u,v)(v,u),\;(v,u)v \mid uv \in E(G)\}.\]
 By Lemma \ref{la:edge-covering-tour-planar} there is a set $E^* \subseteq \binom{\vec{E}(G)}{2}$ such that
 \begin{enumerate}
  \item the graph $C = (\vec{E}(G),E^*)$ is a cycle, and
  \item the multigraph $G^+ = (V(G'),E(G') \cup E^*)$ is planar via an embedding where, for every $(u,v) \in \vec{E}(G)$, its four incident edges $e_1,\dots,e_4$, listed in the cyclic order of the embedding, are alternately contained in the sets $E(G')$ and $E^*$.
 \end{enumerate}
 Actually, for the remainder of the proof, it is more convenient to assume that $C = (\vec{E}(G),E^*)$ is a directed cycle, i.e., let us arbitrarily direct an edge of $E^*$ and then direct all other edges accordingly following the cycle defined by $E^*$.
 
 We construct the graph $H$ in two steps (see also Figure \ref{fig:planar-hardness-many-colors}).
 First, we define
 \[A \coloneqq \vec{E}(G) \times \{1,2,3\}\]
 and 
 \[B \coloneqq E^* \times \{4,\dots,q\}.\]
 Also, we define
 \begin{align*}
  E_{A,B} \coloneqq\;\;\; &\{(\vec{e},1)(\vec{e},2),(\vec{e},2)(\vec{e},3) \mid \vec{e} \in \vec{E}(G)\}\\
                   \cup\; &\{((u,v),3)((v,u),3) \mid uv \in E(G)\}\\
                   \cup\; &\{(\vec{e}_1,i)((\vec{e}_1,\vec{e}_2),4), (\vec{e}_2,i)((\vec{e}_1,\vec{e}_2),5)\mid (\vec{e}_1,\vec{e}_2) \in E^*, i \in \{1,2,3\}\}\\
                   \cup\; &\{((\vec{e}_1,\vec{e}_2),4)((\vec{e}_1,\vec{e}_2),5)\mid ((\vec{e}_1,\vec{e}_2)) \in E^*\}\\
                   \cup\; &\{((\vec{e}_1,\vec{e}_2),4)((\vec{e}_1,\vec{e}_2),i), ((\vec{e}_1,\vec{e}_2),5)((\vec{e}_1,\vec{e}_2),i) \mid (\vec{e}_1,\vec{e}_2) \in E^*, i \in \{6,\dots,q\}\}
 \end{align*}
 Now consider the graph $H' = (A \cup B,E_{A,B})$.
 
 \begin{claim}
  \label{claim:color-transport-along-cycle}
  Let $\chi'$ be a $q$-coloring of $(H')^2$.
  Then
  \[\{\chi'(\vec{e}_1,i) \mid i \in \{1,2,3\}\} = \{\chi'(\vec{e}_2,i) \mid i \in \{1,2,3\}\}\]
  for all $\vec{e}_1,\vec{e}_2 \in \vec{E}(G)$.
  Moreover,
  \[\chi'((u,v),1) \neq \chi'((v,u),1)\]
  for all $uv \in E(G)$.
 \end{claim}
 \begin{claimproof}
  Suppose that $e^* = (\vec{e}_1,\vec{e}_2) \in E^*$.
  Then the set
  \[(\{\vec{e}_1\} \times \{1,2,3\}) \cup \{e^*\} \times \{4,\dots,q\}\]
  forms a clique in $(H')^2$ and all $q$ colors need to be used under $\chi'$.
  The same statement holds for the 
  \[(\{\vec{e}_2\} \times \{1,2,3\}) \cup \{e^*\} \times \{4,\dots,q\}.\]
  It follows that
  \[\{\chi'(\vec{e}_1,i) \mid i \in \{1,2,3\}\} = \{\chi'(\vec{e}_2,i) \mid i \in \{1,2,3\}\}\]
  Since $C = (\vec{E}(G),E^*)$ forms a cycle the first statement follows.
  
  For the second statement suppose that $uv \in E(G)$.
  Then $\chi'((u,v),3) \neq \chi'((v,u),3)$ and $\chi'((u,v),2) \neq \chi'((v,u),3)$.
  In combination with the first statement it follows that
  \[\chi'((u,v),3) = \chi'((v,u),1)\]
  and hence,
  \[\chi'((u,v),1) \neq \chi'((v,u),1).\]
 \end{claimproof}
 
 We say that a $q$-coloring $\chi'$ of $(H')^2$ is \emph{good} if $\chi'((u,v),1) = \chi'((u,w),1)$ for all $uv,uw \in E(G)$.
 
 \begin{claim}
  \label{claim:good-coloring-to-3-coloring}
  Let $\chi'$ be a good $q$-coloring of $(H')^2$.
  Then $G$ is $3$-colorable.
 \end{claim}
 \begin{claimproof}
  Let $\mu\colon V(G) \rightarrow \{1,2,3\}$ be define in such a way that $\mu(u) = \chi'((u,v),1)$ for all $uv \in E(G)$.
  Observe that this is well-defined since $\chi'$ is good.
  Now, $\mu(u) \neq \mu(v)$ for all $uv \in E(G)$ by Claim \ref{claim:color-transport-along-cycle}.
 \end{claimproof}

 \begin{claim}
  \label{claim:3-coloring-to-good-coloring}
  Let $\mu\colon V(G) \rightarrow \{1,2,3\}$ be a $3$-coloring of $G$.
  Then there is a $q$-coloring $\chi'$ of $(H')^2$ such that
  \[\chi'((u,v),1) = \mu(u)\]
  for all $(u,v) \in \vec{E}(G)$.
 \end{claim}
 \begin{claimproof}
  We define $\chi'(e^*,i) \coloneqq i$ for all $(e^*,i) \in B$. Let $uv \in E(G)$.
  Then $\mu(u) \neq \mu(v)$.
  Let $c \in \{1,2,3\}$ be the unique third color that is distinct from $\mu(u)$ and $\mu(v)$.
  We set
  \begin{itemize}
   \item $\chi'((u,v),1) \coloneqq \mu(u)$,
   \item $\chi'((u,v),2) \coloneqq c$, and
   \item $\chi'((u,v),3) \coloneqq \mu(v)$.
  \end{itemize}
  It can be easily verified that $\chi'$ is a $q$-coloring of $(H')^2$.
 \end{claimproof}
 
 Observe that the coloring defined in the last claim is a good $q$-coloring of $(H')^2$.
 Hence, $G$ is $3$-colorable if and only if $(H')^2$ has a good $q$-coloring.
 So to complete the proof, we use equality edges to ensure that every valid coloring of $H'$ is indeed good.
 We define the graph $H$ via
 \[V(H) \coloneqq V(H') \cup \vec{E}(G)\]
 and
 \[E(H) \coloneqq E(H').\]
 For $u \in V(G)$ let $v_1,\dots,v_d$ denote its neighbors ordered cyclically according to the embedding of $G$ (which is inherited from the embedding of $G^+$ in the natural way).
 We define
 \[\eqRel(u) \coloneqq \{(u,v_i)(u,v_{i+1}) \mid i \in [d-1]\} \cup \{(u,v_1)(u,v_d)\}\]
 and
 \[\eqRel \coloneqq \{(u,v)((u,v),1) \mid (u,v) \in \vec{E}(G)\} \cup \bigcup_{u \in V(G)} \eqRel(u).\]
 This completes the construction.
 We have
 \[|V(H)| = |\vec{E}(G)| + |A| + |B| = |\vec{E}(G)| + 3|\vec{E}(G)| + (q-3)|\vec{E}(G)| = (q+1)|\vec{E}(G)| = O(q|V(G)|)\]
 since $G$ is planar.
 Similarly,
 \[|\eqRel| = 2|\vec{E}(G)| = O(|V(G)|).\]
 We have $\deg_{H^{+\eqRel}}(\vec{e}) = 3$ and $\deg_{H^{+\eqRel}}(\vec{e},i) = 4$ for all $i \in \{1,2,3\}$ and $\vec{e} \in \vec{E}(G)$, and $\deg_{H^{+\eqRel}}(e^*,i) = q-1$ for all $e^* \in E^*$ and $i \in \{4,\dots,q\}$.
 Since $q \geq 5$ it follows that $H^{+\eqRel}$ has maximum degree $q-1$.
 Also, $H^{+\eqRel}$ is planar using the properties guaranteed by Lemma \ref{la:edge-covering-tour-planar}.
 Finally, the last condition follows from Claim \ref{claim:3-coloring-to-good-coloring} and \ref{claim:good-coloring-to-3-coloring}.
 Observe that the equality edges ensure that any valid coloring of $(H')^2$ has to be good.
 
 Finally, it is easy to see that $(H,\eqRel)$ can be computed in polynomial time using that the set $E^*$ can be computed in polynomial time by Lemma \ref{la:edge-covering-tour-planar}.
\end{proof}

\subsection{Removing Equality Edges}

Having completed the first step of the reduction chain, it only remains to replace the equality edges with suitable gadgets.
Here, we can use ideas that already appeared in Section \ref{sec:tw-lower-bound}.
However, let us point out that we can not use exactly the same gadgets as in Section \ref{sec:tw-lower-bound} since we are restricted in the number of colors that are available.
Indeed, suppose $uv \in \eqRel$ is an equality edge.
The vertex $u$ may be adjacent (via normal edges) to $q-2$ further vertices and we need to ensure that the equality gadget does not interfere with coloring those vertices.
In Section \ref{sec:tw-lower-bound}, we resolved this issue by introducing additional ``neutral'' colors which is not possible for the present reduction.
Instead, we use a slightly more complicated gadget (see Figure \ref{fig:eq-gadget}) which, in essence, chains together two equality gadgets (using the construction from Section \ref{sec:tw-lower-bound}), to ensure that every valid coloring of the original vertices can be extended to the gadget vertices.

\begin{lemma}
 \label{la:lower-bound-planar}
 Let $q \geq 4$.
 There is a polynomial-time algorithm that, given a connected planar graph $G$, constructs a planar graph $H$ such that
 \begin{enumerate}
  \item $|V(H)| = O(q \cdot |V(G)|)$, and
  \item $G$ is $3$-colorable if and only if $H^2$ is $q$-colorable.
 \end{enumerate}
\end{lemma}

\begin{proof}
 By repeatedly removing vertices of degree at most $2$ we may assume without loss of generality that $G$ has minimum degree $3$ (if $G$ is $2$-degenerate we return a trivial YES-instance).
 
 By Lemmas \ref{la:lower-bound-4-colors-with-eq} and \ref{la:lower-bound-many-colors-with-eq} there is a graph $H'$ and a set $\eqRel \subseteq \binom{V(H')}{2}$ such that
 \begin{enumerate}
  \item $|V(H')| = O(q \cdot |V(G)|)$,
  \item $|\eqRel| = O(|V(G)|)$,
  \item $(H')^{+\eqRel} \coloneqq (V(H'),E(H') \cup \eqRel)$ is planar and has maximum degree $q-1$, and
  \item $G$ is $3$-colorable if and only if there is coloring $\chi'\colon V(H') \rightarrow \{1,\dots,q\}$ such that
   \begin{enumerate}
    \item $\chi'(u) = \chi'(v)$ for all $uv \in \eqRel$, and
    \item $\chi'(u) \neq \chi'(v)$ for all distinct $u,v \in V(H')$ such that $\dist_{H'}(u,v) \leq 2$.
   \end{enumerate}
 \end{enumerate}
 
 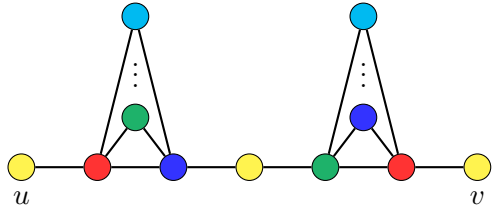
\begin{figure}
  \centering
  \begin{tikzpicture}
   \node[normalvertex,label=below:{$u$},fill=yellow!80] (1) at (0,0) {};
   \node[normalvertex,fill=red!80] (2) at (1,0) {};
   \node[normalvertex,fill=blue!80] (3) at (2,0) {};
   \node[normalvertex,fill=yellow!80] (4) at (3,0) {};
   \node[normalvertex,fill=Green!80] (5) at (4,0) {};
   \node[normalvertex,fill=red!80] (6) at (5,0) {};
   \node[normalvertex,label=below:{$v$},fill=yellow!80] (7) at (6,0) {};
   \node[normalvertex,fill=Green!80] (8) at (1.5,0.66) {};
   \node[normalvertex,fill=blue!80] (9) at (4.5,0.66) {};
   
   \node[normalvertex,fill=cyan!80] (10) at (1.5,2) {};
   \node[normalvertex,fill=cyan!80] (11) at (4.5,2) {};
   
   \node at (1.5,1.33) {$\vdots$};
   \node at (4.5,1.33) {$\vdots$};
   
   \foreach \i/\j in {1/2,2/3,3/4,4/5,5/6,6/7,2/8,3/8,5/9,6/9,2/10,3/10,5/11,6/11}{
    \draw[thick] (\i) edge (\j);
   }
  \end{tikzpicture}
  \caption{The equality gadget ensures that $u$ and $v$ need to be assigned the same color in any square $q$-coloring. Moreover, as long as $u$ (resp.\ $v$) has at most $q-2$ further outside neighbors, every valid coloring of the outside vertices can be extended into the gadget.}
  \label{fig:eq-gadget}
 \end{figure}
 
 We obtain $H$ from $H'$ by inserting an equality gadget between all pairs $uv \in \eqRel$ as follows (see also Figure \ref{fig:eq-gadget}).
 Formally, we define
 \[V(H) \coloneqq V(H') \cup (\eqRel \times \{1,\dots,2q-1\})\]
 and
 \[E(H) \coloneqq E(H') \cup \{u(uv,1),\; v(uv,q+2) \mid uv \in \eqRel\} \cup \{(uv,i)(uv,j) \mid ij \in M, uv \in \eqRel\}\]
 where 
 \begin{align*}
  M =\;\;\; &\{\{1,q-1\},\{q-1,q\},\{q,q+1\},\{q+1,2q-1\}\}\\
     \cup\; &\{\{i,1\},\{i,q-1\},\{q+i,q+1\},\{q+i,2q-1\} \mid i \in \{2,\dots,q-2\}\}
 \end{align*}
 Clearly, $H$ can be computed in polynomial time and $H$ is planar since $(H')^{+\eqRel}$ is planar.
 Also,
 \begin{align*}
  |V(H)| = |V(H')| + (2q-1) \cdot |\eqRel| = O(q \cdot |V(G)|).
 \end{align*}
 So it remains to prove that $G$ is $3$-colorable if and only if $H^2$ is $q$-colorable.
 We start with two basic observations.
 
 \begin{claim}
  \label{claim:distances-preserved}
  Let $u,v \in V(H')$.
  Then $\dist_{H'}(u,v) \leq 2$ if and only if $\dist_{H}(u,v) \leq 2$.
 \end{claim}
 \begin{claimproof}
  This is immediately clear by construction since the inserted gadgets do not allow for any shortcuts for distances at most $2$.
 \end{claimproof}
 
 \begin{claim}
  \label{claim:eq-gadget}
  Let $\chi$ be a $q$-coloring of $H^2$ and let $uv \in \eqRel$.
  Then $\chi(u) = \chi(v)$.
 \end{claim}
 \begin{claimproof}
  Looking at the inserted gadget it is easy to see that $\chi(u) = \chi(uv,q) = \chi(v)$.
 \end{claimproof}

 Now, first suppose that $H^2$ is $q$-colorable via a coloring $\chi\colon V(H) \rightarrow \{1,\dots,q\}$.
 Let $\chi' \coloneqq \chi|_{V(H')}$ be the restriction to $V(H')$.
 Then $\chi'(u) = \chi'(v)$ for all $uv \in \eqRel$ by Claim \ref{claim:eq-gadget}.
 Also, for $u,v \in V(H')$ such that $\dist_{H'}(u,v) \leq 2$, we get that $\dist_{H}(u,v) \leq 2$ by Claim \ref{claim:distances-preserved} which implies that $\chi'(u) \neq \chi'(v)$.
 So $G$ is $3$-colorable by Lemmas \ref{la:lower-bound-4-colors-with-eq} and \ref{la:lower-bound-many-colors-with-eq}.
 
 In the other direction, suppose that $G$ is $3$-colorable.
 By Lemmas \ref{la:lower-bound-4-colors-with-eq} and \ref{la:lower-bound-many-colors-with-eq} there is a coloring $\chi'\colon V(H') \rightarrow \{1,\dots,q\}$ such that
 \begin{enumerate}[label=(\alph*)]
  \item $\chi'(u) = \chi'(v)$ for all $uv \in \eqRel$, and
  \item $\chi'(u) \neq \chi'(v)$ for all distinct $u,v \in V(H')$ such that $\dist_{H'}(u,v) \leq 2$.
 \end{enumerate}
 We claim that we can extend $\chi'$ to a $q$-coloring $\chi$ of $H^2$.
 Formally, we set $\chi(v) \coloneqq \chi'(v)$ for all $v \in V(H')$.
 Then $\chi(u) \neq \chi(v)$ for all $u,v \in V(H')$ such that $\dist_H(u,v) \leq 2$ by Claim \ref{claim:distances-preserved}.
 Now consider some $uv \in \eqRel$.
 Let $a \coloneqq \chi'(u) = \chi'(v)$.
 We can color the vertices of the inserted gadget as follows.
 First, we set $\chi(uv,q) \coloneqq a$.
 Since $H'$ has maximum degree $q-1$ we get that $|N_H[u]| \leq q$.
 This means there is some color $b \in \{1,\dots,q\}$ that is not used so far in $N_H[u]$ (we have $(uv,1) \in N_H[u]$ and this vertex is not colored yet).
 We set $\chi(uv,1) \coloneqq b$.
 Similarly, there is some color $c \in \{1,\dots,q\}$ that is not used so far in $N_H[v]$, and we set $\chi(uv,2q-1) \coloneqq c$.
 The remaining vertices of the gadget are not connected to vertices outside the gadget in $H^2$ and we can easily complete the coloring within the gadget.
 Indeed, there are $q-2$ colors remaining for each of the two sets $\{(uv,i) \mid i \in \{2,\dots,q-1\}\}$ and $\{(uv,i) \mid i \in \{q+1,\dots,2q-2\}\}$ and we only need to ensure that $\chi(uv,q-1) \neq \chi(uv,q+1)$.
 It is easy to see that this is always possible.
\end{proof}

Finally, Theorem \ref{thm:planar-hardness-q} directly follows from Lemma \ref{la:lower-bound-planar} and Theorem \ref{thm:planar-3-col-hardness}.

\section{Conclusion}
\label{sec:conclusion}
We investigated the computational complexity of the \SqCol\ problem on bounded-treewidth and planar graphs.

For bounded-treewidth graphs we presented an algorithm with running time $n^{2^{\ttw + 4}+O(1)}$ where $\ttw$ denotes the treewidth of the input graph.
We also argued that the somewhat unusual exponent $2^{\tw}$ in the running time is essentially optimal by showing that there is no algorithm with running time $f(\tw)n^{(2-\epsilon)^{\tw}}$ for any $\epsilon > 0$ unless \ETH\ fails.
 
For planar graphs, we showed that \SqCol\ can be solved in subexponential time $2^{O(n^{2/3}\log n)}$.
Also, by providing a reduction from \textsc{$3$-Coloring}, we argued that the problem cannot be solved in time $2^{o(\sqrt{n})}$ assuming \ETH. 
However, it remains an open problem to determine the exact complexity of \SqCol\ on planar graphs (assuming \ETH).

Another open question is whether the subexponential algorithm for planar graphs can be generalized to arbitrary $H$-minor-free graphs (for a fixed graph $H$).

Finally, we can also ask for the complexity of determining whether a graph admits a distance-$d$ coloring (i.e., vertices at distance at most $d$ need to receive distinct colors) using $q$ colors for any fixed $d \geq 3$.
For graphs of bounded treewidth a polynomial-time algorithm was presented by Zhou, Kanari, and Nishizeki~\cite{ZhouKN00}.  
What is the precise complexity of computing distance-$d$ colorings on planar (and more generally $H$-minor-free) graphs for some fixed $d \geq 3$?

\bibliographystyle{plainurl}
\small
\bibliography{references}

\begin{thebibliography}{10}

\bibitem{AbboudLW14}
Amir Abboud, Kevin Lewi, and Ryan Williams.
\newblock Losing weight by gaining edges.
\newblock In Andreas~S. Schulz and Dorothea Wagner, editors, {\em Algorithms -
  {ESA} 2014 - 22th Annual European Symposium, Wroclaw, Poland, September 8-10,
  2014. Proceedings}, volume 8737 of {\em Lecture Notes in Computer Science},
  pages 1--12. Springer, 2014.
\newblock \href {https://doi.org/10.1007/978-3-662-44777-2\_1}
  {\path{doi:10.1007/978-3-662-44777-2\_1}}.

\bibitem{DBLP:journals/siamdm/AgnarssonH03}
Geir Agnarsson and Magn{\'{u}}s~M. Halld{\'{o}}rsson.
\newblock Coloring powers of planar graphs.
\newblock {\em {SIAM} J. Discret. Math.}, 16(4):651--662, 2003.
\newblock \href {https://doi.org/10.1137/S0895480100367950}
  {\path{doi:10.1137/S0895480100367950}}.

\bibitem{MR1888178}
Noga Alon and Bojan Mohar.
\newblock The chromatic number of graph powers.
\newblock {\em Combin. Probab. Comput.}, 11(1):1--10, 2002.
\newblock \href {https://doi.org/10.1017/S0963548301004965}
  {\path{doi:10.1017/S0963548301004965}}.

\bibitem{DBLP:journals/endm/BensmailBH15}
Julien Bensmail, Marthe Bonamy, and Herv{\'{e}} Hocquard.
\newblock Strong edge coloring sparse graphs.
\newblock {\em Electron. Notes Discret. Math.}, 49:773--778, 2015.
\newblock \href {https://doi.org/10.1016/j.endm.2015.06.104}
  {\path{doi:10.1016/j.endm.2015.06.104}}.

\bibitem{DBLP:journals/tcs/BlairM12}
Jean R.~S. Blair and Fredrik Manne.
\newblock An efficient self-stabilizing distance-2 coloring algorithm.
\newblock {\em Theor. Comput. Sci.}, 444:28--39, 2012.
\newblock \href {https://doi.org/10.1016/j.tcs.2012.01.034}
  {\path{doi:10.1016/j.tcs.2012.01.034}}.

\bibitem{Bodlaender96}
Hans~L. Bodlaender.
\newblock A linear-time algorithm for finding tree-decompositions of small
  treewidth.
\newblock {\em {SIAM} J. Comput.}, 25(6):1305--1317, 1996.
\newblock \href {https://doi.org/10.1137/S0097539793251219}
  {\path{doi:10.1137/S0097539793251219}}.

\bibitem{DBLP:conf/sirocco/Bodlaender07}
Hans~L. Bodlaender.
\newblock Treewidth: Structure and algorithms.
\newblock In Giuseppe Prencipe and Shmuel Zaks, editors, {\em Structural
  Information and Communication Complexity, 14th International Colloquium,
  {SIROCCO} 2007, Castiglioncello, Italy, June 5-8, 2007, Proceedings}, volume
  4474 of {\em Lecture Notes in Computer Science}, pages 11--25. Springer,
  2007.
\newblock \href {https://doi.org/10.1007/978-3-540-72951-8\_3}
  {\path{doi:10.1007/978-3-540-72951-8\_3}}.

\bibitem{BodlaenderFLPST16}
Hans~L. Bodlaender, Fedor~V. Fomin, Daniel Lokshtanov, Eelko Penninkx, Saket
  Saurabh, and Dimitrios~M. Thilikos.
\newblock (meta) kernelization.
\newblock {\em J. {ACM}}, 63(5):44:1--44:69, 2016.
\newblock \href {https://doi.org/10.1145/2973749} {\path{doi:10.1145/2973749}}.

\bibitem{BousquetMD22}
Nicolas Bousquet, Lucas de~Meyer, Quentin Deschamps, and Th{\'{e}}o Pierron.
\newblock Square coloring planar graphs with automatic discharging.
\newblock {\em CoRR}, abs/2204.05791, 2022.
\newblock \href {http://arxiv.org/abs/2204.05791} {\path{arXiv:2204.05791}},
  \href {https://doi.org/10.48550/arXiv.2204.05791}
  {\path{doi:10.48550/arXiv.2204.05791}}.

\bibitem{MR2974301}
Yuehua Bu and Xubo Zhu.
\newblock An optimal square coloring of planar graphs.
\newblock {\em J. Comb. Optim.}, 24(4):580--592, 2012.
\newblock \href {https://doi.org/10.1007/s10878-011-9409-z}
  {\path{doi:10.1007/s10878-011-9409-z}}.

\bibitem{DBLP:journals/siamdm/ChangK96}
Gerard~J. Chang and David Kuo.
\newblock The l(2, 1)-labeling problem on graphs.
\newblock {\em {SIAM} J. Discret. Math.}, 9(2):309--316, 1996.
\newblock \href {https://doi.org/10.1137/S0895480193245339}
  {\path{doi:10.1137/S0895480193245339}}.

\bibitem{MR3873551}
Lily Chen, Kecai Deng, Gexin Yu, and Xiangqian Zhou.
\newblock Strong edge-coloring for planar graphs with large girth.
\newblock {\em Discrete Math.}, 342(2):339--343, 2019.
\newblock \href {https://doi.org/10.1016/j.disc.2018.10.019}
  {\path{doi:10.1016/j.disc.2018.10.019}}.

\bibitem{MR4078897}
Joanna Chybowska-Sok\'{o}\l, Konstanty Junosza-Szaniawski, and Pawe\l\
  Rz\k{a}\.{z}ewski.
\newblock {$L(2,1)$}-labeling of disk intersection graphs.
\newblock {\em Discrete Appl. Math.}, 277:71--81, 2020.
\newblock \href {https://doi.org/10.1016/j.dam.2019.08.020}
  {\path{doi:10.1016/j.dam.2019.08.020}}.

\bibitem{CyganFKLMPPS15}
Marek Cygan, Fedor~V. Fomin, Lukasz Kowalik, Daniel Lokshtanov, D{\'{a}}niel
  Marx, Marcin Pilipczuk, Michal Pilipczuk, and Saket Saurabh.
\newblock {\em Parameterized Algorithms}.
\newblock Springer, 2015.
\newblock \href {https://doi.org/10.1007/978-3-319-21275-3}
  {\path{doi:10.1007/978-3-319-21275-3}}.

\bibitem{DBLP:journals/ejc/DvorakKNS08}
Zdenek Dvor{\'{a}}k, Daniel Kr{\'{a}}l, Pavel Nejedl{\'{y}}, and Riste
  Skrekovski.
\newblock Coloring squares of planar graphs with girth six.
\newblock {\em Eur. J. Comb.}, 29(4):838--849, 2008.
\newblock \href {https://doi.org/10.1016/j.ejc.2007.11.005}
  {\path{doi:10.1016/j.ejc.2007.11.005}}.

\bibitem{EricksonTB05}
Jeff Erickson, Shripad Thite, and David~P. Bunde.
\newblock Distance-2 edge coloring is np-complete.
\newblock {\em CoRR}, abs/cs/0509100, 2005.
\newblock URL: \url{http://arxiv.org/abs/cs/0509100}, \href
  {http://arxiv.org/abs/cs/0509100} {\path{arXiv:cs/0509100}}.

\bibitem{MR1412876}
R.~J. Faudree, R.~H. Schelp, A.~Gy\'{a}rf\'{a}s, and Zs. Tuza.
\newblock The strong chromatic index of graphs.
\newblock volume~29, pages 205--211. 1990.
\newblock Twelfth British Combinatorial Conference (Norwich, 1989).

\bibitem{DBLP:journals/dam/FialaGKKK18}
Jir{\'{\i}} Fiala, Tomas Gavenciak, Dusan Knop, Martin Kouteck{\'{y}}, and Jan
  Kratochv{\'{\i}}l.
\newblock Parameterized complexity of distance labeling and uniform channel
  assignment problems.
\newblock {\em Discret. Appl. Math.}, 248:46--55, 2018.
\newblock \href {https://doi.org/10.1016/j.dam.2017.02.010}
  {\path{doi:10.1016/j.dam.2017.02.010}}.

\bibitem{DBLP:conf/icalp/FialaGK05}
Jir{\'{\i}} Fiala, Petr~A. Golovach, and Jan Kratochv{\'{\i}}l.
\newblock Distance constrained labelings of graphs of bounded treewidth.
\newblock In Lu{\'{\i}}s Caires, Giuseppe~F. Italiano, Lu{\'{\i}}s Monteiro,
  Catuscia Palamidessi, and Moti Yung, editors, {\em Automata, Languages and
  Programming, 32nd International Colloquium, {ICALP} 2005, Lisbon, Portugal,
  July 11-15, 2005, Proceedings}, volume 3580 of {\em Lecture Notes in Computer
  Science}, pages 360--372. Springer, 2005.
\newblock \href {https://doi.org/10.1007/11523468\_30}
  {\path{doi:10.1007/11523468\_30}}.

\bibitem{DBLP:journals/tcs/FialaGK11}
Jir{\'{\i}} Fiala, Petr~A. Golovach, and Jan Kratochv{\'{\i}}l.
\newblock Parameterized complexity of coloring problems: Treewidth versus
  vertex cover.
\newblock {\em Theor. Comput. Sci.}, 412(23):2513--2523, 2011.
\newblock \href {https://doi.org/10.1016/j.tcs.2010.10.043}
  {\path{doi:10.1016/j.tcs.2010.10.043}}.

\bibitem{FlumG06}
J{\"{o}}rg Flum and Martin Grohe.
\newblock {\em Parameterized Complexity Theory}.
\newblock Texts in Theoretical Computer Science. An {EATCS} Series. Springer,
  2006.
\newblock \href {https://doi.org/10.1007/3-540-29953-X}
  {\path{doi:10.1007/3-540-29953-X}}.

\bibitem{FominLSZ19}
Fedor~V. Fomin, Daniel Lokshtanov, Saket Saurabh, and Meirav Zehavi.
\newblock {\em Kernelization: Theory of Parameterized Preprocessing}.
\newblock Cambridge University Press, 2019.
\newblock \href {https://doi.org/10.1017/9781107415157}
  {\path{doi:10.1017/9781107415157}}.

\bibitem{DBLP:conf/sirocco/FraigniaudHN20}
Pierre Fraigniaud, Magn{\'{u}}s~M. Halld{\'{o}}rsson, and Alexandre Nolin.
\newblock Distributed testing of distance-k colorings.
\newblock In Andrea~Werneck Richa and Christian Scheideler, editors, {\em
  Structural Information and Communication Complexity - 27th International
  Colloquium, {SIROCCO} 2020, Paderborn, Germany, June 29 - July 1, 2020,
  Proceedings}, volume 12156 of {\em Lecture Notes in Computer Science}, pages
  275--290. Springer, 2020.
\newblock \href {https://doi.org/10.1007/978-3-030-54921-3\_16}
  {\path{doi:10.1007/978-3-030-54921-3\_16}}.

\bibitem{DBLP:conf/focs/GhaffariHK18}
Mohsen Ghaffari, David~G. Harris, and Fabian Kuhn.
\newblock On derandomizing local distributed algorithms.
\newblock In Mikkel Thorup, editor, {\em 59th {IEEE} Annual Symposium on
  Foundations of Computer Science, {FOCS} 2018, Paris, France, October 7-9,
  2018}, pages 662--673. {IEEE} Computer Society, 2018.
\newblock \href {https://doi.org/10.1109/FOCS.2018.00069}
  {\path{doi:10.1109/FOCS.2018.00069}}.

\bibitem{DBLP:conf/soda/GolovachL0Z18}
Petr~A. Golovach, Daniel Lokshtanov, Saket Saurabh, and Meirav Zehavi.
\newblock Cliquewidth {III:} the odd case of graph coloring parameterized by
  cliquewidth.
\newblock In Artur Czumaj, editor, {\em Proceedings of the Twenty-Ninth Annual
  {ACM-SIAM} Symposium on Discrete Algorithms, {SODA} 2018, New Orleans, LA,
  USA, January 7-10, 2018}, pages 262--273. {SIAM}, 2018.
\newblock \href {https://doi.org/10.1137/1.9781611975031.19}
  {\path{doi:10.1137/1.9781611975031.19}}.

\bibitem{GroheM09}
Martin Grohe and D{\'{a}}niel Marx.
\newblock On tree width, bramble size, and expansion.
\newblock {\em J. Comb. Theory, Ser. {B}}, 99(1):218--228, 2009.
\newblock \href {https://doi.org/10.1016/j.jctb.2008.06.004}
  {\path{doi:10.1016/j.jctb.2008.06.004}}.

\bibitem{MR999926}
G.~Hal\'{a}sz and V.~T. S\'{o}s, editors.
\newblock {\em Irregularities of partitions}, volume~8 of {\em Algorithms and
  Combinatorics: Study and Research Texts}.
\newblock Springer-Verlag, Berlin, 1989.
\newblock Papers from the meeting held in Fert\H{o}d, July 7--11, 1986.
\newblock \href {https://doi.org/10.1007/978-3-642-61324-1}
  {\path{doi:10.1007/978-3-642-61324-1}}.

\bibitem{DBLP:conf/podc/HalldorssonKM20}
Magn{\'{u}}s~M. Halld{\'{o}}rsson, Fabian Kuhn, and Yannic Maus.
\newblock Distance-2 coloring in the {CONGEST} model.
\newblock In Yuval Emek and Christian Cachin, editors, {\em {PODC} '20: {ACM}
  Symposium on Principles of Distributed Computing, Virtual Event, Italy,
  August 3-7, 2020}, pages 233--242. {ACM}, 2020.
\newblock \href {https://doi.org/10.1145/3382734.3405706}
  {\path{doi:10.1145/3382734.3405706}}.

\bibitem{DBLP:journals/dm/HudakLSS14}
D{\'{a}}vid Hud{\'{a}}k, Borut Luzar, Roman Sot{\'{a}}k, and Riste Skrekovski.
\newblock Strong edge-coloring of planar graphs.
\newblock {\em Discret. Math.}, 324:41--49, 2014.
\newblock \href {https://doi.org/10.1016/j.disc.2014.02.002}
  {\path{doi:10.1016/j.disc.2014.02.002}}.

\bibitem{ImpagliazzoPZ01}
Russell Impagliazzo, Ramamohan Paturi, and Francis Zane.
\newblock Which problems have strongly exponential complexity?
\newblock {\em J. Comput. Syst. Sci.}, 63(4):512--530, 2001.
\newblock \href {https://doi.org/10.1006/jcss.2001.1774}
  {\path{doi:10.1006/jcss.2001.1774}}.

\bibitem{MR2002171}
Daniel Kr\'{a}l and Riste \v{S}krekovski.
\newblock A theorem about the channel assignment problem.
\newblock {\em SIAM J. Discrete Math.}, 16(3):426--437, 2003.
\newblock \href {https://doi.org/10.1137/S0895480101399449}
  {\path{doi:10.1137/S0895480101399449}}.

\bibitem{DBLP:journals/cj/LeeL14}
Chia{-}Lin Lee and Tzong{-}Jye Liu.
\newblock A self-stabilizing distance-2 edge coloring algorithm.
\newblock {\em Comput. J.}, 57(11):1639--1648, 2014.
\newblock \href {https://doi.org/10.1093/comjnl/bxt072}
  {\path{doi:10.1093/comjnl/bxt072}}.

\bibitem{DBLP:conf/stoc/Lin21}
Bingkai Lin.
\newblock Constant approximating k-clique is {W[1]}-hard.
\newblock In Samir Khuller and Virginia~Vassilevska Williams, editors, {\em
  {STOC} '21: 53rd Annual {ACM} {SIGACT} Symposium on Theory of Computing,
  Virtual Event, Italy, June 21-25, 2021}, pages 1749--1756. {ACM}, 2021.
\newblock \href {https://doi.org/10.1145/3406325.3451016}
  {\path{doi:10.1145/3406325.3451016}}.

\bibitem{DBLP:conf/icci/LloydR92}
Errol~L. Lloyd and Subramanian Ramanathan.
\newblock On the complexity of distance-2 coloring.
\newblock In Waldemar~W. Koczkodaj, Peter~E. Lauer, and Anestis~A. Toptsis,
  editors, {\em Computing and Information - ICCI'92, Fourth International
  Conference on Computing and Information, Toronto, Ontario, Canada, May 28-30,
  1992, Proceedings}, pages 71--74. {IEEE} Computer Society, 1992.

\bibitem{DBLP:journals/toc/Marx10}
D{\'{a}}niel Marx.
\newblock Can you beat treewidth?
\newblock {\em Theory Comput.}, 6(1):85--112, 2010.
\newblock \href {https://doi.org/10.4086/toc.2010.v006a005}
  {\path{doi:10.4086/toc.2010.v006a005}}.

\bibitem{DBLP:journals/mp/McCormick83}
S.~Thomas McCormick.
\newblock Optimal approximation of sparse hessians and its equivalence to a
  graph coloring problem.
\newblock {\em Math. Program.}, 26(2):153--171, 1983.
\newblock \href {https://doi.org/10.1007/BF02592052}
  {\path{doi:10.1007/BF02592052}}.

\bibitem{MoharT01}
Bojan Mohar and Carsten Thomassen.
\newblock {\em Graphs on Surfaces}.
\newblock Johns Hopkins series in the mathematical sciences. Johns Hopkins
  University Press, 2001.

\bibitem{MR1438613}
Michael Molloy and Bruce Reed.
\newblock A bound on the strong chromatic index of a graph.
\newblock {\em J. Combin. Theory Ser. B}, 69(2):103--109, 1997.
\newblock \href {https://doi.org/10.1006/jctb.1997.1724}
  {\path{doi:10.1006/jctb.1997.1724}}.

\bibitem{DBLP:journals/jct/MolloyS05}
Michael Molloy and Mohammad~R. Salavatipour.
\newblock A bound on the chromatic number of the square of a planar graph.
\newblock {\em J. Comb. Theory, Ser. {B}}, 94(2):189--213, 2005.
\newblock \href {https://doi.org/10.1016/j.jctb.2004.12.005}
  {\path{doi:10.1016/j.jctb.2004.12.005}}.

\bibitem{DBLP:journals/akcej/PaulPP16}
Satyabrata Paul, Madhumangal Pal, and Anita Pal.
\newblock A linear time algorithm to compute square of interval graphs and
  their colouring.
\newblock {\em {AKCE} Int. J. Graphs Comb.}, 13(1):54--64, 2016.
\newblock \href {https://doi.org/10.1016/j.akcej.2016.02.007}
  {\path{doi:10.1016/j.akcej.2016.02.007}}.

\bibitem{MR115937}
Ian~C. Ross and Frank Harary.
\newblock The square of a tree.
\newblock {\em Bell System Tech. J.}, 39(3):641--647, 1960.
\newblock \href {https://doi.org/10.1002/j.1538-7305.1960.tb03936.x}
  {\path{doi:10.1002/j.1538-7305.1960.tb03936.x}}.

\bibitem{SeymourT93}
Paul~D. Seymour and Robin Thomas.
\newblock Graph searching and a min-max theorem for tree-width.
\newblock {\em J. Comb. Theory, Ser. {B}}, 58(1):22--33, 1993.
\newblock \href {https://doi.org/10.1006/jctb.1993.1027}
  {\path{doi:10.1006/jctb.1993.1027}}.

\bibitem{Stockmeyer73}
Larry~J. Stockmeyer.
\newblock Planar 3-colorability is polynomial complete.
\newblock {\em {SIGACT} News}, 5(3):19--25, 1973.
\newblock \href {https://doi.org/10.1145/1008293.1008294}
  {\path{doi:10.1145/1008293.1008294}}.

\bibitem{DBLP:journals/jct/Thomassen18}
Carsten Thomassen.
\newblock The square of a planar cubic graph is 7-colorable.
\newblock {\em J. Comb. Theory, Ser. {B}}, 128:192--218, 2018.
\newblock \href {https://doi.org/10.1016/j.jctb.2017.08.010}
  {\path{doi:10.1016/j.jctb.2017.08.010}}.

\bibitem{MR1653829}
J.~van~den Heuvel, R.~A. Leese, and M.~A. Shepherd.
\newblock Graph labeling and radio channel assignment.
\newblock {\em J. Graph Theory}, 29(4):263--283, 1998.
\newblock \href
  {https://doi.org/10.1002/(SICI)1097-0118(199812)29:4<263::AID-JGT5>3.3.CO;2-M}
  {\path{doi:10.1002/(SICI)1097-0118(199812)29:4<263::AID-JGT5>3.3.CO;2-M}}.

\bibitem{MR1953224}
Jan van~den Heuvel and Sean McGuinness.
\newblock Coloring the square of a planar graph.
\newblock {\em J. Graph Theory}, 42(2):110--124, 2003.
\newblock \href {https://doi.org/10.1002/jgt.10077}
  {\path{doi:10.1002/jgt.10077}}.

\bibitem{MR2032294}
Wei-Fan Wang and Ko-Wei Lih.
\newblock Labeling planar graphs with conditions on girth and distance two.
\newblock {\em SIAM J. Discrete Math.}, 17(2):264--275, 2003.
\newblock \href {https://doi.org/10.1137/S0895480101390448}
  {\path{doi:10.1137/S0895480101390448}}.

\bibitem{Wegener77}
Gerd Wegener.
\newblock Graphs with given diameter and a coloring problem.
\newblock Technical report, University of Dortmund, 1977.
\newblock \href {https://doi.org/10.17877/DE290R-16496}
  {\path{doi:10.17877/DE290R-16496}}.

\bibitem{ZhouKN00}
Xiao Zhou, Yasuaki Kanari, and Takao Nishizeki.
\newblock Generalized vertex-colorings of partial k-trees.
\newblock {\em IEICE Transactions on Fundamentals of Electronics,
  Communications and Computer Sciences}, E83-A(4):671--678, 2000.

\end{thebibliography}

\end{document}